\documentclass[reqno,a4paper,11pt,draft]{article}
\pdfoutput=1
\usepackage{xcolor}

\newif\iffancy % Dummy macro by Hyungrok to control whether to load bbm and latexeu
\fancytrue

\usepackage[textwidth = 430 pt, textheight = 630 pt]{geometry}

\definecolor{MyDarkBlue}{rgb}{0.15,0.25,0.45}
\usepackage{amsmath,amssymb}
\usepackage{amsfonts}
\usepackage{mathrsfs}
\usepackage{stackengine} %% used in the definition of \chint

\renewcommand\varPhi{\boldsymbol\Phi} %% make varPhi more distinguishable...

\iffancy
\usepackage{bbm}
\else
\newcommand\mathbbm\mathbb
\fi
\usepackage{booktabs}

\usepackage{amsthm}
\usepackage[utf8x]{inputenc}

\usepackage[final,unicode]{hyperref}
\hypersetup{
    hypertexnames=false,
    colorlinks=true,
    citecolor=MyDarkBlue,
    linkcolor=MyDarkBlue,
    urlcolor=MyDarkBlue,
    pdfauthor={Hyungrok Kim and Christian Saemann},
    pdftitle={Adjusted Parallel Transport for Higher Gauge Theories},
    pdfsubject={hep-th math-ph},
    breaklinks=true
}

\PrerenderUnicode{∞} %% this enables the Unicode character ∞ to show up in PDF bookmarks

\usepackage{pgf,tikz}
\usetikzlibrary{cd,decorations.markings,calc}
\usepackage{mathtools}
\let\fn\footnote
\renewcommand{\footnote}[1]{\linespread{1.1}\fn{#1}\linespread{1.29}}

\makeatletter\renewcommand{\section}{\@startsection
    {section}{1}{\z@}{-3.5ex plus -1ex minus
        -.2ex}{2.3ex plus .2ex}{\bf }}
\makeatletter\renewcommand{\subsection}{\@startsection{subsection}{2}{\z@}{-3.25ex
        plus -1ex minus
        -.2ex}{1.5ex plus .2ex}{\bf }}
\makeatletter\renewcommand{\subsubsection}{\@startsection{subsubsection}{3}{-2.45ex}{-3.25ex
        plus -1ex minus -.2ex}{1.5ex plus .2ex}{\it }}
\renewcommand{\thesection}{\arabic{section}}
\renewcommand{\thesubsection}{\arabic{section}.\arabic{subsection}}
\renewcommand{\@seccntformat}[1]{\@nameuse{the#1}.~~}

\renewcommand{\theequation}{\thesection.\arabic{equation}}
\makeatletter \@addtoreset{equation}{section}
\def\Ddots{\mathinner{\mkern1mu\raise\p@
        \vbox{\kern7\p@\hbox{.}}\mkern2mu
        \raise4\p@\hbox{.}\mkern2mu\raise7\p@\hbox{.}\mkern1mu}}
\usepackage[toc,page]{appendix}

\newtheorem{thm}{Theorem}[section]
\renewcommand{\thethm}{\thesection.\arabic{thm}}

\newtheorem{theorem}[thm]{Theorem}

\newtheorem{corollary}[thm]{Corollary}
\newtheorem{remark}[thm]{Remark}

\renewcommand{\appendices}{
    \section*{Appendix}\label{appendices}\setcounter{subsection}{0}
    \addcontentsline{toc}{section}{Appendix}
    \setcounter{equation}{0}
    \makeatletter
    \renewcommand{\theequation}{\Alph{subsection}.\arabic{equation}}
    \renewcommand{\thesubsection}{\Alph{subsection}}
    \renewcommand{\thethm}{\Alph{subsection}.\arabic{thm}}
    \@addtoreset{equation}{subsection}
    \@addtoreset{thm}{subsection}
    \makeatother
}

\newcommand{\hooklongrightarrow}{\lhook\joinrel\longrightarrow}

\newcommand\cm{\mathrm{cm}}

\def\slasha#1{\setbox0=\hbox{$#1$}#1\hskip-\wd0\hbox to\wd0{\hss\sl/\/\hss}}

\def\periodb#1{\setbox0=\hbox{$#1$}#1\hskip-\wd0\hbox to\wd0{-}}

\iffancy
\newcommand{\unit}{\mathbbm{1}}   			% identity map/matrix
\else % AMS Mathbb doesn't have double-struck digits; temporary fix
\newcommand{\unit}{1\!\!1}

\fi
\newcommand{\id}{\mathrm{id}}   			% identity map/matrix
\newcommand{\CA}{\mathcal{A}}    			% cal-letters

\newcommand{\CC}{\mathcal{C}}
\newcommand{\CCC}{\mathscr{C}}

\newcommand{\CCG}{\mathscr{G}}

\newcommand{\CP}{\mathcal{P}}

\newcommand{\frg}{\mathfrak{g}}				% frak-letters
\newcommand{\frL}{\mathfrak{L}}				% frak-letters
\newcommand{\frh}{\mathfrak{h}}				% frak-letters
\newcommand{\frv}{\mathfrak{v}}
\newcommand{\frl}{\mathfrak{l}}

\newcommand{\FR}{\mathbbm{R}}     			% field of real numbers
\newcommand{\RZ}{\mathbbm{Z}}     			% ring of integers
\newcommand{\dd}{\mathrm{d}}     			% total differential
\newcommand{\dpar}{\partial}     			% partial differential
\newcommand{\de}{\mathrm{e}}     			% Euler's number
\newcommand{\di}{\mathrm{i}}     			% imaginary unit
\newcommand{\eps}{{\varepsilon}}			% antisymmetric tensors
\newcommand{\sB}{\mathsf{B}}

\newcommand{\sE}{\mathsf{E}}

\newcommand{\eand}{{\qquad\mbox{and}\qquad}}     		% and etc. in equations
\newcommand{\ewith}{{\qquad\mbox{with}\qquad}}

\newcommand{\der}[1]{\frac{\dpar}{\dpar #1}}   		% partielle ableitung, 1 argument
\newcommand{\astring}{\mathfrak{string}}

\newcommand{\sU}{\mathsf{U}}     			% groups

\newcommand{\sG}{\mathsf{G}}

\newcommand{\sL}{\mathsf{L}}

\newcommand{\sLie}{\mathsf{Lie}}
\newcommand{\sCE}{\mathsf{CE}}

\newcommand{\sH}{\mathsf{H}}

\newcommand{\sString}{\mathsf{String}}

\newcommand{\acton}{\mathbin\vartriangleright}     			% span
\newcommand{\tildeacton}{\mathbin{\tilde\vartriangleright}}     			% span
\renewcommand{\remark}[1]{}     				% remark
\def\tyng(#1){\hbox{\tiny$\yng(#1)$}}			% small Young diagram
\def\tyoung(#1){\hbox{\tiny$\young(#1)$}}			% small Young diagram
\if0 %% original definition of \chint

\fi

\newlength\tmplength
\def\showmybox{\tmplength=\wd0\relax\box0\kern-\tmplength}
\def\mybox#1{\setbox0=\hbox{\ensurestackMath{#1}}}
\def\SIn#1#2{\stackinset{c}{#1}{c}{#2}}
\def\chint{\mathchoice%
{\mybox{\SIn{1.8pt}{-0pt}{\mathsf C}{\displaystyle\phantom{\int}}}\showmybox}% 0pt
{\mybox{\SIn{1.5pt}{-0pt}{\scriptstyle\mathsf C}{\textstyle\phantom{\int}}}\showmybox}% -1.5pt
{\mybox{\SIn{1.2pt}{-0pt}{\scriptscriptstyle\mathsf C}{\scriptstyle\phantom{\int}}}\showmybox}%
{\mybox{\SIn{1.1pt}{-0pt}{\scalebox{0.7}{$\scriptscriptstyle\mathsf C$}}{%
\scriptscriptstyle\phantom{\int}}}\showmybox}%
\int}

\def\ichint{\mathchoice%
{\mybox{\SIn{9.3pt}{-0pt}{\mathsf C}{\displaystyle\phantom{\int}}}\showmybox}% 0pt
{\mybox{\SIn{7pt}{-0pt}{\scriptstyle\mathsf C}{\textstyle\phantom{\int}}}\showmybox}% -1.5pt
{\mybox{\SIn{5.5pt}{-0pt}{\scriptscriptstyle\mathsf C}{\scriptstyle\phantom{\int}}}\showmybox}%
{\mybox{\SIn{4.4pt}{-0pt}{\scalebox{0.7}{$\scriptscriptstyle\mathsf C$}}{%
\scriptscriptstyle\phantom{\int}}}\showmybox}%
\iint}

\def\chichint{\mathchoice%
{\mybox{\SIn{1.8pt}{-0pt}{\mathsf C}{\displaystyle\phantom{\int}}}\showmybox}% 0pt
{\mybox{\SIn{1.5pt}{-0pt}{\scriptstyle\mathsf C}{\textstyle\phantom{\int}}}\showmybox}% -1.5pt
{\mybox{\SIn{1.2pt}{-0pt}{\scriptscriptstyle\mathsf C}{\scriptstyle\phantom{\int}}}\showmybox}%
{\mybox{\SIn{1.1pt}{-0pt}{\scalebox{0.7}{$\scriptscriptstyle\mathsf C$}}{%
\scriptscriptstyle\phantom{\int}}}\showmybox}%
\mathchoice%
{\mybox{\SIn{9.3pt}{-0pt}{\mathsf C}{\displaystyle\phantom{\int}}}\showmybox}% 0pt
{\mybox{\SIn{7pt}{-0pt}{\scriptstyle\mathsf C}{\textstyle\phantom{\int}}}\showmybox}% -1.5pt
{\mybox{\SIn{5.5pt}{-0pt}{\scriptscriptstyle\mathsf C}{\scriptstyle\phantom{\int}}}\showmybox}%
{\mybox{\SIn{4.4pt}{-0pt}{\scalebox{0.7}{$\scriptscriptstyle\mathsf C$}}{%
\scriptscriptstyle\phantom{\int}}}\showmybox}%
\iint}

\def\ichichint{\mathchoice%
{\mybox{\SIn{1.8pt}{-0pt}{\mathsf C}{\displaystyle\phantom{\iiint}}}\showmybox}% 0pt
{\mybox{\SIn{1.5pt}{-0pt}{\scriptstyle\mathsf C}{\textstyle\phantom{\iiint}}}\showmybox}% -1.5pt
{\mybox{\SIn{1.2pt}{-0pt}{\scriptscriptstyle\mathsf C}{\scriptstyle\phantom{\iiint}}}\showmybox}%
{\mybox{\SIn{1.1pt}{-0pt}{\scalebox{0.7}{$\scriptscriptstyle\mathsf C$}}{%
\scriptscriptstyle\phantom{\iiint}}}\showmybox}%
\mathchoice%
{\mybox{\SIn{9.3pt}{-0pt}{\mathsf C}{\displaystyle\phantom{\iiint}}}\showmybox}% 0pt
{\mybox{\SIn{7pt}{-0pt}{\scriptstyle\mathsf C}{\textstyle\phantom{\iiint}}}\showmybox}% -1.5pt
{\mybox{\SIn{5.5pt}{-0pt}{\scriptscriptstyle\mathsf C}{\scriptstyle\phantom{\iiint}}}\showmybox}%
{\mybox{\SIn{4.4pt}{-0pt}{\scalebox{0.7}{$\scriptscriptstyle\mathsf C$}}{%
\scriptscriptstyle\phantom{\iiint}}}\showmybox}%
\iiint}

\newcommand{\sft}{{\sf t}}

\newcommand{\sfid}{\mathsf{id}}

\definecolor{outrageousorange}{rgb}{1.0, 0.43, 0.29}

\newcommand{\astringsk}{\mathfrak{string}_{\rm sk}}
\newcommand{\astringl}{\mathfrak{string}_{\rm lp}}

\newcommand{\ainn}{\mathfrak{inn}}
\newcommand{\sInn}{\mathsf{Inn}}
\newcommand{\sW}{\mathsf{W}}

\DeclareMathOperator\Pexp{P\,exp}
\DeclareMathOperator\dom{dom}
\DeclareMathOperator\codom{codom}

\iffancy
\usepackage{stmaryrd}
\else
\newcommand\llbracket{[\![}
\newcommand\rrbracket{]\!]}
\fi

\begin{document}
    \begin{titlepage}
        \begin{flushright}
            EMPG--19--24
        \end{flushright}
        \vskip2.0cm
        \begin{center}
            {\LARGE \bf
                Adjusted Parallel Transport\\[0.2cm]
                for Higher Gauge Theories
            }
            \vskip1.5cm
            {\Large Hyungrok Kim and Christian Saemann}
            \setcounter{footnote}{0}
            \renewcommand{\thefootnote}{\arabic{thefootnote}}
            \vskip1cm
            {\em Maxwell Institute for Mathematical Sciences\\
                Department of Mathematics, Heriot--Watt University\\
                Colin Maclaurin Building, Riccarton, Edinburgh EH14 4AS, U.K.}\\[0.5cm]
            {Email: {\ttfamily hk55@hw.ac.uk~,~c.saemann@hw.ac.uk}}
        \end{center}
        \vskip1.0cm
        \begin{center}
            {\bf Abstract}
        \end{center}
        \begin{quote}
            Many physical theories, including notably string theory, require non-abelian higher gauge fields defining higher holonomy. Previous approaches to such higher connections on categorified principal bundles require these to be fake flat. This condition, however, renders them locally gauge equivalent to connections on abelian gerbes. For particular higher gauge groups, for example 2-group models of the string group, this limitation can be overcome by generalizing the notion of higher connection. Starting from this observation, we define a corresponding generalized higher holonomy functor which is free from the fake flatness condition, leading to a truly non-abelian parallel transport.
        \end{quote}

    \end{titlepage}
    
    \tableofcontents

    \section{Introduction and results}
        
    \subsection{Motivation}
    
    Non-abelian higher gauge fields arise in a number of physical contexts, ranging from six-dimensional conformal field theory over supergravity theories to string/M-theory. Such gauge fields are meant to describe higher holonomies, arising from a parallel transport along higher-dimensional spaces, e.g.~surfaces. 
    
    In particular, the classical string couples to the Kalb--Ramond 2-form field $B$, which is part of the connection of an abelian gerbe. This is the higher analogue of a particle coupling to a Maxwell gauge potential $A$, which is part of the connection on an abelian principal bundle. If we now want to generalize connections on abelian gerbes to potentially self-interacting ones, mimicking the transition from Maxwell fields to Yang--Mills fields, we face a number of problems. Using the appropriate language of 2-categories and functorial definitions of higher principal bundles, of their connections and of the induced parallel transport, most of these\footnote{e.g.~the Eckmann--Hilton type argument forbidding a na\"ive non-abelian higher parallel transport} are readily overcome.
    
    We arrive at a theory of non-abelian gerbes or higher principal bundles with connections~\cite{Breen:math0106083,Aschieri:2003mw,Baez:2004in,Baez:0511710} together with an induced parallel transport~\cite{Baez:2002jn,Girelli:2003ev,Baez:2004in,Schreiber:0705.0452,Schreiber:0802.0663,Schreiber:2008aa}; see also~\cite{Soncini:2014zra} as well as~\cite{Baez:2010ya} for an introduction. Topologically, these non-abelian higher principal bundles are simultaneous generalizations of (non-abelian) principal fiber bundles and abelian gerbes. The connections they carry, however, merely generalize connections of abelian gerbes. Consistency of the underlying differential cocycles requires that a particular curvature component, known as {\em fake curvature}, vanishes. This fake flatness condition also arises from a higher Stokes' theorem, guaranteeing invariance of the induced higher parallel transport under reparametrizations.
    
    Thus, the fake flatness condition forbids a straightforward interpretation of ordinary principal bundles with connections as non-abelian higher principal bundles with connection. This is surprising because categorification usually implies generalization: a set is trivially a category, a category is trivially a 2-category, a group is trivially a 2-group, and, indeed, a principal bundle is trivially a principal 2-bundle; but not every principal bundle with connection is a principal 2-bundle with connection in the sense of~\cite{Breen:math0106083,Aschieri:2003mw,Baez:2004in,Baez:0511710}. Even worse, connections on non-abelian higher principal bundles are locally gauge equivalent to connections on abelian ones; see~\cite{Saemann:2019dsl} and also~\cite{Gastel:2018joi}. Locally, thus, the extension to non-abelian higher principal bundles is futile, and we can merely hope to answer some topological questions with these, using higher versions of Chern--Simons theories. Other highly interesting theories, such as six-dimensional superconformal field theories involving the tensor multiplet, require local gauge field interactions, also over topologically trivial spaces, and therefore the conventional non-abelian higher principal bundles are inapt for their description.
    
    For certain higher gauge groups, there is a further generalization of the notion of connection that lifts this limitation. A higher gauge algebra\footnote{Note that we always follow the physicists' nomenclature and identify the terms {\em gauge group} and {\em gauge (Lie) algebra} with the structure group and structure Lie algebra of the principal bundle underlying the gauge theory. The gauge group is thus different from the resulting {\em group of gauge transformations}.} $\frL$ gives rise to a differential graded algebra $\sW(\frL)$, called its {\em Weil algebra}. The kinematical data of a higher gauge theory over some local patch $U$ of spacetime is fully encoded in a morphism of differential graded algebras from $\sW(\frL)$ to the de~Rham complex $\Omega^\bullet(U)$. However, the na\"ive generalization of the notion of Weil algebra of a Lie algebra to the Weil algebra of a higher Lie algebra is problematic: the induced definition of invariant polynomials is not compatible with quasi-isomorphisms, which are the appropriate notion of isomorphisms for higher Lie algebras. For particular higher Lie algebras $\frL$, this incompatibility can be overcome by particular deformations of the Weil algebra $\sW(\frL)$~\cite{Sati:2008eg,Sati:2009ic}.
    
    At the field theory level, the BRST complex describing infinitesimal gauge transformations and their actions on the fields arising from morphisms $\sW(\frL)\rightarrow \Omega^\bullet(U)$ is not closed. It closes only up to equations of motion corresponding to the fake curvature condition. The aforementioned deformations of the Weil algebra also cure this problem. Such deformed Weil algebras that induce a closed BRST complex were called {\em adjusted Weil algebras} in~\cite{Saemann:2019dsl}.
    
    If the higher gauge group is the \emph{string 2-group}\footnote{more precisely: a 2-group model of the string group}, a higher relative of the spin group, the adjustment leads to \emph{differential string structures}. These are expected to arise in the context of string theory and M-theory; see~\cite{Saemann:2017rjm,Saemann:2017zpd,Saemann:2019dsl}.
    
    Because the adjustment of the Weil algebra lifts the fake flatness condition from the higher differential cocycles, it is natural to wonder if it does the same for the higher parallel transport. This is the main question in this paper, which we answer in the affirmative.

    \subsection{Results}
    
    To simplify the presentation we restrict some parts of our discussion to the example of the loop model of the string Lie 2-group; there is, however, no reason to believe that it does not apply to arbitrary higher Lie groups admitting an adjustment.
    
    We start from the observation that the dual of the Weil algebra of a higher Lie $n$-algebra $\frL$ is isomorphic to the Lie $(n+1)$-algebra of inner derivations $\ainn(\frL)$ of $\frL$; details are found in appendix~\ref{app:inner_weil}. For a Lie 2-algebra $\frL$, this leads to the Lie 3-algebra $\ainn(\frL)$, which can be described in two ways: first, as a 3-term $L_\infty$-algebra and second, as a 2-crossed module of Lie algebras, as done in~\cite{Roberts:0708.1741}. The former is directly obtained from the Weil algebra, while the latter contains some additional information and is readily integrated to a 2-crossed module of Lie groups. In appendix~\ref{app:Lie3_2_crossed}, we explain in detail the correspondence between Lie 3-algebras and 2-crossed modules of Lie algebras.
    
    The adjustment of the Weil algebra amounts thus to an adjustment of the algebra of inner derivations. One result that certainly deserves further study is that the data required to lift the Lie 3-algebra into a 2-crossed module of Lie algebras is precisely the data encoding the adjustment of the Weil algebra; see sections~\ref{ssec:loop_model_of_Lie_algebra} and~\ref{ssec:adjusted_inner_derivation}. There, we also compute the corresponding integrated adjusted inner automorphism Lie 2-groups in the form of a 2-crossed module of Lie groups and compare them to the unadjusted forms.
    
    For simplicity, we focus our discussion on {\em local} parallel transport over a contractible patch $U$ of the spacetime manifold; gluing the local picture to a global one is mostly a technicality. Local parallel transport in ordinary gauge theory with gauge group\footnote{i.e.~the structure group of an underlying principal bundle} $\sG$ is essentially a functor $\Phi$ from the path groupoid $\CP U$, which has points in $U$ as its objects and paths between these as morphisms, to the one-object groupoid $\sB\sG$ which has $\sG$ as its group of morphisms:
    \begin{equation}
        \Phi\colon \CP U\longrightarrow \sB\sG~,\hspace{2cm}
        \begin{tikzcd}
            \mbox{paths} \arrow[d,shift left] \arrow[d,shift right]\arrow[r]{}{\Phi_1} & \sG \arrow[d,shift left] \arrow[d,shift right]\\
            U \arrow[r]{}{\Phi_0} & * 
        \end{tikzcd}
    \end{equation}
    see section~\ref{ssec:parallel} for technical details. Any path is thus associated to a group element such that constant paths are mapped to $\unit_\sG$ and composition of paths leads to multiplication of the corresponding group elements. These are the axioms of a parallel transport. A connection 1-form is readily extracted from considering infinitesimal paths and conversely, a connection 1-form maps a path to a group element by the usual path-ordered exponential.
    
    An alternative yet equivalent picture is obtained from the short exact sequence of groupoids
    \begin{equation}
        *\longrightarrow \begin{array}{c} \sG\\[-0.75ex] \downdownarrows\\[-0.3ex]\sG\end{array} \hooklongrightarrow \sInn(\sG)\longrightarrow \begin{array}{c} \sG\\[-0.75ex] \downdownarrows\\[-0.75ex]*\end{array}\longrightarrow *~,
    \end{equation}
    where $\sInn(\sG)$ is the Lie 2-group of inner automorphisms. Instead of a functor from $\CP U$ to $\sB\sG$, we can also consider a 2-functor $\varPhi$ from the 2-groupoid $\CP_{(2)} U$ of points, paths and paths or homotopies between paths to $\sB\sInn(\sG)$. We find that the globular identities\footnote{i.e.~the relations between domain and codomain for morphisms and higher morphisms: the domain of the domain is the domain of the codomain, and the codomain of the domain is the codomain of the codomain} for $\varPhi$ reduce the defining data to the same as for $\Phi$, which simply corresponds to Stokes' theorem at the level of connection 1-forms and curvature 2-forms.
    
    Gauge transformations are encoded in natural transformations between the corresponding functors. In the case of 2-functors $\varPhi:\CP_{(2)} U\rightarrow \sB\sInn(\sG)$, we must restrict the 2-natural transformations to obtain the correct set of gauge transformations, as we explain in detail in section~\ref{ssec:derived_parallel_ordinary}.
    
    In the context of higher gauge theory with Lie 2-group $\CCG$, similarly, a strict 2-functor $\Phi\colon\CP_{(2)}U\rightarrow \sB\CCG$ induces a strict 3-functor $\varPhi\colon\CP_{(3)} U \rightarrow \sB\sInn(\CCG)$. In both cases, the fake curvature condition appears as a necessary condition for the existence of these functors.
    
    If we replace, however, the 3-group of inner automorphisms $\sInn(\sString(\sG))$ with the 3-group of {\em adjusted} inner automorphisms $\sInn_{\rm adj}(\CCG)$ and consider a strict 3-functor
    \begin{equation}
        \varPhi^{\rm adj}\colon\CP_{(3)}U\longrightarrow \sB\sInn_{\rm adj}(\sString(\sG))~,
    \end{equation}
    we obtain a new higher-dimensional parallel transport. The higher-dimensional Stokes' theorem is automatically satisfied and merely enforces the definition of higher curvatures together with the corresponding higher Bianchi identities. This parallel transport is truly non-abelian and underlies the self-interacting field theories constructed from the kinematical data arising from the adjusted Weil algebra. Contrary to the unadjusted parallel transport, this 3-functor only simplifies to a 2-functor if the underlying connection is fake flat.
    
    Altogether, we conclude that while for ordinary (higher) gauge theories there exist two fully equivalent ways of defining parallel transport, this is no longer the case if we aim for an adjusted higher parallel transport, where only one of these definitions is possible. In particular, a general higher parallel transport along $d$-dimensional volumes with underlying gauge $d$-group $\CCG$ which admits an adjusted higher $(d+1)$-group of inner automorphisms $\sInn_{\rm adj}(\CCG)$ is based on a strict $(d+1)$-functor
    \begin{equation}
        \varPhi^{\rm adj}\colon\CP_{(d+1)}U\longrightarrow \sB\sInn_{\rm adj}(\CCG)~.
    \end{equation}
    Gauge transformations arise from appropriately restricted $d$-natural transformations. For $d=1$, $\sInn(\CCG)$ is always adjusted, since there are no higher curvatures, and the 2-functor simplifies to a functor $\Phi:\CP U\rightarrow \sB\CCG$, reproducing the usual higher transport. If an adjustment is not possible for $d>1$, then only an unadjusted parallel transport exists, which is locally gauge equivalent to an abelian one. 
    
    Our discussion extends in principle straightforwardly to higher dimensions, except that one should use simplicial models of the required higher path groupoids and higher Lie groups, as~in~e.g.~\cite{Jurco:2016qwv}; the technicalities of higher coherence laws will otherwise overwhelm.
        
    \subsection{Why categories, groupoids and all that?}
    
    Because we hope to reach also a less ``categorically minded'' audience with this paper, we outline the interrelations between all the categorical and higher structures and the reasons for them arising in this paper.
    
    An experimenter observes the Aharonov--Bohm effect and concludes that nature associates to each path a phase, i.e.~an element of $\sU(1)$. The phases add when paths are concatenated; the phases invert when paths are inverted. One would call the set of paths a group, except that paths only compose when the endpoints match. We instead call it a \emph{groupoid}, or more precisely the {\em path groupoid} of the space, and regard the points in space as objects and the paths as maps or {\em morphisms} between their endpoints, called {\em domain} and {\em codomain}:
    \begin{equation}
        \begin{gathered}
            \begin{tikzcd}
                \codom(\gamma)=y&[-1.2cm]\bullet &[1cm] \arrow[l,bend right,swap]{}{\textrm{path}~\gamma}\bullet&[-1.2cm]x=\dom(\gamma)
            \end{tikzcd}\\
            \begin{tikzcd}
                x &[-1.2cm]\bullet \arrow[loop,swap]{}{\id_x}
            \end{tikzcd}~~~
            \begin{tikzcd}
                y&[-1.2cm]\bullet \arrow[r,bend left]{}{\gamma^{-1}} &[1cm] \bullet&[-1.2cm]x
            \end{tikzcd}~~~
            \begin{tikzcd}
                \bullet & \ar[l,bend right,swap]{}{\gamma_2}\bullet & \bullet \ar[l,bend right,swap]{}{\gamma_1}\ar[ll,bend left]{}{\gamma_2\circ \gamma_1}
            \end{tikzcd}
        \end{gathered}
    \end{equation}
    The group of phases, $\sU(1)$, can also be regarded as a groupoid with morphisms $\sU(1)$ taking a single object $*$ to itself:
    \begin{equation}
        \begin{tikzpicture}[baseline=(current bounding box.center)]
            \node (a) {$*$};
            \path[scale=2] % <--- scaled loops size
            (a) edge [loop right, "$\de^{\frac{1}{2}\pi\di}$"]
            (a) edge [loop below, "$\de^{\pi\di}$"] (a)
            (a) edge [loop above, "$\unit$"] (a);
            \draw [densely dotted]  (240:1.5em) arc (240:120:1.5em);
            \draw [densely dotted]  (60:1.5em) arc (60:30:1.5em);
            \draw [densely dotted]  (330:1.5em) arc (330:300:1.5em);
        \end{tikzpicture}
    \end{equation}
    We call this groupoid $\sB\sU(1)$. The Berry phase, which mathematicians call \emph{holonomy}, maps the objects and morphism in the path groupoid to objects and morphisms in the groupoid $\sB\sU(1)$ of phases. This generalization of a group homomorphism is called a {\em functor}.\footnote{The terminology is borrowed from philosophy: more general groupoids are called {\em categories}, which Saunders Mac~Lane took from Kant; he also took the term {\em functor} from Carnap.}
    
    Over time, physicists have discovered two variations on the theme. One, discovered by Yang and Mills, replaces the abelian group of phases $\sU(1)$ with non-abelian ones, as necessary for describing strong and weak nuclear forces. The other variation generalizes paths to surfaces and higher-dimensional spaces, as necessary for field theory on higher-dimensional spacetimes. String theory seems to require both variations at the same time in e.g.~stacks of NS5- or M5-branes, which have strings in them with self-interacting higher non-abelian gauge fields. Defining the right generalization of the underlying phases is important to fundamentally understand this physics.
    
    Both variations are captured by essentially obvious generalizations of the holonomy functor, exemplifying the utility of functorial descriptions of mathematical objects. Replacing the groupoid $\sB\sU(1)$ by the groupoid $\sB\sG$ for a non-abelian gauge group $\sG$ is straightforward. The generalization to a higher-dimensional parallel transport requires the development of higher-dimensional groupoids containing points, paths between points, paths between paths, etc., but also this construction is not hard, using geometric intuition. Here, a new feature is that higher morphisms compose in multiple ways: e.g.~in the case of the 2-groupoid of points, paths and paths-between-paths, paths-between-paths compose both vertically and horizontally:
    
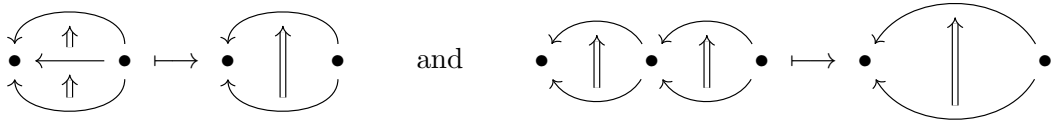
\begin{equation*}
        \begin{tikzcd}
            \bullet & \lar[""{name=A},""'{name=B}] \lar[bend left=90, ""'{name=C}] \lar[bend right=90, ""{name=D}] \bullet \arrow[Rightarrow, from=B, to=D] \arrow[Rightarrow, from=C, to=A] 
        \end{tikzcd}
        \longmapsto 
        \begin{tikzcd}
            \bullet & \lar[bend left=90, ""'{name=C}] \lar[bend right=90, ""{name=D}] \bullet \arrow[Rightarrow, from=C, to=D] 
        \end{tikzcd}
        \eand
        \begin{tikzcd}
            \bullet & \lar[bend left=60, ""'{name=A}] \lar[bend right=60, ""{name=B}] \bullet \arrow[Rightarrow, from=A, to=B] & 
            \lar[bend left=60, ""'{name=C}] \lar[bend right=60, ""{name=D}] \bullet \arrow[Rightarrow, from=C, to=D]
        \end{tikzcd}
        \longmapsto
        \begin{tikzcd}
            \bullet && 
            \ar[ll, bend left=60, ""'{name=C}] \ar[ll, bend right=60, ""{name=D}] \bullet \arrow[Rightarrow, from=C, to=D]
        \end{tikzcd}
    \end{equation*}
    This gives rise to the term {\em higher-dimensional algebra} for higher categories and higher groupoids. Higher-dimensional groupoids with single objects then describe higher analogues of groups, just as $\sB\sG$ describes the group $\sG$. This process of adding morphisms between morphisms is known as {\em categorification}.
    
    The heavy use of groupoids and their higher-dimensional generalizations is thus due to the ease with which they allow us to reproduce and subsequently generalize relevant mathematical definitions, guaranteeing consistency from the outset.
    
    All terms we use are generalizations or categorifications of the mathematical terms underlying ordinary gauge theories. It is more convenient to describe them by equivalent mathematical objects, and one can easily get lost in the nomenclature.
    
    As described above, the gauge group is categorified to a higher Lie group. We describe Lie 2-groups and Lie 3-groups with the more economical language of {\em crossed modules of Lie groups} and {\em 2-crossed modules of Lie groups}; see appendix~\ref{app:hypercrossed_modules}.
    Just as Lie groups differentiate to Lie algebras, higher Lie groups differentiate to higher Lie algebras. For Lie $n$-algebras, we use three models: we start from {\em $n$-term $L_\infty$-algebras}, see appendix~\ref{app:L_infinity_defs}, which we also describe dually via their function algebras, known as {\em Chevalley--Eilenberg algebras}, and the function algebra of their inner derivations, known as {\em Weil algebras}; see section~\ref{ssec:limitations}. The third description is in terms of \((n-1)\)-crossed modules of Lie algebras; see again appendix~\ref{app:hypercrossed_modules} for the definitions and appendix~\ref{app:Lie3_2_crossed} for the relation to $n$-term $L_\infty$-algebras.
    
    The higher analogue of a principal bundle is a {\em principal $n$-bundle} or an \emph{\((n-1)\)-gerbe}.\footnote{Standard nomenclature often assumes gerbes to be abelian while principal $n$-bundles are unrestricted.} Locally, the description of connections is easily described as morphisms from the Weil algebra of the gauge $L_\infty$-algebra to the de~Rham complex of the local patch; see section~\ref{ssec:limitations}. For a gauge Lie $n$-algebra, one obtains 1-, 2-, \dots, $n$-forms valued in particular parts of the Lie $n$-algebra and corresponding 2-, 3-, \dots, $(n+1)$-form components of the total curvature. All of the latter, except for the form with top degree, are known as {\em fake curvatures}.
    
    For higher-dimensional parallel transport we need {\em higher groupoids}, which, as clear from the higher path groupoid, are essentially collections of objects, morphisms between objects, and higher morphisms between morphisms, such that all morphisms are invertible. We mostly work with strict higher groupoids, i.e.~those for which composition of morphisms is strictly associative and unital. Strict higher \((n+1)\)-groupoids are readily defined by replacing the group of morphisms of a (1-)category with the \(n\)-groupoid of morphisms.
    
    As mentioned before, a \emph{higher group} is defined by a higher groupoid with a single object.\footnote{Pedantically, a \emph{higher group} is obtained by truncating the single object and instead regarding groupoid 1-morphisms as the objects of the higher group and groupoid \((k+1)\)-morphisms as \(k\)-morphisms of the higher group. This generalizes the relation between the \(\sG\) and \(\sB\sG\).}
    Our 3-groups are {\em Gray groups}, which means that the two different ways of evaluating the diagram
    \begin{equation}
        \begin{tikzcd}
            \bullet & \lar[""{name=Aa},""'{name=Ba}] \lar[bend left=90, ""'{name=Ca}] \lar[bend right=90, ""{name=Da}] \bullet \arrow[Rightarrow, from=Ba, to=Da] \arrow[Rightarrow, from=Ca, to=Aa]
            & \lar[""{name=E},""'{name=F}] \lar[bend left=90, ""'{name=G}] \lar[bend right=90, ""{name=H}] \bullet \arrow[Rightarrow, from=F, to=H] \arrow[Rightarrow, from=G, to=E] 
        \end{tikzcd}
    \end{equation}
    are the same, up to isomorphism. For further details, see the literature cited in the respective sections.

    \section{Adjusted Weil algebras}

    \subsection{Local kinematical data of gauge theories from differential graded algebras}
    
    Henri Cartan~\cite{Cartan:1949aaa,Cartan:1949aab} discovered a particularly elegant and useful description of local connection forms on principal fiber bundles. Let $\frg$ be a Lie algebra\footnote{either a finite-dimensional Lie algebra, or an infinite-dimensional Lie algebra with a suitable notion of dual space} with basis $e_\alpha$ and structure constants $f^\gamma_{\alpha\beta}$, such that
    \begin{equation}
        [e_\alpha,e_\beta] \eqqcolon f^\gamma_{\alpha\beta}e_\gamma~.
    \end{equation}
    Dually, $\frg$ can be regarded as the \emph{(graded-commutative) differential graded algebra}
    \begin{equation}
        \sCE(\frg)\coloneqq \big(\bigodot\nolimits^\bullet \frg[1]^*,Q_\sCE\big)=\big(\CC^\infty_{\rm pol}(\frg[1]),Q_\sCE\big)~,
    \end{equation}
    which consists of polynomials in the coordinate functions $t^\alpha\in \frg[1]^*$ of degree one and whose differential $Q_\sCE$ is the homological vector field 
    \begin{equation}
        Q_\sCE=-\tfrac12 f^\gamma_{\alpha\beta}t^\alpha t^\beta \der{t^\gamma}~,~~~|Q|=1~,~~~Q^2=0~.
    \end{equation}
    We call $\sCE(\frg)$ the \emph{Chevalley--Eilenberg algebra} of $\frg$.
    
    Similarly, the Chevalley--Eilenberg algebra of the grade-shifted tangent bundle $T[1]U$ of a local patch $U$ of some manifold $M$ can be identified with the de~Rham complex of $U$,
    \begin{equation}
        \sCE(T[1]U)=\big(\CC^\infty(T[1]U),\dd\big)=\big(\Omega^\bullet(U),\dd\big)~.
    \end{equation}
    Morphisms of differential graded algebras
    \begin{equation}\label{eq:dga_morph_CE}
        \CA:\sCE(\frg)\to \sCE(T[1] U)
    \end{equation}
    preserve the graded algebra structure and are therefore fixed by the image of $t^\alpha$,
    \begin{equation}
        \CA(t^\alpha) \eqqcolon A^\alpha\in \Omega^1(U)~,
    \end{equation}
    a Lie algebra-valued differential form or local connection 1-form $A\coloneqq A^\alpha \tau_\alpha$ on $U$. Compatibility with the differentials on $\sCE(\frg)$ and $\sCE(T[1]U)$ enforces flatness of this connection,
    \begin{equation}
        \begin{aligned}
            (\dd\circ \CA)(t^\alpha)&=(\CA\circ Q_\sCE)(t^\alpha)\\
            \dd A^\alpha&=\CA(-\tfrac12 f^\alpha_{\beta\gamma}t^\beta t^\gamma)=-\tfrac12 f^\alpha_{\beta\gamma}A^\beta \wedge A^\gamma\\
            &\implies F\coloneqq \dd A+\tfrac12 [A,A]=0~.
        \end{aligned}    
    \end{equation}
    Gauge transformations are encoded in partially flat homotopies between two morphisms $\CA$ and $\tilde \CA$ of type~\eqref{eq:dga_morph_CE}. 
    
    To describe non-flat connections, we enlarge the Chevalley--Eilenberg algebra $\sCE(\frg)$ to the {\em Weil algebra}
    \begin{equation}
        \sW(\frg)\coloneqq \left(\bigodot\nolimits^\bullet(\frg[1]^*\oplus \frg[2]^*),Q_\sW\right)=\big(\CC^\infty_{\rm pol}(\frg[1]^*\oplus \frg[2]^*),Q_\sW\big)~,
    \end{equation}
    which consists of polynomials in the coordinate functions $t^\alpha\in\frg[1]^*$ and $\hat t^\alpha=\sigma(t^\alpha)\in \frg[2]^*$, where $\sigma:\frg[1]^*\to \frg[2]^*$ is the shift isomorphism. We extend $\sigma$ trivially to $\frg[1]^*\oplus \frg[2]^*$ by $\sigma(\frg[2]^*)\coloneqq 0$ and as a derivation to $\bigodot\nolimits^\bullet(\frg[1]^*\oplus \frg[2]^*)$. We also extend $Q_\sCE$ to $\bigodot\nolimits^\bullet(\frg[1]^*\oplus \frg[2]^*)$ by demanding that
    \begin{equation}
        Q_\sCE \sigma\coloneqq -\sigma Q_\sCE~.
    \end{equation}
    The homological vector field $Q_\sW$ on $\frg[1]\oplus \frg[2]$ is then defined as
    \begin{equation}
        Q_\sW=Q_\sCE+\sigma~.
    \end{equation}
    Explicitly, we have
    \begin{equation}\label{eq:Weil_of_ordinary_Lie}
        Q_\sW\colon~t^\alpha \mapsto -\tfrac12 f^\alpha_{\beta\gamma} t^\beta t^\gamma + \hat t^\alpha \eand
        \hat t^\alpha\mapsto -f^\alpha_{\beta\gamma} t^\beta \hat t^\gamma~,
    \end{equation}
    where $f^\alpha_{\beta\gamma}$ are again the structure constants of $\frg$.
    
    Without going into further details, we note that the Chevalley--Eilenberg algebra of the tangent Lie algebroid $T[1]U$ can be seen as the Weil algebra of the manifold $U$ regarded as the trivial Lie algebroid over itself, $\sCE(T[1]U)=\sW(U)$.
    
    Non-flat connections are described as morphisms of differential graded algebras
    \begin{equation}\label{eq:dga_morph_Weil}
        \CA\colon \sW(\frg)\to \sW(U)~,
    \end{equation}
    which are fixed by their action on the generators $t^\alpha$ and $\hat t^\alpha$. We define
    \begin{equation}
        \begin{aligned}
            A&\coloneqq A^\alpha \tau_\alpha~,~~~&A^\alpha&\coloneqq \CA(t^\alpha)~,\\
            F&\coloneqq F^\alpha \tau_\alpha~,~~~&F^\alpha&\coloneqq \CA(\hat t^\alpha)~.
        \end{aligned}
    \end{equation}
    Compatibility with the differentials and the graded algebra structure implies that
    \begin{equation}
        F=\dd A+\tfrac12[A,A]\eand \dd F+[A,F]=0~.
    \end{equation}
    We thus recover the definition of the curvature and the Bianchi identity. As stated above, gauge transformations are obtained by partially flat homotopies. Recall that a homotopy between morphisms $\CA,\tilde \CA\colon \sW(\frg)\rightarrow \sW(U)$ is given by a morphism 
    \begin{equation}
        \hat \CA\colon\sW(\frg)\rightarrow \sW(U\times I)~~~\mbox{with}~~~\hat \CA(x,0)=\CA(x)~,~~~\hat \CA(x,1)=\tilde \CA(x)~,
    \end{equation}
    where $I=[0,1]$ and $x\in U$. Let $t$ be the coordinate on $I$. The potential 1-form and the curvature 2-form now naturally decompose into two parts:
    \begin{equation}
        \hat A=\hat A_x+\hat \alpha_t\dd t~,~~~\hat F=\hat F_x+\hat \varphi_t\wedge \dd t~,~~~\hat A_x\left(\der t\right)=\hat F_x\left(\der t\right)=0~.
    \end{equation}
    Partial flatness means $\varphi_t=0$, and we can compute
    \begin{equation}
        \delta A\coloneqq\der{t}\hat A(t)\Big|_{t=0}=\dd_x \hat \alpha_t+[\hat A_x,\hat \alpha_t]\Big|_{t=0}=:\dd_x \alpha+[A,\alpha]~,
    \end{equation}
    and we recover the usual form of infinitesimal gauge transformations.

    \subsection{Limitations of conventional higher gauge theories}\label{ssec:limitations}
    
    A particularly nice feature of Cartan's description of gauge potentials in terms of morphisms of differential graded algebras is its generality: one can easily replace both the domain and the codomain with more general differential graded algebras. In this paper, we are interested in more general domains~\cite{Kotov:2007nr,Sati:2008eg}; see e.g.~\cite{Ritter:2015zur} for more general codomains.
    
    An obvious generalization of the source $\sCE(\frg)$ is obtained by replacing the graded vector space $\frg[1]$ by a more general, $\RZ$-graded vector space 
    \begin{equation}
        E=\bigoplus_{i\in \RZ} E_i~,
    \end{equation}
    again endowed with a nilquadratic vector field $Q$ of degree~1. The resulting differential graded algebras are the Chevalley--Eilenberg algebras of {\em $L_\infty$-algebras}.\footnote{Generalizing $E$ to a vector bundle directly yields Chevalley--Eilenberg algebras of \emph{$L_\infty$-algebroids}.} These are graded vector spaces $\frL=E[-1]$ together with a set of {\em higher products}
    \begin{equation}
        \mu_i\colon\frL^{\wedge i} \to \frL
    \end{equation}
    of degree $|\mu_i|=2-i$. The explicit form of the higher products can be derived from $Q_\sCE$; see appendix~\ref{app:L_infinity_defs} for explicit formulas and our conventions. Because $Q^2=0$, the higher products $\mu_i$ satisfy a generalization of the Jacobi identity, the \emph{homotopy Jacobi identity}; see appendix~\ref{app:L_infinity_defs}.
    If $\frL$ is an $L_\infty$-algebra with underlying graded vector space of the form
    \begin{equation}
        \frL=\frL_{-n+1}\oplus \frL_{-n+2}\oplus \dotsb \oplus \frL_{-1}\oplus \frL_0~,
    \end{equation}
    we say that $\frL$ is an \emph{\(n\)-term \(L_\infty\)-algebra}; it is a model of a Lie \(n\)-algebra. We call an \(L_\infty\)-algebra \emph{strict} if \(\mu_i=0\) for \(i\ge3\).
    
    Morphisms of differential graded algebras from the Chevalley--Eilenberg algebra $\sCE(\frL)$ of an $L_\infty$-algebra $\frL$ to the Weil algebra $\sW(U)$ of a patch $U$ of some manifold yield flat higher connections. General connections can be described by morphisms from the Weil algebra of $\frL$ to $\sW(U)$. The definition of the Weil algebra of an \(L_\infty\)-algebra is a straightforward generalization of the Weil algebra of a Lie algebra:
    \begin{equation}
        \sW(\frL)\coloneqq \big(\bigodot\nolimits^\bullet(\frL[1]^*\oplus \frL[2]^*),Q_\sW\big)~,~~~Q_\sW\coloneqq Q_\sCE+\sigma~,
    \end{equation}
    where $\sigma$ is the trivial extension of the shift isomorphism $\frL[1]^*\to \frL[2]^*$ to $\frL[1]^*\oplus \frL[2]^*$ and further, as a derivation, to $\bigodot\nolimits^\bullet(\frL[1]^*\oplus \frL[2]^*)$, and where $Q_\sCE$ is the extension of the Chevalley--Eilenberg differential by the rule $Q_\sCE \sigma \coloneqq -\sigma Q_\sCE$.
    
    For definiteness, consider a Lie 2-algebra $\frL=\frL_{-1}\oplus \frL_0$. Morphisms of differential graded algebras $\sW(\frL)\to \sW(U)$, where $U$ is a local patch of some manifold $M$, encode the following kinematical data:
    \begin{subequations}\label{eq:unadjusted_fields}
        \begin{align}
            A&\in \Omega^1(U)\otimes \frL_0~,\\
            B&\in \Omega^2(U)\otimes \frL_{-1}~,\\
            F&=\dd A+\tfrac12 \mu_2(A,A)+\mu_1(B)\in \Omega^2(U)\otimes \frL_0~,\\
            \dd F&=-\mu_2(A,F)+\mu_1(H)~,\notag\\
            H&=\dd B+\mu_2(A,B)+\tfrac{1}{3!}\mu_3(A,A,A)\in \Omega^3(U)\otimes \frL_{-1}~,\label{eq:unadjusted_fields_H}\\
            \dd H&=-\mu_2(A,H)+\mu_2(F,B)-\tfrac12\mu_3(A,A,F)~.\notag
        \end{align}
    \end{subequations}
    The infinitesimal gauge transformations are again induced by partially flat infinitesimal homotopies between two morphisms of differential graded algebras. Here, they are parametrized by
    \begin{equation}\label{eq:gauge_parameters}
        (\alpha,\Lambda)\in (\Omega^0(U)\otimes \frL_0)\times (\Omega^1(U)\otimes \frL_{-1})
    \end{equation}
    and read as 
    \begin{subequations}
        \begin{align}
            \delta A&=\dd \alpha-\mu_1(\Lambda)+\mu_2(A,\alpha)~,\\
            \delta B&=\dd \Lambda+\mu_2(A,\Lambda)+\mu_2(B,\alpha)-\tfrac12\mu_3(A,A,\alpha)~,\label{eq:unadjusted_B_gauge_transformation}
\\
            \delta F &= \mu_2(F,\alpha)~,\\
            \delta H &=\mu_2(H,\alpha)+\mu_2(F,\Lambda)-\mu_3(F,A,\alpha)~.\label{eq:unadjusted_H_gauge_transformation}
        \end{align}
    \end{subequations}
    The commutator of two infinitesimal gauge transformations is
    \begin{equation}
        [\delta_{\alpha+\Lambda},\delta_{\alpha'+\Lambda'}]=\delta_{\mu_2(\alpha+\Lambda,\alpha'+\Lambda')+\mu_3(A,\alpha+\Lambda,\alpha'+\Lambda')}+\mu_3(F,\alpha,\alpha')~,
    \end{equation}
    and we have run into the following severe limitation. Gauge transformations only close if the theory is abelian (and thus $\mu_i=0$ for $i\geq 2$) or if the {\em fake curvature}\footnote{In a general higher gauge theory, a \emph{fake curvature} is any curvature form other than the top form.} $F$ vanishes. The situation is not improved by restricting to strict $L_\infty$-algebras (for which $\mu_i$ with $i\geq 3$ vanishes), since there the condition $F=0$ reappears when composing finite gauge transformations. 
    
    Fake flatness also arises in the conventional definition of higher parallel transport; see~e.g.~\cite{Schreiber:0705.0452,Schreiber:2008aa}. Section~\ref{ssec:derived_parallel_higher} further discusses this point.
    
    For all these reasons, fake flatness $F=0$ is a fixed part of the conventional definition of connections on principal 2-bundles in the literature, cf.~\cite{Breen:math0106083,Aschieri:2003mw,Baez:2004in,Baez:0511710}.
    
    The fake flatness condition $F=0$ is now highly problematic due to the following theorem~\cite{Saemann:2019dsl}; see also~\cite{Gastel:2018joi} for a detailed analysis of the involved gauges:
    \begin{theorem}
        A connection on a non-abelian principal 2-bundle is locally gauge equivalent to a connection on an abelian principal 2-bundle.
    \end{theorem}
    This is somewhat surprising. Topologically, ordinary principal bundles are easily interpreted as principal 2-bundles. A Lie group $\sG$ is readily seen as a Lie 2-group, e.g.~in the form of the crossed module of Lie groups\footnote{See appendix~\ref{app:hypercrossed_modules} for definitions.} $*\to\sG$. The cocycles defining a principal 2-bundle with structure 2-group $*\to\sG$ are precisely those of an ordinary principal bundle. As soon as we endow the principal bundle with a connection, however, this embedding breaks; only flat principal bundles can be 2-bundles.

    We also note that the form of the gauge transformations of $H$ makes it very hard to imagine a covariant equation of motion. In particular, a non-abelian (2,0)-theory would involve the self-duality equation in six dimensions; however, the equation $H=\star H$ is not covariant unless $F=0$.
    
    The above observations are \emph{not} specific to kinematical data derived from Lie 2-algebras, but rather constitute a generic feature of higher gauge theories; see~e.g.~the discussion of homotopy Maurer--Cartan theory in~\cite{Jurco:2018sby}. Thus, higher gauge theory as conventionally defined is fake flat and locally abelian. This is well-known in the context of BRST/BV quantization, where these higher gauge theories lead to an ``open'' complex, which closes only modulo equations of motion.

    \subsection{Examples of adjusted Weil algebras}\label{ssec:adjusted_Weils}
    
    The problems outlined in the previous section can be eliminated for some gauge $L_\infty$-algebras by deforming their Weil algebras. This deformation was first discussed in the context of the string Lie 2-algebra in~\cite{Sati:2008eg,Sati:2009ic}; see also~\cite{Schmidt:2019pks} and~\cite{Saemann:2019dsl}. 
    
    Given an $L_\infty$-algebra $\frL$, the Weil algebra $\sW(\frL)$ projects onto the Chevalley--Eilenberg algebra $\sCE(\frL)$. We call a deformation $\sW_{\rm adj}(\frL)$ of $\sW(\frL)$ an \emph{adjusted Weil algebra}~\cite{Saemann:2019dsl}, if the underlying graded algebra is isomorphic to $\sW(\frL)$, the projection onto the Chevalley--Eilenberg algebra is not deformed, and the resulting BRST complex is closed. The last condition amounts to closure of the gauge transformations without any restriction on gauge parameters or gauge fields.
    
    This deformation to an adjusted Weil algebra can already be motivated on purely algebraic grounds: the Weil algebra contains the vector space of invariant polynomials, whose definition is only compatible with quasi-isomorphism\footnote{This is the appropriate notion of isomorphism here; see appendix~\ref{app:Equivalence}.} after the deformation; see~\cite{Saemann:2019dsl}.
    
    \paragraph{Skeletal model of the string 2-algebra.}
    As a first example, consider the skeletal string Lie 2-algebra 
    \begin{equation}
        \begin{array}{ccccccc}
            \astringsk(\frg)&=&(& \FR &\xrightarrow{~\mu_1~} &\frg&)\\
            &&& r&\xmapsto{~\mu_1~} & 0
        \end{array}
    \end{equation}
    for some metric Lie algebra\footnote{A \emph{metric Lie algebra} is a Lie algebra equipped with a nondegenerate (but not necessarily positive-definite) bilinear form \(\langle-,-\rangle\) that is invariant under the adjoint action.} $\frg$ with
    \begin{subequations}
        \begin{align}
            &\mu_2\colon\frg\wedge \frg \to \frg~, &\mu_2(a_1,a_2)&=[a_1,a_2]~,\\
            &\mu_3\colon\frg\wedge \frg\wedge \frg \to \FR~, &\mu_3(a_1,a_2,a_3)&=\big\langle a_1,[a_2,a_3]\big\rangle~,
        \end{align}
    \end{subequations}
    where $\langle-,-\rangle \colon\frg\times\frg\to\mathbb R$ denotes the invariant metric on $\frg$. Let $e_\alpha$ be a basis of $\frg$ and $f^\alpha_{\beta\gamma}$ and $\kappa_{\alpha\beta}$ be the structure constants and the components of the metric, respectively. The unadjusted Weil algebra is generated by coordinate functions $t^\alpha, r$ of degrees~1 and~2 respectively on $\frL[1]$ as well as their shifted copies $\hat t^\alpha=\sigma t^\alpha$ and $\hat r=\sigma r$ of degrees~2 and~3 respectively. The differential acts according to 
    \begin{equation}
        \begin{aligned}
            Q_\sW~&:~&t^\alpha &\mapsto -\tfrac12 f^{\alpha}_{\beta\gamma} t^\beta  t^\gamma + \hat t^\alpha~,~~~&r &\mapsto \tfrac1{3!} f_{\alpha\beta\gamma} t^\alpha  t^\beta   t^\gamma + \hat r~,\\
            &&\hat t^\alpha &\mapsto -f^\alpha_{\beta\gamma} t^\beta   \hat t^\gamma~,~~~&\hat r &\mapsto -\tfrac12 f_{\alpha\beta\gamma} t^\alpha  t^\beta   \hat t^\gamma~,
        \end{aligned}
    \end{equation}
    where $f_{\alpha\beta\gamma}\coloneqq \kappa_{\alpha\delta}f^\delta_{\beta\gamma}$. An adjusted form of this Weil algebra which we shall denote by $\sW_{\rm adj}(\astringsk(\frg))$ has (by definition) the same generators, but the differential $Q_{\sW_{\rm adj}}$ acts as
    \begin{equation}
        \begin{aligned}
            Q_{\sW_{\rm adj}}~&:~&t^\alpha &\mapsto -\tfrac12 f^\alpha_{\beta\gamma} t^\beta  t^\gamma + \hat t^\alpha~,~~~&r &\mapsto \tfrac{1}{3!} f_{\alpha\beta\gamma} t^\alpha  t^\beta  t^\gamma -\kappa_{\alpha\beta}t^\alpha\hat t^\beta+ \hat r~,\\
            &&\hat t^\alpha &\mapsto -f^\alpha_{\beta\gamma} t^\beta  \hat t^\gamma~,~~~&\hat r&\mapsto \kappa_{\alpha\beta}\hat t^\alpha\hat t^\beta~.
        \end{aligned}
    \end{equation}
    The kinematical data for a gauge theory on a local patch $U$ is then given by 
    \begin{subequations}
        \begin{align}
            A&\in\Omega^1(U)\otimes\frg~,\\
            B&\in\Omega^2(U)~,\\
            F&\coloneqq \dd A+\tfrac12 [A,A]\in \Omega^2(U)\otimes \frg~,&
            \dd F+[A,F]&=0~,\\
            H&\coloneqq \dd B-\tfrac{1}{3!}\mu_3(A,A,A)+\langle A,F\rangle\in \Omega^3(U)~,&
            \dd H-\langle F,F\rangle&=0
        \end{align}
    \end{subequations}
    with gauge transformations
    \begin{subequations}
        \begin{align}
            \delta A&=\dd \alpha+\mu_2(A,\alpha)~,\\
            \delta B&=\dd \Lambda+\langle\alpha,F\rangle-\tfrac12 \mu_3(A,A,\alpha)~,\\
            \delta F&=-\mu_2(F,\alpha)~,\\
            \delta H&=0~,
        \end{align}
    \end{subequations}
    where \(\alpha\in\Omega^0(U)\otimes\frg\) and \(\Lambda\in\Omega^1(U)\) parametrize infinitesimal gauge transformations. The commutator of two gauge transformations now closes as expected, and the BRST complex of these fields is indeed closed; see~\cite{Saemann:2019dsl}. Moreover, writing down covariant field equations for $H$ has become easier.
    
    Such connections arise naturally in the context of heterotic supergravity, as well as in non-abelian self-dual strings and six-dimensional superconformal field theories~\cite{Saemann:2017rjm,Saemann:2017zpd,Saemann:2019dsl}. For references to the original literature on string structures and a detailed explanation of their relevance, see also~\cite{Saemann:2017zpd}.
    
    \paragraph{Loop model of the string 2-algebra.}
    Since we wish to discuss parallel transport, we need finite descriptions of gauge transformations and their actions. The skeletal model is not a strict \(L_\infty\)-algebra, and hence not well suited for integration. It is more convenient to work with the loop model, which is quasi-isomorphic\footnote{Quasi-isomorphism implies that the two models define physically equivalent kinematical data; see appendix~\ref{app:Equivalence} for definitions.} to the skeletal model:
    \begin{equation}
        \begin{array}{ccccccc}
            \astringl(\frg)&=&(& \hat L_0\frg &\xrightarrow{~\mu_1~} &P_0\frg&)\\
            &&& (\lambda,r)&\xmapsto{~\mu_1~} & \lambda
        \end{array}
    \end{equation}
    where $P_0\frg$ and $L_0\frg$ are based path and loop spaces, respectively, of $\frg$ and $\hat L_0\frg=L_0\frg\oplus\FR$ is the vector space underlying the Lie algebra obtained by the Kac--Moody extension; for technical details see appendix~\ref{app:path_groups}. The loop model is a strict 2-term \(L_\infty\)-algebra (i.e.~\(\mu_i=0\) for \(i\ge3\)); the unary product \(\mu_1\) was given above, and the binary product \(\mu_2\) is as follows:
    \begin{subequations}
        \begin{align}
            P_0\frg\wedge  P_0\frg&\to P_0\frg~,&(\gamma_1,\gamma_2)&\mapsto[\gamma_1,\gamma_2]~,\\
            P_0\frg\otimes\hat L_0 \frg[1]&\to \hat L_0 \frg[1]~,~~~&\big(\gamma,(\lambda,r)\big)&\mapsto\left([\gamma,\lambda]\; ,\; -2\int_0^1 \dd\tau \big\langle\gamma(\tau),\dot \lambda(\tau)\big\rangle\right)~,
        \end{align}
    \end{subequations}
    where $\dot{-}$ labels the derivative with respect to the path or loop parameter.
    
    The corresponding Weil algebra is generated by coordinate functions $t^{\alpha \tau}$, $r^{\alpha \tau}$, $r_0$ as well as their shifted counterparts. The differential \(Q_\sW\) acts as
    \begin{equation}\label{eq:Weil_string_Lie_2_loop}
        \begin{aligned}
            t^{\alpha\tau} &\mapsto -\tfrac12 f^{\alpha}_{\beta\gamma} t^{\beta\tau} t^{\gamma\tau} - r^{\alpha\tau} +\hat{t}^{\alpha\tau}~,~~~ 
            &\hat{t}^{\alpha\tau}&\mapsto -f^\alpha_{\beta\gamma}t^{\beta\tau}\hat{t}^{\gamma\tau}+\hat{r}^{\alpha\tau}~,\\
            r^{\alpha\tau} &\mapsto -f^\alpha_{\beta\gamma} t^{\beta\tau}r^{\gamma\tau} + \hat{r}^{\alpha\tau}~,~~~
            &\hat{r}^{\alpha\tau}&\mapsto-f^\alpha_{\beta\gamma} t^{\beta\tau}\hat{r}^{\gamma\tau} + f^\alpha_{\beta\gamma} \hat{t}^{\beta\tau}r^{\gamma\tau} ~,\\
            r_0 &\mapsto 2\int_0^1\dd \tau\, \kappa_{\alpha\beta}t^{\alpha\tau} \dot{r}^{\beta\tau}+\hat{r}_0 ~,~~~&\hat{r}_0&\mapsto 2\int_0^1\dd\tau\,\kappa_{\alpha\beta}\left(t^{\alpha\tau} \smash{\dot{\hat{r}}}^{\beta\tau}-\hat{t}^{\alpha\tau}\dot{r}^{\beta\tau}\right)~,
        \end{aligned}
    \end{equation}
    and here, an adjustment reads as
    \begin{equation}\label{eq:Weil_string_Lie_2_loop_adj}
        \begin{aligned}
            Q_{\sW_{\rm adj}}:~t^{\alpha\tau} &\mapsto -\tfrac12 f^{\alpha}_{\beta\gamma} t^{\beta\tau} t^{\gamma\tau} - r^{\alpha\tau} +\hat{t}^{\alpha\tau}~,~~~&\hat{t}^{\alpha\tau}&\mapsto -f^\alpha_{\beta\gamma}t^{\beta\tau}\hat{t}^{\gamma\tau}+\chi^{\alpha\tau}(t,\hat{t})+\hat{r}^{\alpha\tau}~,\\
            r^{\alpha\tau} &\mapsto -f^\alpha_{\beta\gamma} t^{\beta\tau}r^{\gamma\tau} +\chi^{\alpha\tau}(t,\hat t)+ \hat{r}^{\alpha\tau}~,~
            &\hat{r}^{\alpha\tau}&\mapsto0~,\\
            r_0 &\mapsto 2\int_0^1\dd\tau\,\kappa_{\alpha\beta}t^{\alpha\tau} \dot r^{\beta\tau}+\chi(\dot t,\hat t)+\hat{r}_0&\hat{r}_0&\mapsto -\chi(\dot{\hat t},\hat t)~,
        \end{aligned}
    \end{equation}
    where we introduced a function $\chi$ with components
    \begin{subequations}
    \begin{align}
        \chi^{\alpha\tau}(t,\hat t)&\coloneqq f^\alpha_{\beta\gamma}(t^{\beta\tau}\hat t^{\gamma\tau}-\ell(\tau)t^{\beta 1}\hat t^{\gamma 1})~,\\
        \chi(\dot t,\hat{t})&\coloneqq 2\int_0^1\dd\tau\,\kappa_{\alpha\beta} \dot t^{\alpha\tau} \hat t^{\beta\tau}~,\\
        \chi(\dot{\hat t},\hat{t})&\coloneqq 2\int_0^1\dd\tau\,\kappa_{\alpha\beta} \dot{\hat t}^{\alpha\tau} \hat t^{\beta\tau}~,
    \end{align}
    \end{subequations}
    and where $\ell(\tau)$ is an arbitrary smooth function $\ell\colon[0,1]\to [0,1]$ with $\ell(0)=0$ and $\ell(1)=1$. The kinematical data encoded in a morphism $\sW_{\rm adj}(\astringl(\frg))\to \sW(U)$ is then
    \begin{subequations}\label{eq:adjusted_loop_fields}
        \begin{align}
            A&\in \Omega^1(U)\otimes P_0\frg~,\\
            B&\in\Omega^2(U)\otimes \hat L_0\frg~,\\
            F&\coloneqq \dd A+\tfrac12 [A,A]+\mu_1(B)~, &
            \dd F+[A,F]-\mu_1(\chi(A,F))&=\mu_1(H)~,\label{eq:fake_curvature_bianchi_identity}\\
            H&\coloneqq \dd B+\mu_2(A,B)-\chi(A,F)~, &
            \dd H+\chi(F,F)&=0\label{eq:H_bianchi_identity}
        \end{align}
    \end{subequations}
    with gauge transformations
    \begin{subequations}
        \begin{align}
            \delta A&=\dd \alpha+\mu_2(A,\alpha)+\mu_1(\Lambda)~,\\
            \delta B&=\dd \Lambda+\mu_2(A,\Lambda)+\mu_2(\alpha,B)-\chi(\alpha,F)~,\\
            \delta F&=-\mu_1(\chi(\alpha,F))-\mu_2(F,\alpha)~,\\
            \delta H&=0~,
        \end{align}
    \end{subequations}
    where the gauge transformations are parametrized by elements
    \begin{equation}
        \alpha\in\Omega^0(U)\otimes P_0\frg\eand
        \Lambda\in\Omega^1(U)\otimes \hat L_0 \frg
    \end{equation}
    and where $\chi$ is here the function
    \begin{equation}\label{eq:ext-kappa-definition}
        \begin{aligned}
            \chi & \colon &P_0\frg \times P_0\frg &\to \hat L_0\frg \\
            &&(\gamma_1,\gamma_2)&\mapsto \left([\gamma_1,\gamma_2] - \ell(\tau)\partial([\gamma_1,\gamma_2]), 2\int_0^1\dd\tau \langle\dot\gamma_1,\gamma_2\rangle \right).
        \end{aligned}
    \end{equation}
    
    If we now look at just the transformations parametrized by $\alpha$ and trivial $\Lambda$, the various fields transform under different \(P_0\sG\)-representations, as a result of the adjustment. For example, \(H\), which before adjustment transformed under the adjoint representation of \(\sG\), is now invariant. Similarly, the fake curvature \(F\) now transforms differently, and the covariant derivative acts on it as
    \begin{equation}
        \mathrm d_AF \coloneqq \dd F+[A,F] - \mu_1(\chi(A,F))~,
    \end{equation}
    which can be seen from~\eqref{eq:fake_curvature_bianchi_identity}. The 2-form potential \(B\), which used to transform on its own, now forms a multiplet with \(F\), unlike in the unadjusted case.\footnote{After adjustment, the fake curvature \(F\) still transforms as a representation on its own, but \(B\) only forms a representation together with \(F\).} This reflects the fact that the adjustment of the Weil algebra requires an adjustment of the 2-crossed module (in which the parallel transport functor takes value) encoding the representations.
    
    The advantage of the crossed module of Lie algebras $\astring_{\rm lp}(\frg)$ over the 2-term $L_\infty$-algebra $\astring_{\rm sk}(\frg)$ is now that it readily integrates to the crossed module of Lie groups
    \begin{equation}
        \sString_{\rm lp,\cm}(\sG)=(L_0\sG\rightarrow P_0\sG)~.
    \end{equation}
    The integration of $\astring_{\rm sk}(\frg)$ is much harder; see~\cite{Schommer-Pries:0911.2483,Demessie:2016ieh}.

    \paragraph{The loop model of a Lie algebra.}
    There is an interesting truncation in the loop model of the string Lie 2-algebra, namely the 2-term \(L_\infty\)-algebra
    \begin{equation}\label{eq:loop_model_of_Lie_algebra}
        \frg_{\rm lp}\coloneqq  (L_0\frg\hooklongrightarrow P_0\frg)~.
    \end{equation}
    This 2-term \(L_\infty\)-algebra is quasi-isomorphic to the Lie algebra $\frg$; see appendix~\ref{app:Equivalence}. Together with the adjustment, it also allows us to interpret the connection on an ordinary principal fiber bundle as a connection on a principal 2-bundle~\cite{Saemann:2017rjm}.
    We can construct a $\frg_{\rm lp}$-valued connection \((A_\text{lp},B_\text{lp})\) from a $\frg$-connection \(A\in\Omega^1(M)\otimes\frg\) as
    \begin{equation}
        A_{\rm lp}=A\ell(\tau)~,~~~B_{\rm lp}=\tfrac12[A,A](\ell(\tau)-\ell^2(\tau))~.
    \end{equation}
    Note that 
    \begin{equation}
        F_{\rm lp}=\dd A_{\rm lp}+\tfrac12[A_{\rm lp},A_{\rm lp}]+\sft(B_{\rm lp})=\ell(\tau)F_{\rm sk}~. 
    \end{equation}
    Infinitesimal gauge transformations translate according to
    \begin{equation}
        \alpha_{\rm lp}=\alpha_{\rm sk}\ell(\tau)\eand \Lambda_{\rm lp}=[\alpha_{\rm sk},A_{\rm sk}](\ell(\tau)-\ell^2(\tau))~.
    \end{equation}
    Thus, gauge transformations are mapped to gauge transformations and gauge orbits are mapped to gauge orbits. The inverse map is the endpoint evaluation map $\dpar\colon P_0\frg\to \frg$:
    \begin{equation}
        A_{\rm sk}=\dpar A_{\rm lp}\eand \alpha_{\rm sk}=\dpar \alpha_{\rm lp}~.
    \end{equation}
    
    We use both 2-term \(L_\infty\)-algebras $\astringl(\frg)$ and $\frg_{\rm lp}$ as examples for our further discussion leading to an adjusted parallel transport.

    \section{Weil algebras and inner derivations}
    
    The Weil algebra can be interpreted as the inner derivation 2-crossed module of Lie algebras, and the exponentials of potentials and curvatures take values in the Lie 2-group corresponding to this 2-crossed module. After we adjust the Weil algebra, we need to construct the corresponding adjusted 2-crossed module and the Lie 2-group. This is a prerequisite to discussing the parallel transport functor, which takes values in this 2-group. 
    
     \subsection{Inner automorphisms of Lie groups}\label{ssec:inner_group}
    
    The Weil algebra $\sW(\frg)$ of a Lie algebra $\frg$ encodes a 2-term \(L_\infty\)-algebra with underlying graded vector space $\frg\oplus\frg[1]$, which is isomorphic to the 2-term \(L_\infty\)-algebra of inner derivations, $\ainn(\frg)$; see~appendix~\ref{app:inner_weil}. The latter sits in the short exact sequence of graded vector spaces 
    \begin{equation}
        *\longrightarrow 
        \begin{array}{c} *\\[-0.75ex] \downarrow\\[-0.75ex]\frg \end{array}\hooklongrightarrow \ainn(\frg)\longrightarrow 
        \begin{array}{c} \frg[1]\\[-0.3ex] \downarrow\\[-0.75ex]* \end{array}\longrightarrow *~.
    \end{equation}
    For a Lie group \(\sG\) integrating $\frg$, this sequence is the infinitesimal version of the short exact sequence of groupoids,
    \begin{equation}\label{ses:groupoids_Lie}
        *\longrightarrow \begin{array}{c} \sG\\[-0.75ex] \downdownarrows\\[-0.3ex]\sG\end{array} \hooklongrightarrow \sInn(\sG)\longrightarrow \begin{array}{c} \sG\\[-0.75ex] \downdownarrows\\[-0.75ex]*\end{array}\longrightarrow *~.
    \end{equation}
    Here, $\sG\rightrightarrows \sG$ is the Lie group $\sG$, trivially regarded as a groupoid, while $(\sG\rightrightarrows *)=\sB\sG$ is the one-object groupoid with $\sG$ as the group of morphisms. Moreover, $\sInn(\sG)$ is the action groupoid of $\sG$ onto itself by left-multiplication. This is a 2-group, the {\em 2-group of inner automorphisms} of $\sG$. The embedding in the sequence~\eqref{ses:groupoids_Lie} is in fact a morphism of Lie 2-groups, while the second map is merely a groupoid morphism; the groupoid $\sG\rightrightarrows *$ does not admit a 2-group structure unless $\sG$ is abelian.
    
    These structures have important topological interpretations. The geometric realization $|\sB\sG|$ of the nerve of $\sB\sG$ is the classifying space of $\sG$. Applying the same operations to $\sInn(\sG)$, we recover the universal bundle $|\sE\sG|$ of $\sG$ over $|\sB\sG|$. Also, the action groupoid $\sInn(\sG)=(\sG\rtimes \sG\rightrightarrows \sG)$ is equivalent to the trivial 2-group $(*\rightrightarrows *)$; equivalently, $\ainn(\frg)=(\frg[1]\xrightarrow{~\id~} \frg)$ is quasi-isomorphic to the 0-term \(L_\infty\)-algebra\footnote{by the minimal model theorem; see~appendix~\ref{app:Equivalence}}. This corresponds to the universal bundle \(|\sE\sG|\) being contractible.

    \subsection{Inner automorphisms of strict Lie 2-groups}\label{ssec:inner_derivations}
    
    The generalization to the case of a strict Lie 2-group is discussed in detail in~\cite{Roberts:0708.1741}. The inner automorphisms of a strict Lie 2-group $\CCG$ with corresponding crossed module of Lie groups $\CCG_{\rm \cm}=(\sH\xrightarrow{~\tilde \sft~}\sG,\tildeacton)$ form a Lie 3-group, which is conveniently encoded in the following 2-crossed module\footnote{see again appendix~\ref{app:hypercrossed_modules} for the definition} of Lie groups \((\sInn_{\rm \cm}(\CCG),\acton,\{-,-\})\):
    \begin{equation}
        \begin{matrix}
            \sInn_{\rm \cm}(\CCG)&=&(&\sH&\xrightarrow{~\sft~}&\sH\rtimes \sG&\xrightarrow{~\sft~}&\sG&)\\
            &&&h&\xmapsto{~\sft~}&( h^{-1},\tilde\sft(h))\\
            &&& & &(h,g)&\xmapsto{~\sft~}&\tilde \sft(h)g
        \end{matrix}
    \end{equation}
    where the products and actions are evident, in particular
    \begin{equation}
        (h_1,g_1)(h_2,g_2)=\big(h_1(g_1\tildeacton h_2),g_1g_2\big)~,~~~(h_1,g_1)^{-1}=(g_1^{-1}\tildeacton h_1^{-1},g_1^{-1})~,
    \end{equation}
    and the Peiffer lifting is
    \begin{equation}
        \{(h_1,g_1),(h_2,g_2)\}=(g_1g_2g_1^{-1}\tildeacton h_1) h_1^{-1}
    \end{equation}
    for all $g_1,g_2\in \sG$ and $h_1,h_2\in \sH$.    
    
    There is now a higher analogue of sequence~\eqref{ses:groupoids_Lie} involving 2-groupoids. The crossed module of Lie groups $\CCG_{\rm \cm}$ corresponds to a monoidal category $\CCG=(\sH\rtimes \sG\rightrightarrows \sG)$ (see equation~\eqref{eq:mon_cat_from_2_group}), which is trivially regarded as a strict 2-category with only identity 2-morphisms. Moreover, $\sInn_{\rm \cm}(\sH\xrightarrow{\tilde \sft}\sG)$ corresponds to a monoidal 2-category $\sInn(\CCG)$ encoding a 3-group\footnote{more precisely, a \emph{Gray group}; see equation~\eqref{eq:mon_cat_from_3_group}}. We present its globular structure for use in section~\ref{sec:parallel_transport}.
    \begin{subequations}
    \begin{equation}
        \sInn(\CCG)\coloneqq\big(\sH\rtimes((\sH\rtimes \sG)\rtimes \sG)  \rightrightarrows  (\sH\rtimes \sG)\rtimes \sG  \rightrightarrows \sG\big)
    \end{equation}
    \begin{equation}
        \begin{tikzcd}[column sep=4.0cm,row sep=large]
            \tilde \sft(h^2_2)g_2^2g_2^1 & \ar[l, bend left=45, "{\big((h_2^2,g_2^2),g_2^1\big)}", ""{name=U,inner sep=1pt,above}] \ar[l, bend right=45, "{\big((h_2^2(h_2^1)^{-1}\,,\,\tilde \sft(h_2^1)g_2^2),g_2^1\big)}", swap, ""{name=D,inner sep=1pt,below}] g_2^1
            \arrow[Rightarrow,from=U, to=D, "{\big(h_2^1,(h_2^2,g_2^2),g_2^1\big)}",swap]
        \end{tikzcd}
    \end{equation}
    \end{subequations}
    Finally, we have the 2-groupoid\footnote{The component \(\sH\times\sG\) in~\eqref{eq:2-groupoid_BG} is merely a manifold, not a Lie group, since \(\sB\CCG\) is not a 3-group, but merely a 2-groupoid.}
    \begin{equation}\label{eq:2-groupoid_BG}
        \sB\CCG=\big((\sH\times \sG)  \rightrightarrows \sG \rightrightarrows *\big)~.
    \end{equation}
    These three 2-groupoids now fit in the short exact sequence
    \begin{subequations}\label{eq:ses_2-groupoids}
        \begin{equation}
            * \longrightarrow \CCG \xrightarrow{~\Upsilon~} \sInn(\CCG) \xrightarrow{~\Pi~} \sB\CCG \longrightarrow *~,
        \end{equation}
        whose components are as follows:
        \begin{equation}
            \begin{tikzcd}
                \sH\rtimes\sG \rar["\Upsilon_2"] \dar[shift left] \dar[shift right] & \sH\rtimes\big((\sH\rtimes\sG)\rtimes\sG\big) \rar["\Pi_2"] \dar[shift left] \dar[shift right] & \sH\times \sG \dar[shift left] \dar[shift right] \\
                \sH\rtimes\sG \rar["\Upsilon_1"] \dar[shift left] \dar[shift right] & (\sH\rtimes\sG)\rtimes\sG \rar["\Pi_1"] \dar[shift left] \dar[shift right] & \sG \dar[shift left] \dar[shift right] \\
                \sG \rar["\Upsilon_0"] & \sG \rar["\Pi_0"] & * 
            \end{tikzcd}
        \end{equation}
        where the strict 2-functors $\Upsilon$ and $\Pi$ are given by
        \begin{equation}\label{eq:3-group_short_exact_sequence_components}
            \begin{aligned}
                \Upsilon_2\colon&&(h,g) &\mapsto (\unit_\sH,h,g,\unit_\sG)~,~~~&\Pi_2\colon&&(h_1,h_2,g_1,g_2)&\mapsto (h_1,g_2)~,\\
                \Upsilon_1\colon&&(h,g) &\mapsto (h,g,\unit_\sG)~,~~~&\Pi_1\colon&&(h,g_1,g_2)&\mapsto g_2~,\\
                \Upsilon_0\colon&&g &\mapsto g~,~~~&\Pi_0\colon&&g&\mapsto *~.
            \end{aligned}
        \end{equation}
    \end{subequations}
    Again, $\Upsilon$ is also a morphism of strict 3-groups.
    
    At an infinitesimal level, $\CCG$ (and, more evidently, $\CCG_{\rm \cm}$) differentiates to the crossed module of Lie algebras $(\frh\xrightarrow{~\tilde \sft~}\frg)$, where $\frg$ and $\frh$ are the Lie algebras of $\sG$ and $\sH$. Its 2-crossed module of inner derivations has the underlying complex~\cite{Roberts:0708.1741}
    \begin{equation}\label{eq:complex_inn_2cm}
        \begin{array}{ccccccccc}
            \ainn(\frh\xrightarrow{\tilde \sft}\frg)&=&(& 
            \frh &\xrightarrow{~\sft~}&\frh \rtimes \frg&\xrightarrow{~\sft~}&\frg&)~,\\
            &&&b &\xmapsto{~\sft~} & \big(-b,\tilde\sft(b)\big)\\
            &&&&& (b,a) &\xmapsto{~\sft~} & \tilde\sft(b)+a
        \end{array}
    \end{equation}
    with the $\frg$-actions 
    \begin{equation}
        a\acton b \coloneqq  a \tildeacton b\eand a_1 \acton (b,a_2)\coloneqq (a_1\tildeacton b,[a_1,a_2])
    \end{equation}
    and the usual Lie bracket on $\frh \rtimes \frg$, viz. 
    \begin{equation}
        [(b_1,a_1),(b_2,a_2)]\coloneqq  \big([b_1,b_2]+a_1\tildeacton b_2-a_2\tildeacton b_1,[a_1,a_2]\big)~,
    \end{equation}
    leading to the Peiffer lifting
    \begin{equation}
        \{(b_1,a_1),(b_2,a_2)\}\coloneqq a_2\tildeacton b_1
    \end{equation}
    for all $a_1,a_2\in\frg$, $b_1,b_2\in \frh$.
    
    The infinitesimal version of the short exact sequence of 2-groupoids~\eqref{eq:ses_2-groupoids} is the following short exact sequence of graded vector spaces:
    \begin{equation}\label{eq:ses_3-vector}
        \begin{tikzcd}
            * \rar["\upsilon_2"] \dar & \frh\rar["\pi_2"] \dar & \frh \dar \\
            \frh\rar["\upsilon_1"] \dar & \frh\rtimes \frg \rar["\pi_1"] \dar & \frg \dar\\
            \frg \rar["\upsilon_0"] & \frg \rar["\pi_0"] & * 
        \end{tikzcd}
        ~,\hspace{2cm}
        \begin{aligned}
         \upsilon_2\colon&&*&\mapsto 0~,~~~&\pi_2\colon&&b&\mapsto b~,\\[0.5cm]
         \upsilon_1\colon&&b&\mapsto (b,0)~,~~~&\pi_1\colon&&(b,a)&\mapsto a~,\\[0.5cm]
         \upsilon_0\colon&&a&\mapsto a~,~~~&\pi_0\colon&&a&\mapsto *~.
        \end{aligned}
    \end{equation}
    
    Every 2-crossed module of Lie algebras defines a strict 3-term \(L_\infty\)-algebra, while a strict 3-term \(L_\infty\)-algebra \emph{almost} determines a 2-crossed module, with the missing data being the antisymmetric part \(\llbracket-,-\rrbracket\) of the Peiffer lifting \(\{-,-\}\); see appendix~\ref{app:Lie3_2_crossed}. Are there 2-crossed modules corresponding to the unadjusted and adjusted Weil algebras? In both cases, the answer is yes. The unadjusted Weil algebra corresponds to the inner derivation 2-crossed module; see appendix~\ref{app:inner_weil}. The case of the adjusted Weil algebra is treated in section~\ref{ssec:adjusted_inner_derivation}.

    \subsection{Simplification by coordinate transformation}\label{ssec:simplification}
    
    It is convenient to slightly simplify the description of the inner automorphism 3-group and related Lie 3-algebras. This does not change the definitions, but merely the descriptions.
    
    In the semidirect product \(\frh\rtimes\frg\) in the complex~\eqref{eq:complex_inn_2cm}, we define the Lie subalgebras
    \begin{align}
        \frh' &\coloneqq \operatorname{im}\sft \subseteq\frh\rtimes\frg~,&\frg' &\coloneqq \ker\sft\subseteq\frh\rtimes\frg~,
    \end{align}
    which are isomorphic to \(\frh\) and \(\frg\), respectively.
    The inner semidirect product \(\frh'\rtimes\frg'\) equals the whole Lie algebra \(\frh\rtimes\frg\).  So we can use the primed coordinates to talk about \(\frh\rtimes\frg=\frh'\rtimes\frg'\).
    This amounts to a coordinate transformation (or reparametrization) of $\frh\rtimes \frg$ to $\frh'\rtimes \frg'$,
    \begin{equation}
       (b,a)\mapsto (b',a')\coloneqq(-b,a+\tilde\sft(b))~,
    \end{equation}
    and it simplifies the differentials in the complex~\eqref{eq:complex_inn_2cm} as follows:
    \begin{equation}\label{eq:inner_algebra_alternative_basis}
        \begin{matrix}
            \ainn(\frh\xrightarrow{\tilde \sft}\frg)&=&(& 
            \frh &\xrightarrow{~\sft~}&\frh' \rtimes \frg'&\xrightarrow{~\sft~}&\frg&)~,\\
            &&&b &\xmapsto{~\sft~} & \big(b,0\big)\\
            &&&&& (b,a) &\xmapsto{~\sft~} & a
        \end{matrix}
    \end{equation}
    The changes to the 2-crossed module structure maps under this reparametrization are readily derived; we merely note that the semidirect product structure is preserved. Under this coordinate change, the presentations of the chain maps $\upsilon$ and $\pi$ in~\eqref{eq:ses_3-vector} change to
    \begin{equation}
        \begin{tikzcd}
            * \rar["\upsilon_2"] \dar & \frh\rar["\pi_2"] \dar & \frh \dar \\
            \frh\rar["\upsilon_1"] \dar & \frh'\rtimes \frg' \rar["\pi_1"] \dar & \frg \dar\\
            \frg \rar["\upsilon_0"] & \frg \rar["\pi_0"] & * 
        \end{tikzcd}
        ~,\hspace{1cm}
        \begin{aligned}
         \upsilon_2\colon&&*&\mapsto 0~,~~~&\pi_2\colon&&b&\mapsto b~,\\[0.5cm]
         \upsilon_1\colon&&b&\mapsto (-b,\tilde \sft(b))~,~~~&\pi_1\colon&&(b,a)&\mapsto \tilde \sft(b)+a~,\\[0.5cm]
         \upsilon_0\colon&&a&\mapsto a~,~~~&\pi_0\colon&&a&\mapsto *~.
        \end{aligned}
    \end{equation}

    At the finite level, i.e.~the level of the 2-crossed module of Lie groups $\sInn_{\rm \cm}(\CCG)$, we have corresponding Lie closed subgroups
    \begin{align}
        \sH' &\coloneqq \exp\frh' \le \sH \rtimes\sG ~,&
        \sG' &\coloneqq \exp\frg' \le \sH \rtimes\sG ~,
    \end{align}
    and a corresponding reparametrization of $\sH\rtimes \sG$ as $\sH'\rtimes \sG'$,
    \begin{equation}
        (h,g)\mapsto(h',g')\coloneqq(h^{-1},\tilde \sft(h)g)~,
    \end{equation}
    leading to the normal complex
    \begin{equation}\label{eq:inner_group_alternative_basis}
        \begin{matrix}
            \sInn_{\rm \cm}(\CCG)&=&(&\sH&\xrightarrow{~\sft~}&\sH'\rtimes \sG'&\xrightarrow{~\sft~}&\sG&)~,\\
            &&&h&\xmapsto{~\sft~}&(h,\unit_\sH)\\
            &&& & &(h,g)&\xmapsto{~\sft~}& g
        \end{matrix}
    \end{equation}
    The presentations of the functors in the short exact sequence~\eqref{eq:ses_2-groupoids} change in the obvious manner; in particular,
    \begin{equation}\label{eq:ses_2-groupoids_reparametrization}
        \Upsilon_1\colon(h,g)\mapsto(h^{-1},g,\tilde \sft(h))\eand \Pi_1\colon(h,g_1,g_2)\mapsto \tilde\sft(h)g_2~.
    \end{equation}

    \subsection{Example: loop model of a Lie algebra}\label{ssec:loop_model_of_Lie_algebra}
    
    Before treating the adjusted Weil algebra of the string Lie 2-algebra, we first consider the simpler example of the adjusted and unadjusted Weil algebras of the 2-term \(L_\infty\)-algebra $\frg_{\rm lp}\coloneqq L_0\frg\overset\sft\to P_0\frg$, which is quasi-isomorphic to the Lie (1-)algebra $\frg$; see~\eqref{eq:quasi-iso-morphs}.
    
    The Weil algebra $\sW(\frg_{\rm lp})$ is generated by coordinate functions $(t^{\alpha \tau},r^{\alpha \tau},\hat t^{\alpha \tau},\hat r^{\alpha \tau})$, cf.~the similar parametrization of $\sW(\astring_{\rm lp}(\frg))$ in~\eqref{eq:Weil_string_Lie_2_loop}. We first perform the reparametrization explained in the previous section, which amounts to the coordinate change 
    \begin{equation}\label{eq:coord_change}
        (t^{\alpha \tau},r^{\alpha \tau},\hat t^{\alpha \tau},\hat r^{\alpha \tau})\rightarrow(t^{\alpha \tau},r^{\alpha \tau},\tilde t^{\alpha \tau},\hat r^{\alpha \tau})\ewith\tilde t^{\alpha \tau}=\hat t^{\alpha \tau}-r^{\alpha \tau}~.
    \end{equation}
    This simplifies the differential of the unadjusted Weil algebra to
    \begin{equation}
        \begin{aligned}
            Q_\sW~&\colon~&t^{\alpha\tau} &\mapsto -\tfrac12 f^{\alpha}_{\beta\gamma} t^{\beta\tau} t^{\gamma\tau} +\tilde{t}^{\alpha\tau}~,~~~ 
            &\tilde{t}^{\alpha\tau}&\mapsto -f^\alpha_{\beta\gamma}t^{\beta\tau}\tilde{t}^{\gamma\tau}~,\\
            &&r^{\alpha\tau} &\mapsto -f^\alpha_{\beta\gamma} t^{\beta\tau}r^{\gamma\tau} + \hat{r}^{\alpha\tau}~,~~~
            &\hat{r}^{\alpha\tau}&\mapsto-f^\alpha_{\beta\gamma} t^{\beta\tau}\hat{r}^{\gamma\tau} + f^\alpha_{\beta\gamma} \tilde{t}^{\beta\tau}r^{\gamma\tau} ~.
        \end{aligned}
    \end{equation}
    The differential of the adjusted Weil algebra also simplifies to
    \begin{equation}
        \begin{aligned}
            Q_{\sW_{\rm adj}}\colon~t^{\alpha\tau} &\mapsto -\tfrac12 f^{\alpha}_{\beta\gamma} t^{\beta\tau} t^{\gamma\tau} + \tilde t^{\alpha \tau}~,~~~&\tilde{t}^{\alpha\tau}&\mapsto -f^\alpha_{\beta\gamma}t^{\beta\tau}\tilde {t}^{\gamma\tau}~,\\
            r^{\alpha\tau} &\mapsto f^\alpha_{\beta\gamma}(t^{\beta\tau}\tilde t^{\gamma\tau}-\ell(\tau)t^{\beta 1}\tilde t^{\gamma 1})+ \hat{r}^{\alpha\tau}~,~
            &\hat{r}^{\alpha\tau}&\mapsto0~.
        \end{aligned}
    \end{equation}
    
    We now focus on the adjusted Weil algebra, as the unadjusted case is trivially constructed following the discussions in section~\ref{ssec:inner_derivations} and appendix~\ref{app:inner_weil}. Dualization to a 3-term $L_\infty$-algebra yields the complex of Lie algebras 
    \begin{equation}
        \begin{array}{cccccccc}
            \sW_\text{adj}(L_0\frg \to P_0\frg) &=(&L_0 \frg &\xrightarrow{~\mu_1~} & L_0\frg'\rtimes P_0\frg' & \xrightarrow{~\mu_1~} & P_0\frg&)\\
            &&\lambda &\xrightarrow{~\mu_1~} &(\lambda,0)\\
            &&&&(\lambda,\gamma)&\xrightarrow{~\mu_1~}&\gamma
        \end{array}
    \end{equation}
    endowed with binary products
    \begin{subequations}
        \begin{align}
            \mu_2&\colon&P_0\frg\wedge P_0\frg&\to P_0 \frg~, &
            (\gamma_1,\gamma_2)&\mapsto [\gamma_1,\gamma_2]~,
            \\
            &&P_0\frg\wedge(L_0\frg'\rtimes P_0\frg')&\to L_0\frg'\rtimes P_0\frg'~, &
            (\gamma_1,(\lambda_2,\gamma_2))&\mapsto\big(-\chi(\gamma_1,\gamma_2),[\gamma_1,\gamma_2]\big)~,\\
            &&P_0\frg\wedge L_0\frg &\to L_0\frg~, &
            (\gamma_1,\lambda_2)&\mapsto0~,
            \\
            &&(L_0\frg'\rtimes P_0'\frg)^{\times2}&\to L_0 \frg~, &
            \big((\lambda_1,\gamma_1),(\lambda_2,\gamma_2)\big)&\mapsto0~,
        \end{align}
    \end{subequations}
    where
    \begin{align}\label{eq:chi_loop_model_of_Lie_algebra}
        \chi&\colon & P_0\mathfrak g \times P_0\mathfrak g &\to L_0\mathfrak g~, & (\gamma_1,\gamma_2) & \mapsto [\gamma_1,\gamma_2] - \ell \cdot \partial[\gamma_1,\gamma_2]
    \end{align}
    is the projection of the Lie bracket of two paths to based loops.
    
    Just as in the unadjusted case (see~appendix~\ref{app:inner_weil}), the adjusted Weil algebra admits a lift to a 2-crossed module. There are, in fact, two possible 2-crossed modules of Lie algebras. Both options have the same underlying complex of graded vector spaces,
    \begin{equation}
        L_0 \frg \quad \overset\sft\longrightarrow \quad L_0\frg\oplus P_0\frg \quad\overset\sft\longrightarrow\quad P_0\frg~,
    \end{equation}
    but their Lie brackets, induced by a choice of the antisymmetric parts of the Peiffer liftings \(\llbracket-,-\rrbracket\), differ.
    
    The first option is fixed by imposing the ordinary Lie brackets on \(L_0\frg\) and \(L_0\frg\rtimes P_0\frg\). This determines the Peiffer brackets uniquely by~\eqref{eq:h_Lie_bracket} and~\eqref{eq:l_Lie_bracket}, as \(\mu_1\colon L_0\frg\to L_0\frg \rtimes P_0\frg\) is injective. All other compatibility relations hold, and the required Peiffer bracket is
    \begin{equation}
        \{(\lambda_1,\gamma_1),(\lambda_2,\gamma_2)\}
        =
        \llbracket(\lambda_1,\gamma_1),(\lambda_2,\gamma_2)\rrbracket
        =
        \chi(\lambda_1+\gamma_1,\lambda_2+\gamma_2)~,
    \end{equation}
    leading to the 2-crossed module
    \begin{equation}
        \ainn_{\rm adj}(\frg_{\rm lp})\coloneqq (L_0 \frg \quad \overset\sft\longrightarrow \quad L_0\frg\rtimes P_0\frg \quad\overset\sft\longrightarrow\quad P_0\frg)~.
    \end{equation}
    The two 2-crossed modules of Lie algebras $\ainn(\frg_{\rm lp})$ and $\ainn_{\rm adj}(\frg_{\rm lp})$ are not isomorphic as 2-crossed modules, but their underlying complexes of Lie algebras agree, including the Lie brackets. They differ in the Peiffer lifting and the actions of $P_0\frg$ on $L_0\frg\rtimes P_0\frg$ and $L_0\frg$.

    The second option arises from setting 
    \begin{equation}
        \{-,-\}=\llbracket-,-\rrbracket=0~,
    \end{equation}
    which is possible according to~corollary~\ref{cor:trivial_peiffer_lifting}, since \(\mu_2\colon (L_0\frg \rtimes P_0\frg)^{\wedge2}\to L_0\frg\) vanishes. This case is simpler, but it changes the Lie brackets of the components considerably. Let \(\overset\circ\frg\) denote the abelian Lie algebra over the vector space \(\frg\). Then this case corresponds to the 2-crossed module of Lie algebras
    \begin{equation}
        L_0 \overset\circ\frg \quad \overset\sft\longrightarrow \quad L_0\overset\circ\frg\rtimes P_0\frg \quad\overset\sft\longrightarrow\quad P_0\frg~.
    \end{equation}
    The Lie bracket on $L_0\frg$ vanishes by~\eqref{eq:l_Lie_bracket}, and the Lie bracket on $L_0\frg\oplus P_0\frg$ also becomes “more abelian” by~\eqref{eq:h_Lie_bracket}. 
    
    While both options are mathematically consistent, the second option “forgets” the natural structure of the path and loop spaces, and deviates too far from our original $L_\infty$-algebra. More importantly, only the first option seems possible after we extend \(L_0\frg\) to \(\hat L_0\frg\) for the string 2-algebra; see section~\ref{ssec:adjusted_inner_derivation}.
    Finally, it seems very significant that for the “correct” option, the antisymmetric part of the Peiffer lifting is \emph{precisely} the map \(\chi\), required for adjusting the Weil algebra, that also appears during the lifting of 3-term \(L_\infty\)-algebras to 2-crossed modules. This fact hints at a deeper connection between \(\llbracket-,-\rrbracket\) and $\chi$.
    
    We now integrate the 2-crossed module obtained from the first option,\footnote{The second option can also be straightforwardly integrated; this produces the 2-crossed module of Lie groups \((L_0\sG\to L_0\overset\circ\frg\rtimes P_0\sG\to P_0\sG)\), where the vector space \(L_0\overset\circ\frg\) is now interpreted as an abelian Fréchet--Lie group.} which is essentially straightforward\footnote{While general 3-term $L_\infty$-algebras are very hard to integrate, there is no difficulty or obstruction to integrating 2-crossed modules of Lie algebras~\cite[Theorem~10]{Martins:2009aa}. The fact that we deal with 2-crossed modules of Fréchet–Lie algebras is not a problem, since all of the components, being path or loop algebras on the Lie algebra \(\frg\), admit obvious integrations to path or loop algebras on the Lie group \(\sG\). See also~\cite{Baez:2005sn} for more details on Fr\'echet--Lie algebras and groups. The only possible ambiguity is the usual one involving the center/fundamental group, which amounts to consistently using the same integration \(\sG\) of the Lie algebra \(\frg\).}: : we simply have to integrate the Lie algebras in each component of the crossed module. The integration of the actions is then automatically compatible. A verification of the successful integration is then the straightforward differentiation. 
    
    For example, the crossed module of Lie algebras $\frg_{\rm lp}$ integrates to the crossed module of Lie groups $\sG_{\rm lp,\rm \cm}=(L_0\sG\xrightarrow{~\sft~} P_0\sG)$, with pointwise multiplication, pointwise action of $P_0\sG$ on $L_0\sG$ and $\sft$ being the embedding. Differentiation (by applying the tangent functor) directly recovers $\frg_{\rm lp}$. Correspondingly the 2-crossed module of Lie groups resulting from the integration of $\ainn_{\rm adj}(\frg_{\rm lp})$ is
    \begin{equation}
        \sInn_{\rm adj,\cm}(\sG_{\rm lp})\coloneqq(L_0 \sG\xrightarrow{~\sft~}L_0\sG\rtimes P_0\sG\xrightarrow{~\sft~}P_0\sG)
    \end{equation}
    with the given product structure and the evident pointwise $P_0\sG$-actions. The Peiffer lifting is fixed by the relation
    \begin{equation}
     \sft(\{h_1,h_2\})=h_1 h_2 h_1^{-1}(\sft(h_1)\acton h_2^{-1})~,
    \end{equation}
    cf.~\eqref{eq:global_peiffer_lifting_identity}, because $\sft$ is injective. The fact that $\sInn_{\rm adj,\cm}(\sG_{\rm lp})$ integrates $\ainn_{\rm adj}(\frg_{\rm lp})$ follows from straightforward differentiation.
    
    Without adjustment, we would have arrived at the 2-crossed module of Lie groups $\sInn_{\rm \cm}(\sG_{\rm lp})$. The difference between the latter and $\sInn_{\rm adj,\cm}(\sG_{\rm lp})$ is seen from the difference of the corresponding 2-crossed modules of Lie algebras: While the underlying normal complexes and the products in each degree agree, the Peiffer lifting and the action of $P_0\sG$ on $L_0 \sG\rtimes P_0 \sG$ and $L_0\sG$ are different. Since $\sft\colon L_0\sG\rightarrow L_0\sG\rtimes P_0\sG$ is injective, the \(P_0\sG\)-actions fix the Peiffer lifting. At the level of the corresponding monoidal 2-categories encoding the Gray groups $\sInn_{\rm \cm}(\sG_{\rm lp})$ and $\sInn_{\rm adj,\cm}(\sG_{\rm lp})$, we thus encounter the same globular structure. Also, there is no modification to the short exact sequence~\eqref{eq:ses_2-groupoids}.

    \subsection{Adjusted inner derivations of the string Lie 2-algebra}\label{ssec:adjusted_inner_derivation}

    We now readily construct the main example: the adjusted Weil algebras of the string Lie 2-algebra $\astringl(\frg)$ defined by the differential~\eqref{eq:Weil_string_Lie_2_loop_adj}. The coordinate change~\eqref{eq:coord_change} leads to the differential graded algebra
    \begin{equation}\label{eq:Weil_string_Lie_2_loop_adj_cc}
        \begin{aligned}
            Q_{\sW_{\rm adj}}\colon~t^{\alpha\tau} &\mapsto -\tfrac12 f^{\alpha}_{\beta\gamma} t^{\beta\tau} t^{\gamma\tau} + \tilde t^{\alpha \tau}~,~~~&\tilde{t}^{\alpha\tau}&\mapsto -f^\alpha_{\beta\gamma}t^{\beta\tau}\tilde {t}^{\gamma\tau}~,\\
            r^{\alpha\tau} &\mapsto f^\alpha_{\beta\gamma}(t^{\beta\tau}\tilde t^{\gamma\tau}-\ell(\tau)t^{\beta 1}\tilde t^{\gamma 1})+ \hat{r}^{\alpha\tau}~,~
            &\hat{r}^{\alpha\tau}&\mapsto0~,\\
            r_0 &\mapsto 2\int_0^1\dd\tau\,\kappa_{\alpha\beta} \dot t^{\alpha\tau} \tilde t^{\beta\tau}+\hat{r}_0~,&\hat{r}_0&\mapsto -2\int_0^1\dd\tau\,\kappa_{\alpha\beta} \dot{\tilde t}^{\alpha\tau}\tilde t^{\beta\tau}~.
        \end{aligned}
    \end{equation}
    Dually, we have the 3-term $L_\infty$-algebras $\sW^*_{\mathrm{adj}}(\astringl(\frg))$ with cochain complex
    \begin{equation}
        \begin{matrix}
            (&
            \hat L_0 \frg & \xrightarrow{~\mu_1~} & \hat L_0\frg'\rtimes P_0\frg' & \xrightarrow{~\mu_1~} & P_0\frg &)\\
            &\lambda+r & \xmapsto{~\mu_1~} &(\lambda+r,0)\\
            &&& (\lambda+r,\gamma) & \xmapsto{~\mu_1~} &\gamma
        \end{matrix} 
    \end{equation}
    which is endowed with the binary products \(\mu_2\)
    \begin{subequations}
        \begin{align}
            P_0\frg\wedge P_0\frg&\to P_0 \frg~, &
            (\gamma_1,\gamma_2)&\mapsto[\gamma_1,\gamma_2]~, \\
            P_0\frg\wedge(\hat L_0\frg'\rtimes P_0\frg')&\to \hat L_0\frg'\rtimes P_0\frg'~, &
            \big(\gamma_1,(\lambda_2+r_2,\gamma_2)\big)&\mapsto\left(-\chi(\gamma_1,\gamma_2),[\gamma_1,\gamma_2]\right)~, \\
            P_0\frg\wedge \hat L_0\frg &\to \hat L_0\frg~, &
            (\gamma_1,\lambda_2+r_2)&\mapsto0~, \\
            (\hat L_0\frg'\rtimes P_0\frg')^{\times2}&\to \hat L_0 \frg~, &
            \!\!\!\!\!\!\!\!\!\!\!\!\!\!%\!\!\!\!\!\!\!\!\!
            \big((\lambda_1+r_1,\gamma_1),(\lambda_2+r_2,\gamma_2)\big)&\mapsto
            -\chi(\gamma_1,\gamma_2)-\chi(\gamma_2,\gamma_1)~,
        \end{align}
    \end{subequations}
    where \(\chi\) was defined in~\eqref{eq:ext-kappa-definition}.\footnote{This \(\chi\) is analogous to, but naturally different from, the \(\chi\) in~\eqref{eq:chi_loop_model_of_Lie_algebra} used in section~\ref{ssec:loop_model_of_Lie_algebra}.} The extension to a 2-crossed module of Lie algebras 
    \[
    \ainn_\text{adj}(\hat L_0\frg\to P_0\frg)=\big(\ainn_\text{adj}(\hat L_0\frg\to P_0\frg),\sft,\acton,\{-,-\}\big)
    \]
    has underlying cochain complex of Lie algebras
    \begin{equation}
        \begin{matrix}
            \ainn_\text{adj}(\hat L_0\frg\to P_0\frg) &=& (&
            \hat L_0\frg & \overset\sft\longrightarrow & \hat L_0\frg'\rtimes P_0\frg' & \overset\sft\longrightarrow & P_0\frg &)\\
            &&&\lambda+r & \xmapsto{~\sft~} & (\lambda+r,0) \\
            &&&&& (\lambda+r,\gamma) & \xmapsto{~\sft~} & \gamma
        \end{matrix} 
    \end{equation}
    with
    \begin{subequations}
        \begin{align}
            [\gamma_1,\gamma_2]_{P_0\frg}&=\mu_2(\gamma_1,\gamma_2)~,\\
            [(\lambda_1+r_1,\gamma_1),(\lambda_2+r_2,\gamma_2)]_{\hat L_0\frg\rtimes P_0\frg}&=
            \big([\lambda_1,\lambda_2]+[\gamma_1,\lambda_2]+[\lambda_1,\gamma_2]
            ,~[\gamma_1,\gamma_2]\big)\notag\\
            &=\big(\chi(\lambda_1+\gamma_1,\lambda_2+\gamma_2)-\chi(\gamma_1,\gamma_2)
            ,~[\gamma_1,\gamma_2]\big)~,\\
            [\lambda_1+r_1,\lambda_2+r_2]_{\hat L_0\frg}&=\chi(\lambda_1,\lambda_2)~,\\
            \gamma_1\acton(\lambda_2+r_2,\gamma_2)&=\mu_2\big(\gamma_1,(\lambda_2+r_2,\gamma_2)\big)~,\\
            \gamma_1\acton(\lambda_2+r_2)&=0~,\\
            \{(\lambda_1+r_1,\gamma_1),(\lambda_2+r_2,\gamma_2)\}&=\chi(\lambda_1+\gamma_1,\lambda_2+\gamma_2)~.
        \end{align}        
    \end{subequations}
    The Peiffer bracket is again precisely the function $\chi$ encoding the adjustment of the Weil algebra. Unlike the case of \(\frg_\text{lp}\) in section~\ref{ssec:loop_model_of_Lie_algebra}, here $\chi$ (and thus the Peiffer lifting $\{-,-\}$) is no longer purely antisymmetric, due to a boundary term. The symmetric part of the Peiffer bracket corresponds to the non-vanishing higher product \(\mu_2\colon(\hat L_0\frg\rtimes P_0\frg)^{\times2}\to\hat L\frg\). The antisymmetric part of the Peiffer bracket is the additional structure map $\llbracket-,-\rrbracket$ of the 2-crossed module of Lie algebras.
    
    Integrating $\ainn_{\rm adj}(\astring_{\rm lp}(\frg))$, we arrive at the 2-crossed module of Lie groups 
    \begin{equation}
        \begin{matrix}
            \sInn_{\text{adj},\cm}(\sString_{\rm lp}(\sG)) &=& (&
            \hat L_0 \sG & \overset\sft\longrightarrow & \hat L_0\sG'\rtimes P_0\sG' & \overset\sft\longrightarrow & P_0\sG  &)\\
            &&&(l,r) & \xmapsto{~\sft~} &\big((l,r),\unit_{P_0\sG}\big)\\
            &&&&& ((l,r),p) & \xmapsto{~\sft~} & p
        \end{matrix} 
    \end{equation}
    The $P_0\sG$-actions on the bases $L_0\sG$ and $L_0\sG\rtimes P_0\sG$ of the principal \(\sU(1)\)-bundles $\hat L_0 \sG$ and $\hat L_0\sG\rtimes P_0\sG$ are the same as in $\sInn_{\rm adj,\cm}(\sG_{\rm lp})$. The $P_0\sG$-action on the $\sU(1)$-fibers are the canonical ones as in the loop model of the string Lie 2-group. Explicit expressions are best constructed indirectly, after trivializing the circle bundles;\footnote{that is, the technique of constructing a nontrivial \(\sU(1)\)-bundle on a Fréchet–Lie group \(\sG\) as a quotient of the trivial \(\sU(1)\)-bundle on the path group \(P_0\sG\)} see~\cite{Baez:2005sn}.
    The Peiffer lifting is fixed by~\eqref{eq:global_peiffer_lifting_identity}:
    \begin{multline}
        \Big\{\big((l_1,r_1),p_1\big),\big((l_2,r_2),p_2\big)\Big\}=\\
        \big((l_1,r_1),p_1\big)\big((l_2,r_2),p_2\big)
        \big((l_1,r_1),p_1\big)^{-1}
        \left(p_1\acton\big((p_2,r_2),p_2\big)^{-1}\right)~,
    \end{multline}
    where all products are taken in the semidirect product \(\hat L_0\sG'\rtimes P_0\sG'\).
    
    The 3-group constructed from $\sInn_{\text{adj},\cm}(\sString_{\rm lp}(\sG))$ as in~\eqref{eq:mon_cat_from_3_group} is then 
    \begin{subequations}
    \begin{equation}
        \sInn_\text{adj}(\sString_{\rm lp}(\sG))=\left( 
            \hat L_0\sG\rtimes\big((\hat L_0'\sG\rtimes P_0\sG')\rtimes P_0\sG\big)   \rightrightarrows   (\hat L_0\sG'\rtimes P_0\sG')\rtimes P_0\sG)  \rightrightarrows P_0\sG \right)
    \end{equation}
    with the following globular structure.
    \begin{equation}
        \begin{tikzcd}[column sep=4.0cm,row sep=large]
            p_1p_2 & \ar[l, bend left=45, "{\big(\ell_2,p_1,p_2\big)}", ""{name=U,inner sep=1pt,above}] \ar[l, bend right=45, "{\big(\ell_1\ell_2,p_1,p_2\big)}", swap, ""{name=D,inner sep=1pt,below}] p_2
            \arrow[Rightarrow,from=U, to=D, "{\big(\ell_1,\ell_2,p_1,p_2\big)}",swap]
        \end{tikzcd}
    \end{equation}
    \end{subequations}
    
    We have a short exact sequence of 2-groupoids,
    \begin{equation}\label{eq:ses_groupoids_adjusted}
        *\xrightarrow{~~~}\CCG \xrightarrow{~\Upsilon~} \sInn_{\rm adj}(\CCG) \xrightarrow{~\Pi~}\sB\CCG \xrightarrow{~~~} *~,
    \end{equation}
    where the functors $\Upsilon$ and $\Pi$ are again given by the obvious embedding and projection functors. This is the adjusted analogue of~\eqref{eq:ses_2-groupoids}. As complexes of globular sets, this complex is identical to that in~\eqref{eq:ses_2-groupoids};\footnote{Of course, as 2-functors between 2-groupoids, the 2-functors \(\Upsilon\) and \(\Pi\) in~\eqref{eq:ses_groupoids_adjusted} are different from the 2-functors \(\Upsilon\) and \(\Pi\) in~\eqref{eq:ses_2-groupoids}, simply because the (co)domains are inequivalent 2-groupoids. As we are not much concerned with 2-groupoids beyond their globular structure, however, we abuse notation and do not notate the two pairs differently.} in particular the presentation~\eqref{eq:3-group_short_exact_sequence_components}, as well as the presentation~\eqref{eq:ses_2-groupoids_reparametrization} with the reparametrization~\eqref{eq:inner_algebra_alternative_basis}, continue to be valid after adjustment.
    
    \section{Parallel transport}\label{sec:parallel_transport}
    
    We now discuss the main topic of this paper: the consistent definition of a higher, truly non-abelian parallel transport. The key features are already visible over local patches, and gluing the construction to a global one is, in principle, a mere technicality; see~e.g.~\cite{Wang:1512.08680}. For clarity of our discussion, we always work on local patches or, equivalently, a contractible manifold $U$.

    \subsection{Ordinary parallel transport and connections}\label{ssec:parallel}
    
    The fact that the holonomies around all smooth loops encode a connection has been known in the literature since at least the 1950s~\cite{Kobayashi:1954aa}. The picture we use is inspired by the treatment of loops in~\cite{Barrett:1991aj} (see also~\cite{Gambini:1980yz}), and generalized to paths in~\cite{Caetano:1993zf} (see also~\cite{Schreiber:0705.0452}).
    
    Let \(\sG\) be a Lie group. Parallel transport encoded in a connection on a principal $\sG$-bundle $P$ over the contractible manifold $U$ amounts to an assignment of a group element $g\in \sG$ to each path $\gamma\colon[0,1]\rightarrow U$ in the base manifold. Composition of paths translates to multiplication of the corresponding group elements.
    The paths and points of the base manifold naturally combine to the path groupoid $\CP U$. This is the category which has paths as morphisms, their endpoints being the domains and the codomains. Since we can invert paths by reversing their orientation, $\CP U$ is indeed a groupoid.\footnote{For the detailed definition of $\CP U$, see appendix~\ref{app:path_space}.} Regarding $\sG$ as the one-object groupoid $\sB\sG=(\sG\rightrightarrows *)$, we see that parallel transport is precisely a functor
    \begin{equation}
        \Phi\colon \CP U\longrightarrow \sB\sG~,\hspace{2cm}
        \begin{tikzcd}
            \mbox{paths} \arrow[d,shift left] \arrow[d,shift right]\arrow[r]{}{\Phi_1} & \sG \arrow[d,shift left] \arrow[d,shift right]\\
            U \arrow[r]{}{\Phi_0} & * 
        \end{tikzcd}
    \end{equation}
    
    Given a connection in terms of a $\frg\coloneqq \sLie(\sG)$-valued 1-form $A$ on $U$, we can construct the parallel transport functor as
    \begin{equation}
        \Phi_1(\gamma)=\Pexp\int_\gamma A\in \sG~,
    \end{equation}
    where $\Pexp(\dots)$ is the path-ordered exponential well-known in physics. Mathematically, $\Phi_1(\gamma)=g(1)$, where $g$ is the (unique) solution $g(t)$ to the differential equation\footnote{given here for clarity for matrix Lie groups, the abstract analogue being evident}
    \begin{equation}\label{eq:path_ordered_ode}
        \left(\frac{\dd}{\dd t} g(t)\right)g(t)^{-1}=-\iota_{\dot \gamma(t)}A(\gamma(t)) ~,~~~g(0)=\unit_\sG~,
    \end{equation}
    where $\iota_{\dot \gamma(t)}$ denotes the contraction with the tangent vector to $\gamma$ at $\gamma(t)$. Conversely, given a functor $\Phi$, the corresponding connection \(A\in\Omega^1(U)\otimes\frg\) is obtained as follows. Let $x\in U$ be a point and $v\in T_xU$ a tangent vector at $x$. We choose a path $\gamma[0,1]\rightarrow U$ such that $\gamma(\tfrac12)=x$ and $\dot \gamma(\tfrac12)=v$. Just as any general path, $\gamma$ gives rise to a function $g(t)=\Phi(\gamma_t)\colon[0,1]\rightarrow \sG$, where $\gamma_t$ is the truncation of $\gamma$ at $\gamma(t)$ (with appropriate reparametrization). We can then use equation~\eqref{eq:path_ordered_ode} and define
    \begin{equation}
        -\iota_v A(x)=\left.\left(\frac{\dd}{\dd t} g(t)\right)g(t)^{-1}\right|_{t=\tfrac12}~,
    \end{equation}
    where $A$ is independent of the choice of $\gamma$ and the reparametrization in the truncation of $\gamma$ to $\gamma_t$. Thus the parallel transport functor \(\Phi\) contains exactly the same information as the connection \(A\).    
    
    Since connections correspond to functors, it is rather obvious that gauge transformations correspond to natural transformations.\footnote{In general, for functors between general categories, one distinguishes between \emph{natural transformations} and \emph{natural isomorphisms}, where the latter is a natural transformation whose components are all isomorphisms. For functors between groupoids, as in our case, all natural transformations are natural isomorphisms.}\footnote{Gauge transformations can be thought of in two different but equivalent perspectives: the “physicist’s”, where the gauge fields and Wilson lines are objects defined on the manifold that are acted upon by gauge transformations; and the “mathematician’s”, where the gauge fields are invariant objects defined on the total space of principal bundles, where the apparent gauge transformations correspond to different local trivializations of the pincipal bundle. In this paper, we work with an already locally trivialized bundle, so that the formulae appear as actions of the gauge transformations; but they can be equally well interpreted as the result of the gauge transformations’ changing the choice of local trivialization.} A natural transformation $\eta\colon\Phi\Rightarrow\tilde \Phi$ between two functors of Lie groupoids $\Phi,~\tilde\Phi\colon\CP U\to \sB\sG$ is encoded in a function $\eta\colon U\to \sG$ such that 
    \begin{equation}
        \tilde \Phi_1(\gamma)=\eta(\gamma(1))^{-1}\Phi_1(\gamma)\eta(\gamma(0))
    \end{equation}
    for each path $\gamma$. This is precisely the gauge transformation law for a Wilson line. Let $A$ and $\tilde A$ be the connection 1-forms associated with $\Phi$ and $\tilde \Phi$, respectively. The functions $g(t)$ and $\tilde g(t)$ appearing in equation~\eqref{eq:path_ordered_ode}, are related by
    \begin{equation}
        \tilde g(t)=\eta(\gamma(t))^{-1} g(t)\eta(\gamma(0))~,
    \end{equation}
    and equation~\eqref{eq:path_ordered_ode} for $\tilde A$ induces then the usual gauge transformations,
    \begin{equation}
        \tilde A(x)=\eta(x)^{-1} A(x) \eta(x)+\eta(x)^{-1} \dd \eta(x)~.
    \end{equation}
    
    Altogether, the parallel transport functor is kinematically omniscient: it contains all information about gauge configurations and gauge orbits.

    \subsection{Ordinary parallel transport and the derivative parallel transport functor}\label{ssec:derived_parallel_ordinary}
    
    To see the curvature 2-form $F=\dd A+\tfrac12[A,A]$ of $A$ arise from parallel transport, we trivially extend $\Phi$ to a strict 2-functor $\varPhi$ as follows. First, we extend the path groupoid $\CP U$ to a path 2-groupoid $\CP_{(2)} U$, whose objects are the points of $U$, whose 1-morphisms are the paths, and whose 2-morphisms between two paths $\gamma_1,\gamma_2\colon x\to y$ are {\em bigons}, i.e.~surfaces bounded by $\gamma_1\circ \gamma^{-1}_2$.\footnote{For technical details, see appendix~\ref{app:path_space}.}
    Similarly, we extend $\sB\sG$ to 
    \begin{equation}
        \sB\sInn(\sG) = (\sG \rtimes \sG \rightrightarrows \sG \rightrightarrows *)~,
    \end{equation}
    which is a 2-groupoid with one object $*$, over which we have the morphism 2-group $\sInn(\sG)$. As explained in section~\ref{ssec:inner_group}, this is the action groupoid for the action of $\sG$ onto itself by left-multiplication with morphisms
    \begin{equation}
        g_1g_2 \xleftarrow{~~(g_1,g_2)~~} g_2~.
    \end{equation}
    Then we can construct the \emph{derivative parallel transport 2-functor}\footnote{not to be confused with the unrelated concept of derived functors in homological algebra}~\cite{Schreiber:0802.0663}, which is a strict 2-functor
    \begin{equation}
        \varPhi\colon \CP_{(2)} U\longrightarrow \sB\sInn(\sG)~,\hspace{2cm}
        \begin{tikzcd}
            \text{surfaces} \rar["\varPhi_2"] \dar[shift left] \dar[shift right] & \sG \rtimes \sG \dar[shift left] \dar[shift right] \\
            \text{paths} \arrow[d,shift left] \arrow[d,shift right]\arrow[r]{}{\varPhi_1} & \sG \arrow[d,shift left] \arrow[d,shift right]\\
            U \arrow[r]{}{\varPhi_0} & * 
        \end{tikzcd}
    \end{equation}
    It assigns to each path $\gamma$ an element $\varPhi_1(\gamma)=g_\gamma$ in $\sG$ and to each surface $\sigma$ an element $\varPhi(\sigma)=(g^1_\sigma,g^2_\sigma)$ in $\sG\rtimes \sG$, as follows: 
    \begin{equation}
        \begin{tikzcd}[column sep=2.5cm,row sep=large]
            x_2 & \ar[l, bend left=45, "\gamma_1", ""{name=U,inner sep=1pt,above}] \ar[l, bend right=45, "\gamma_2", swap, ""{name=D,inner sep=1pt,below}] x_1
            \arrow[Rightarrow,from=U, to=D, "\sigma",swap]
        \end{tikzcd}~~~\overset\varPhi\longmapsto~~~
        \begin{tikzcd}[column sep=2.5cm,row sep=large]
            * & \ar[l, bend left=45, "g_{\gamma_1}", ""{name=U,inner sep=1pt,above}] \ar[l, bend right=45, "g_{\gamma_2}", swap, ""{name=D,inner sep=1pt,below}] *
            \arrow[Rightarrow,from=U, to=D, "{(g^1_\sigma,g^2_\sigma)}",swap]
        \end{tikzcd}
    \end{equation}
    Compatibility with the domain and codomain maps \(\dom,\codom\) implies that
    \begin{equation}
        g_{\gamma_1}=\dom(\varPhi(\sigma))=g^2_\sigma\eand g_{\gamma_2}=\codom(\varPhi(\sigma))=g^1_\sigma g^2_\sigma~.
    \end{equation}
    Thus $\varPhi(\sigma)$ is fully fixed by the $g_\gamma$, and the strict 2-functor $\varPhi$ is determined by the \mbox{(1-)functor} $\Phi$.
    
    At an infinitesimal level, the additional data for surfaces encodes the curvature, and $\varPhi$ being determined by $\Phi$ amounts to a non-abelian version of Stokes' theorem~\cite[Section~3.2]{Schreiber:0802.0663}. In terms of the component fields \(A\in\Omega^1(U)\otimes\frg\) and its curvature \(F\in\Omega^2(U)\otimes\frg=\dd A+\tfrac12[A,A]\), we can write $g_\gamma$ and $g_\sigma$ as 
    \begin{equation}
        g_\gamma = \Pexp\int_\gamma A\eand g_\sigma = \Pexp\ichint_{A}(-F)~.
    \end{equation}
    The additional minus sign in front of the curvature $F$ is explained in appendix~\ref{app:inner_weil}. The second integral is a path-ordered integral over a path in path space, and $\chint_A(-F)$ is a \emph{Chen form} as described in appendix~\ref{app:chen_forms}. Briefly, we view \(\sigma\) as a path \(\check\sigma\) on the space of paths between two points on the boundary of $\sigma$, \(x_0\) and \(x_1\), and the 2-form \(F\) as a 1-form \(-\check F = \chint_A(-F)\) on the space of paths between \(x_0\) and \(x_1\). Then
    \[
    \Pexp\ichint_{A} (-F) \coloneqq \Pexp\int_{\check\sigma}(-\check F)~.
    \]
    This is now of course equivalent to a differential equation on the path space.
    
    Given a bigon \(\sigma\colon\gamma_1\to\gamma_2\), since \(\dpar\sigma = \gamma_1\cup\bar\gamma_2\), the globular identity
    \begin{equation}
        g_{\gamma_1}g_{\gamma_2}^{-1} = g_\sigma^{-1}
    \end{equation}
    becomes
    \begin{equation}
        \Pexp\oint_{\dpar\sigma}A=\Pexp\ichint_{A} F~,
    \end{equation}
    where \(F=\dd A+\tfrac12[A,A]\) is the ordinary curvature and
    \begin{equation}
        g_\sigma = \Pexp\ichint_{A}(-F)~.
    \end{equation}
    
    Conversely, we can recover the fields and curvatures from the derivative parallel transport 2-functor \(\varPhi\). We have already explained how to recover \(A\) as above. As for \(F\), since \(\varPhi\) assigns elements of \(\sG\) to parametrized surfaces \(\sigma\), i.e.~paths \(\check\sigma\) in the space of paths between \(x_1\) and \(x_2\), we can do the same procedure as for \(A\) to recover the corresponding 1-form \(\check F\) on path space, and translate it to a 2-form \(F\) on \(U\).
    
    We now discuss gauge transformations. Just as for the plain parallel transport functor, we should identify gauge transformations with natural transformations. The general notion of 2-natural transformations between 2-functors between 2-groupoids is that of \emph{pseudonatural transformations}.\footnote{For 2-natural transformations between 2-functors between general 2-categories, one distinguishes between \emph{lax natural 2-transformations}, whose component 2-cells need not be invertible, and \emph{weak natural 2-transformations} or {\em pseudonatural transformations}, whose component 2-morphisms must be invertible (but not necessarily trivial) by definition; see e.g.~\cite{Jurco:2014mva}. However, for 2-groupoids, all 2-morphisms are invertible, and the two classes coincide.} A pseudonatural transformation $\eta\colon\varPhi\rightarrow \tilde \varPhi$ between two strict 2-functors $\varPhi,~\tilde \varPhi\colon \CP_{(2)}U\rightarrow \sB\sInn(\sG)$ is encoded in maps
    \begin{equation}
        \eta_1\colon \CP_{(2)}U_0=U\mapsto \sG\eand \eta_2\colon \CP_{(2)}U_1\mapsto \sG\rtimes \sG~,
    \end{equation}
    where $\CP_{(2)}U_1$ are the paths or 1-morphisms in $\CP_{(2)}U$, such that for each path $x_1\xleftarrow{~~\gamma~~}x_0$, we have the commuting diagram
    \begin{equation}
        \begin{tikzcd}[column sep=2.5cm,row sep=large]
            * \ar[d,swap]{}{\eta_1(x_1)} & * \ar[l,swap]{}{\varPhi(\gamma)}\ar[d]{}{\eta_1(x_0)}\\
            * & \ar[l]{}{\tilde\varPhi(\gamma)}* \ar[Rightarrow,ul,swap]{}{\eta_2(\gamma)}
        \end{tikzcd}
    \end{equation}
    implying that 
    \begin{equation}
        \eta_2(\gamma)=\big(\eta_1(x_1)\varPhi(\gamma)\eta_1^{-1}(x_0)\tilde \varPhi^{-1}(\gamma) ,\tilde \varPhi(\gamma)\eta_1(x_0)\big)\in \sG\rtimes \sG~.
    \end{equation}
    The coherence axioms for a pseudonatural transformation are then automatically satisfied. 
    
    The additional freedom in the gauge transformations allows for a pseudonatural transformation $\eta$ between any strict 2-functor $\varPhi$ and the trivial strict 2-functor $\unit$
    \begin{equation}
        \unit(x)=*~,~~~\unit(\gamma)=\unit_\sG~,~~~\unit(\sigma)=(\unit_\sG,\unit_\sG)
    \end{equation}
    for all $x\in \CP_{(2)}U_0$, $\gamma\in \CP_{(2)}U_1$ and $\gamma\in \CP_{(2)}U_2$. Explicitly, $\eta$ is given by 
    \begin{equation}
        \eta_1(x)=\unit_\sG\eand\eta_2(\gamma)=(\varPhi(\gamma),\unit_\sG)~.
    \end{equation}
    This transformation reflects the fact that $\sInn(\sG)$ is equivalent to the trivial 2-group and that $\sB\sInn(\sG)$ is equivalent to the trivial 3-groupoid.
    
    We thus need to restrict the allowed gauge transformations in an obvious way. The short exact sequence~\eqref{ses:groupoids_Lie} leads to the following commutative diagram:
    \begin{equation}
        \begin{tikzcd}
            \CP U\ar[d,"\Phi",swap] \ar[r,hookrightarrow] & \CP_{(2)}U  \ar[d,"\varPhi"] \\
            \sB \sG \ar[r,hookrightarrow] & \sB \sInn(\sG)
        \end{tikzcd}
    \end{equation}    Furthermore, if we fix endpoints \(x_0,x_1\in U\), we can decategorify\footnote{in the sense of taking hom-categories, thus shifting 1-morphisms to be objects and 2-morphisms to be 1-morphisms} the above diagram, and add a new functor \(\varPhi_\mathrm{curv}(x_0,x_1)\), which is a truncation of \(\varPhi(x_0,x_1)\) to surfaces only:
    \begin{equation}
        \begin{tikzcd}
            \CP U(x_0,x_1)\ar[d,"{\Phi(x_0,x_1)}",swap] \ar[r,hookrightarrow] & \CP_{(2)}U(x_0,x_1)  \ar[d,"{\varPhi(x_0,x_1)}"']  \ar[dr,"{\varPhi_\mathrm{curv}(x_0,x_1)}"] \\
            \sG \ar[r,hookrightarrow] & \sInn(\sG) \ar[r,twoheadrightarrow,"\Pi"] & \sB\sG
        \end{tikzcd}
    \end{equation}
    Here, $\sG$ is regarded as the discrete category $\sG\rightrightarrows\sG$. Note that the decategorification is necessary because \(\sB\sB\sG\) does not make sense as a 2-category in general: there is no compatible monoidal product for non-abelian $\sG$ due to the Eckmann--Hilton argument. The functor \(\varPhi_\text{curv}(x_0,x_1)\) therefore does {\em not} extend to a 2-functor \(\varPhi\colon\CP_{(2)}U\to\sB\sB\sG\): 
    
    Clearly, we are only interested in transformations of $\varPhi$ that originate from transformations of $\Phi$ and which become trivial\footnote{This does {\em not} imply that the curvatures do not transform under gauge transformations.} on $\varPhi_{\rm curv}$. That is, for any two derivative parallel transport 2-functors $\varPhi,\tilde \varPhi$ connected by such a transformation, we have
    \begin{equation}
        \varPhi_{\rm curv}(x_0,x_1)=\Pi\circ \varPhi(x_0,x_1)=\Pi\circ \tilde \varPhi(x_0,x_1)=\tilde \varPhi_{\rm curv}(x_0,x_1)~.
    \end{equation}
    Equivalently, these natural transformations are rendered trivial by the whiskering\footnote{{\em Whiskering} is the horizontal composition of a trivial 2-morphism, here $\id_\Pi\colon\Pi\Rightarrow \Pi$ in the higher category of 2-groupoids, 2-functors and 2-natural isomorphisms, with another 2-morphism, here $\eta$.}
    \begin{equation}
        \begin{tikzcd}
            \sB\sG & \sInn(\sG) \lar["\Pi",swap]  &  \CP_{(2)} U(x_0,x_1) \lar[bend left=50, "{\varPhi(x_0,x_1)}", ""'{name=A}] \lar[bend right=50, "{\tilde \varPhi(x_0,x_1)}"', ""{name=B}] \ar[from=A, to=B, Rightarrow,"\eta"]
        \end{tikzcd}
    \end{equation}
    This is simply achieved by demanding that $\eta_2$ be trivial:
    \begin{equation}
        \eta_2(\gamma)=\big(\unit,\tilde \Phi(\gamma)\eta(x_0)\big)~.
    \end{equation}
    Such natural transformations are known as {\em strict 2-natural transformations}.

    \subsection{Unadjusted higher parallel transport and connections}\label{ssec:derived_parallel_higher}
    
    Higher-dimensional generalizations of parallel transport have been studied since the 1990s. First discussions for higher principal bundles are found in~\cite{0817647309}; appropriate higher path spaces where discussed in~\cite{caetano1998family}. The higher-dimensional parallel transport for abelian higher principal bundles was then fully developed in~\cite{Gajer:1997:155-207,Gajer:1999:195-235,Mackaay:2000ac}. The non-abelian extension was discussed in~\cite{Chepelev:2001mg},~\cite{Baez:2002jn}~\cite{Girelli:2003ev},~\cite{Baez:2004in} and further, in great detail, in the papers~\cite{Schreiber:0705.0452,Schreiber:0802.0663,Schreiber:2008aa}; see also~\cite{Alvarez:1997ma} for earlier considerations and~\cite{Li:2019nqu} for a recent discussion. We also need the structures underlying the higher parallel transport along volumes, discussed in~\cite{Martins:2009aa}.
    
    Let $\CCG$ be a strict Lie 2-group with underlying monoidal category $(\sH\rtimes \sG\rightrightarrows \sG)$ with morphisms 
    \begin{equation}
        \sft(h)g\xleftarrow{~~(h,g)~~}g~;
    \end{equation}
    the corresponding crossed module of Lie groups is $\CCG_{\rm \cm}=(\sH\xrightarrow{~\sft~}\sG,\acton)$, cf.~appendix~\ref{app:hypercrossed_modules}. Parallel transport over a local patch $U$ with gauge 2-group $\CCG$ is then described by strict 2-functors from the path 2-groupoid $\CP_{(2)} U$ to $\sB\CCG$,
    \begin{equation}
        \Phi\colon \CP_{(2)}U \to \sB\CCG~,\hspace{2cm}
        \begin{tikzcd}
            \text{surfaces} \rar["\Phi_2"] \dar[shift left] \dar[shift right] & \sH\rtimes\sG \dar[shift left] \dar[shift right] \\
            \text{paths} \rar["\Phi_1"] \dar[shift left] \dar[shift right] & \sG \dar[shift left] \dar[shift right] \\
            U \rar["\Phi_0"] & *
        \end{tikzcd}
    \end{equation}
    which assign to each path $\gamma$ a group element $g_\gamma\in \sG$ and to each surface $\sigma$ a group element $\Phi(\sigma)=(h_\sigma,g_\sigma)\in \sH\rtimes \sG$:
    \begin{equation}
        \begin{tikzcd}[column sep=2.5cm,row sep=large]
            x_2 & \ar[l, bend left=45, "\gamma_1", ""{name=U,inner sep=1pt,above}] \ar[l, bend right=45, "\gamma_2", swap, ""{name=D,inner sep=1pt,below}] x_1
            \arrow[Rightarrow,from=U, to=D, "\sigma",swap]
        \end{tikzcd}~~~\overset\Phi\longmapsto~~~
        \begin{tikzcd}[column sep=2.5cm,row sep=large]
            * & \ar[l, bend left=45, "g_{\gamma_1}", ""{name=U,inner sep=1pt,above}] \ar[l, bend right=45, "g_{\gamma_2}", swap, ""{name=D,inner sep=1pt,below}] *
            \arrow[Rightarrow,from=U, to=D, "{(h_\sigma,g_\sigma)}",swap]
        \end{tikzcd}
    \end{equation}
    Compatibility with domain and codomain maps in the morphism categories amounts to
    \begin{equation}
        g_{\gamma_1}=g_\sigma\eand g_{\gamma_2}=\sft(h_\sigma)g_\sigma~.
    \end{equation}
    Let \(\frg\) and \(\frh\) be the Lie algebras of \(\sG\) and \(\sH\), respectively. Then, the kinematical data consists of fields
    \begin{equation}
        A\in\Omega^1(U)\otimes\frg\eand B\in\Omega^2(U)\otimes\frh
    \end{equation}
    and their relation to the parallel transport functor is given by
    \begin{equation}
        g_\gamma=\Pexp\int_\gamma A\eand
        h_\sigma = \Pexp\ichint_A B~,
    \end{equation}
    where $\check B=\chint_A B$ is again a Chen form; see~appendix~\ref{app:chen_forms}. Part of the data defining this Chen form is the \(P_0\sG\)-representation of \(B\) (which is part of the data of the crossed module \(\hat L_0\sG\xrightarrow\sft P_0\sG\)) as well as the \(P_0\sG\)-connection \(A\).
    
    In general, the globular structure of the codomain of the 2-functor (in this case, the crossed module of Lie groups $\CCG$) translate to (possibly non-abelian) Stokes’ theorems on the curvatures. In this case, the globular structure requires that the condition known as \emph{fake flatness} holds, namely
    \begin{equation}
        F = \mathrm dA+\tfrac12[A,A]+\mu_1(B) = 0~. 
    \end{equation}
    
    To derive this, one needs some technical setup. The crux of the argument, however, is simple to describe. The identity
    \begin{equation}
        g_{\gamma_1}g_{\gamma_2}^{-1} = \sft(h_\sigma^{-1})
    \end{equation}
    for $\dpar\sigma=\gamma_1\cup\bar\gamma_2$ translates to
    \begin{equation}
        \Pexp\int_{\dpar \sigma}A = \sft\left(\Pexp\ichint_A B\right).
    \end{equation}
    By the non-abelian Stokes' theorem,
    \begin{equation}
        \Pexp\int_{\dpar\sigma}A = \Pexp\ichint_A(\mathrm dA+\tfrac12[A,A])~.
    \end{equation} 
    Since our closed surface was arbitrary, we get
    \begin{equation}\label{eq:fake_flatness}
        \mathrm dA+\tfrac12[A,A] = -\sft(B)~,
    \end{equation}
    and thus \(F \coloneqq \mathrm dA+\tfrac12[A,A] + \sft(B) = 0\). This sketch can be made rigorous~\cite{Baez:2004in} (see also~\cite{Schreiber:2008aa}) using Chen forms; see appendix~\ref{app:chen_forms} for details.
    
    In other words, the globular structure of the crossed module means that the parallel transport 2-functor induces a Stokes’ theorem that, unfortunately, renders all physical theories based on it essentially abelian, as reviewed in section~\ref{ssec:limitations}.
    
    Gauge transformations between two strict 2-functors $\Phi,~\tilde \Phi\colon \CP_{(2)}U\rightarrow \sB\CCG$ are again given by appropriate natural transformations, which are here the general pseudonatural transformations $\eta\colon\Phi\rightarrow \tilde \Phi$. These are encoded in maps
    \begin{equation}
        \eta_1\colon \CP_{(2)}U_0=U\mapsto \sG\eand \eta_2=(\eta_2^1,\eta_2^2)\colon \CP_{(2)}U_1\mapsto \sH\rtimes \sG~,
    \end{equation}
    where $\CP_{(2)}U_1$ are the paths or 1-morphisms in $\CP_{(2)}U$, such that for each path $x_1\xleftarrow{~~\gamma~~}x_0$, we have the commutative diagram
    \begin{equation}
        \begin{tikzcd}[column sep=2.5cm,row sep=large]
            * \ar[d,swap]{}{\eta_1(x_1)} & * \ar[l,swap]{}{\Phi(\gamma)}\ar[d]{}{\eta_1(x_0)}\\
            * & \ar[l]{}{\tilde\Phi(\gamma)}* \ar[Rightarrow,ul,swap]{}{\eta_2(\gamma)}
        \end{tikzcd}~~~\Rightarrow\hspace{1.5cm}
        \begin{aligned}
            \eta_2^2(\gamma)&=\tilde \Phi(\gamma)\eta_1(x_0)~,\\
            \sft(\eta_2^1(\gamma))\tilde \Phi(\gamma)\eta_1(x_0)&=\eta_1(x_1)\Phi(\gamma)~.
        \end{aligned}
    \end{equation}
    
    We also have higher-order natural transformations (sometimes called \emph{modifications}) between the pseudonatural transformations \(\eta,\tilde\eta\colon\Phi\Rightarrow\tilde\Phi\); these correspond to the fact that the gauge parameters themselves gauge-transform. 
    
    \subsection{Unadjusted higher derivative parallel transport}    

    To make the curvatures visible, we can again categorify once more and consider a strict 3-functor $\varPhi$ from the path 3-groupoid $\CP_{(3)} U$ to $\sB\sInn(\CCG)$. The path 3-groupoid $\CP_{(3)} U$ is the evident extension of the path 2-groupoid $\CP_{(2)}U$ by adding 3-morphisms consisting of 3-dimensional homotopies between pairs of bigons; for details see appendix~\ref{app:path_space}. The 3-groupoid $\sB\sInn(\CCG)$ has one object and its morphism 2-category is $\sInn(\CCG)$, as defined in section~\ref{ssec:inner_derivations}.
    \begin{equation}
        \varPhi\colon \CP_{(3)}U \to \sB\sInn(\CCG)~,\qquad
        \begin{tikzcd}
            \text{volumes} \rar["\varPhi_3"] \dar[shift left] \dar[shift right] & \sH \rtimes \big((\sH \rtimes \sG) \rtimes \sG\big) \dar[shift left] \dar[shift right] \\
            \text{surfaces} \rar["\varPhi_2"] \dar[shift left] \dar[shift right] & (\sH \rtimes \sG) \rtimes \sG \dar[shift left] \dar[shift right]\\
            \text{paths} \rar["\varPhi_1"] \dar[shift left] \dar[shift right] & \sG \dar[shift left] \dar[shift right]\\
            U \rar["\varPhi_0"] & *
        \end{tikzcd}
    \end{equation}
    
    Explicitly, the strict 3-functor $\varPhi$ therefore amounts to assignments
    \tikzset{Rightarrow/.style={double equal sign distance,>={Implies},->},triple/.style={-,preaction={draw,Rightarrow}},quad/.style={preaction={draw,Rightarrow,shorten >=0pt},shorten >=1pt,-,double,double distance=0.2pt}}
    \begin{equation}
        \begin{tikzcd}[column sep=3cm,row sep=large]
            x_2 & \ar[l, bend left=60, "\gamma_1", ""{name=U,inner sep=1pt,above}] \ar[l, bend right=60, "\gamma_2", swap, ""{name=D,inner sep=1pt,below}] x_1
            \arrow[Rightarrow, bend left=60,from=U, to=D, "\sigma_2",""{name=L,inner sep=1pt,right}]
            \arrow[Rightarrow, bend right=60,from=U, to=D, "\sigma_1",""{name=R,inner sep=1pt,left},swap]
            \arrow[triple,from=R,to=L,"\rho",swap]
        \end{tikzcd}~~~\overset{\varPhi}\longmapsto~~~
        \begin{tikzcd}[column sep=5cm,row sep=large]
            * & \ar[l, bend left=60, "\varPhi(\gamma_1)", ""{name=U,inner sep=1pt,above}] \ar[l, bend right=60, "\varPhi(\gamma_2)", swap, ""{name=D,inner sep=1pt,below}] *
            \arrow[Rightarrow, bend left=70,from=U, to=D, "{\varPhi(\sigma_2)}",""{name=L,inner sep=1pt,right}]
            \arrow[Rightarrow, bend right=70,from=U, to=D, "{\varPhi(\sigma_1)}",""{name=R,inner sep=1pt,left},swap]
            \arrow[triple,from=R,to=L,"{\varPhi(\rho)}",swap]
        \end{tikzcd}
    \end{equation}
    where, using the reparametrization introduced in section~\ref{ssec:simplification},
    \begin{equation}
        \begin{gathered}
            \varPhi(\gamma)=g_\gamma\in \sG~,~~~\varPhi(\sigma)=(h_\sigma,g^1_\sigma,g^2_\sigma)\in (\sH'\rtimes \sG')\rtimes \sG~,\\
            \varPhi(\rho)=(h^1_\rho,h^2_\rho,g^1_\rho,g^2_\rho)\in \sH\rtimes \big((\sH' \rtimes\sG')\rtimes \sG\big)
        \end{gathered}
    \end{equation}
    with
    \begin{subequations}
        \begin{align}
            g_{\gamma_1}&=g^2_{\sigma_1}=g^2_{\sigma_2}~,\\
            g_{\gamma_2}&=g^1_{\sigma_1}g_{\gamma_1}=g^1_{\sigma_2}g_{\gamma_1}~,\\
            \big(h_{\sigma_1},g^1_{\sigma_1},g^2_{\sigma_1}\big)&=\big(h^2_\rho,g^1_\rho,g^2_\rho\big)~,\\
            \big(h_{\sigma_2},g^1_{\sigma_2},g^2_{\sigma_2}\big)&=\big(h^1_\rho h^2_\rho,g^1_\rho,g^2_\rho\big)~.
        \end{align}
    \end{subequations}
    Now, $h_\sigma$ fixes $h^1_\rho$ and $h^2_\rho$, and $g_\gamma$ fixes \(g^1_\sigma\) and $g^2_\sigma$, which in turn fix \(g^1_\rho\) and \(g^2_\rho\). Altogether, the strict 3-functor $\varPhi\colon\CP_{(3)}U\to \sB\sInn(\CCG)$ is fully determined by the strict 2-functor $\Phi\colon\CP_{(2)}U\to \sB\CCG$.
    
    In terms of the gauge potential and curvature forms~\eqref{eq:unadjusted_fields}, the 3-functor \(\varPhi\) can be parametrized according to
    \begin{subequations}
        \begin{align}
            &&g_\gamma &= \Pexp\int_\gamma A~,&& \\
            &&h_\sigma = \Pexp\ichint_{A}&B~,~~~
            g_\sigma = \Pexp\ichint_{A}(-\tilde F)~, \\
            &&h_\rho &= \Pexp\ichichint_{A,B} (-H)~,\label{eq:443c}
        \end{align}
    \end{subequations}
    where \(\tilde F \coloneqq F - \mu_1(B) = \dd A+\tfrac12[A,A]\) is the ordinary Yang--Mills curvature. This assignment is fixed by the mapping between the Weil algebra and the inner automorphism 2-crossed module; see appendix~\ref{app:inner_weil}, which also explains the origin of the minus signs appearing in front of the curvatures.
    
    The Chen form (see appendix~\ref{app:chen_forms}) relating $H$ to $h_{\rho}$ is obtained by lifting $H$ first to a 2-form $\chint H$ on path space using the $\sG$-connection $A$ and then, further to a 1-form $\chichint H$ on surface space. The last step requires that $H$ form an $\sH$-representation, which is only the case if $F=0$, according to equation~\eqref{eq:unadjusted_H_gauge_transformation}. Under a \(\sH\)-gauge transformation parametrized by \(\Lambda\), \(H\) mixes with \(F\), and cannot form an \(\sH\)-representation by itself. Fake flatness enters the picture yet again. 
    
    Similarly, in defining \(\chint_A B\), we must require \(\sB\) to form a \(\sG\)-representation, which is only the case if \(\mu_3(A,A,-)=0\) according to equation~\eqref{eq:unadjusted_B_gauge_transformation}.
        
    The globular structure of $\sB\sInn(\CCG)$ now induces Stokes’ theorems as follows. Given a 1-morphism \(\sigma\colon\gamma_1\to\gamma_2\) and a 2-morphism \(\rho\colon\sigma_1\to\sigma_2\), we have the globular identities
    \begin{equation}\label{eq:derivative_globular_identities}
            g_{\gamma_1}g_{\gamma_2}^{-1} = g_\sigma^{-1}~,~~~h_{\sigma_1}h_{\sigma_2}^{-1} = h_\rho^{-1}~,~~~
            g_{\sigma_1}g_{\sigma_2}^{-1} = \unit~.        
    \end{equation}
    The first identity fixes
    \begin{subequations}\label{eq:derivative_stokes_theorems}
        \begin{align}
            \dd A+\tfrac12[A,A] &= \tilde F \coloneqq F-\mu_1(B) ~.\label{eq:derivative_stokes_theorem_1}
            \intertext{The second and third translate to the identities}
            \dd_AB &= H~, \label{eq:derivative_stokes_theorem_2}\\
            \dd_A\tilde F &= 0~. \label{eq:derivative_stokes_theorem_3}
        \end{align}
    \end{subequations}
    Equations~\eqref{eq:derivative_stokes_theorem_1} and~\eqref{eq:derivative_stokes_theorem_3} hold automatically; equation~\eqref{eq:derivative_stokes_theorem_2}, however, only holds if \(\tfrac1{3!}\mu_3(A,A,A) = 0\), according to equation~\eqref{eq:unadjusted_fields_H}.

    The derivative parallel transport 3-functor now fits into the following commutative diagram:
    \begin{equation}\label{diag:commut_functors}
        \begin{tikzcd}
            \CP_{(2)} U\ar[d,"\Phi",swap] \ar[r,hookrightarrow] & \CP_{(3)}U  \ar[d,"\varPhi"] \\
            \sB \CCG \ar[r,hookrightarrow] & \sB \sInn(\CCG)
        \end{tikzcd}
    \end{equation}
    which makes it clear how gauge transformations should be defined. As in the case of ordinary gauge theory, we can fix endpoints \(x_0,x_1\in U\) and decategorify, considering the hom 2-categories. Then we can add the 2-functor \(\varPhi_\text{curv}\), which is a truncation of \(\varPhi\) to integrals of field strengths only:
    \begin{equation}
        \begin{tikzcd}
            \CP_{(2)} U(x_0,x_1)\ar[d,"{\Phi(x_0,x_1)}",swap] \ar[r,hookrightarrow] & \CP_{(3)}U(x_0,x_1)  \ar[d,"{\varPhi(x_0,x_1)}"']  \ar[dr,"{\varPhi_\mathrm{curv}(x_0,x_1)}"] & \\
            \CCG \ar[r,hookrightarrow] & \sInn(\CCG) \ar[r,twoheadrightarrow,"\Pi"] & \sB\CCG
        \end{tikzcd}
    \end{equation}
    Gauge transformations are then 3-natural transformations \(\varPhi\Rightarrow\tilde\varPhi\), which are general enough to include the pseudonatural transformations of 2-categories, and whose induced 2-natural transformations become trivial on the induced curvature 2-functors:
    \begin{equation}
        \varPhi_{\rm curv}(x_0,x_1)=\Pi\circ \varPhi(x_0,x_1)=\Pi\circ \tilde \varPhi(x_0,x_1)=\tilde \varPhi_{\rm curv}(x_0,x_1)~.
    \end{equation}

    \subsection{Adjusted higher parallel transport}\label{ssec:adjusted_parallel_transport}
    
    Above, we saw that we have two equivalent definitions of parallel transport. For an ordinary parallel transport based on a Lie group $\sG$ over a contractible manifold $U$, we can use either a functor $\Phi\colon\CP U\to \sB\sG$ or a strict 2-functor $\varPhi\colon\CP_{(2)}U\to \sB\sInn(\sG)$ with a restricted set of (higher) natural isomorphisms. This picture clearly generalizes to higher categorifications\footnote{As remarked in the introduction, one should use simplicial models for the higher categories in order to avoid the technicalities arising from higher coherence conditions.}.
    
    In the case of a strict gauge 2-group $\CCG=(\sH\rtimes\sG\rightrightarrows \sG)$, the globular structure of the 2-crossed module $\sInn(\CCG)$ induces fake flatness~\eqref{eq:fake_flatness}, which renders the theory essentially abelian. We have seen before that an adjustment of the Weil algebra, if it exists, can remove the necessity for fake flatness~\eqref{eq:fake_flatness}. The same is true in the case of higher parallel transport: the adjusted Weil algebra leads to an adjusted 2-crossed module of Lie groups, whose adjusted globular structure obviates the need for fake flatness.
    
    Since we need an adjustment, we must start from a gauge 2-group that admits one. Adjusted parallel transport for an adjustable crossed module of Lie groups $\CCG$ is then defined as a 3-functor
    \begin{equation}\label{eq:adj_parallel_transport_functor}
        \varPhi^\text{adj}\colon\CP_{(3)}U\to \sB\sInn_{\rm adj}(\CCG)~,
    \end{equation}
    which is the analogue of the derivative parallel transport 3-functors. There is no analogue of the 2-functor 
    \begin{equation}
        \Phi\colon\CP_{(2)}U\rightarrow \sB\sString_{\rm lp}(\sG)
    \end{equation}
    for $\varPhi^{\rm adj}$, unlike the other cases discussed so far in this section. This is as expected: adjustment is crucial to the existence of a well-defined notion of non-abelian higher parallel transport, and this is only visible at the level of the Weil algebra \(\sW_\text{adj}(\sLie(\CCG))\) or, correspondingly, the inner automorphism 2-group \(\sInn_\text{adj}(\CCG)\). It \emph{is}, however, possible to truncate the 3-functor to a 2-functor sensitive only to the curvatures, and for every pair of endpoints \(x_0,x_1\in U\) we have a commutative diagram
    \begin{equation}
        \begin{tikzcd}
            & \CP_{(3)}U(x_0,x_1)  \ar[d,"{\varPhi^\mathrm{adj}(x_0,x_1)}"']  \ar[dr,"{\varPhi^\mathrm{adj}_\mathrm{curv}(x_0,x_1)}"] & \\
            \CCG \ar[r,hookrightarrow] & \sInn_{\rm adj}(\CCG) \ar[r,twoheadrightarrow,"\Pi"] & \sB\CCG
        \end{tikzcd}
    \end{equation}
    where the bottom row is the short exact sequence~\eqref{eq:ses_groupoids_adjusted}. This diagram is the adjusted analogue of diagram~\eqref{diag:commut_functors}, without the nonexistent 2-functor $\Phi$. Similarly to the previous cases, the admissible gauge transformations are those natural transformations \(\eta\colon\varPhi\to\varPhi\) that are rendered trivial by the following whiskering.
    \begin{equation}
        \begin{tikzcd}
            \sB\CCG & \sInn(\CCG) \lar["\Pi",swap]  &  \CP_{(3)} U(x_0,x_1)\lar[bend left=50, "{\varPhi(x_0,x_1)}", ""'{name=A}] \lar[bend right=50, "{\tilde \varPhi(x_0,x_1)}"', ""{name=B}] \ar[from=A, to=B, Rightarrow,"\eta"]
        \end{tikzcd}
    \end{equation}
    
    To explain the 3-functor in more detail, let us focus on the archetypical example: (the generalization of) the loop model of the string group,
    \begin{equation}
        \sString_{\rm lp}(\sG)=\big(\hat L_0\sG\rtimes P_0\sG\rightrightarrows P_0\sG\big)\ewith \sLie(\CCG)=\big(\hat L_0 \frg\rightarrow P_0\frg\big)~,
    \end{equation}
    where $\sG$ is a finite-dimensional Lie group whose Lie algebra $\frg$ is metric.\footnote{Thus, $\sG$ admits a bi-invariant (pseudo-)Riemannian metric.} The Weil algebra of $\sLie(\CCG)$ admits an adjustment as discussed in section~\ref{ssec:adjusted_Weils}, and thus $\CCG$ admits an adjusted 3-group of inner automorphisms as explained in section~\ref{ssec:adjusted_inner_derivation}. Other examples of 2-groups admitting and adjustment can be treated similarly; in particular the discussion for the group $\sG_{\rm lp}=(L_0\sG\rtimes P_0\sG\rightrightarrows P_0\sG)$ discussed in section~\ref{ssec:loop_model_of_Lie_algebra} follows by truncation. 
    
    The 3-functor~\eqref{eq:adj_parallel_transport_functor} is then of the following form:
    \begin{subequations}
        \begin{equation}
            \varPhi^\text{adj}\colon\CP_{(3)}U\to \sB\sInn_{\rm adj}(\sString_{\rm lp}(\sG))
        \end{equation}
        with components consisting of the following maps:
        \begin{equation}
            \begin{tikzcd}
                \text{volumes} \dar[shift left] \dar[shift right] \rar["\varPhi^\text{adj}_3"] & \hat L_0\sG \rtimes \big((\hat L_0\sG \rtimes P_0\sG) \rtimes P_0\sG \big) \dar[shift left] \dar[shift right] \\
                \text{surfaces} \dar[shift left] \dar[shift right] \rar["\varPhi^\text{adj}_2"] & (\hat L_0\sG \rtimes P_0\sG) \rtimes P_0\sG \dar[shift left] \dar[shift right] \\
                \text{paths}  \dar[shift left] \dar[shift right] \rar["\varPhi^\text{adj}_1"] & P_0\sG \dar[shift left] \dar[shift right] \\
                U \rar["\varPhi^\text{adj}_0"] & *
            \end{tikzcd}
        \end{equation}
    \end{subequations}
    
    In terms of the fields~\eqref{eq:adjusted_loop_fields} taking values in the adjusted Weil algebra, all components of the 3-functor \(\varPhi\) can be covariantly defined:
    \begin{subequations}\label{eq:adj_hol_3_functor}
        \begin{align}
            g_\gamma &= \Pexp\int_\gamma A~, \\
            (h_\sigma,g_\sigma) &= \Pexp\ichint_{A}\binom B{-\tilde F}~, \\
            h_\rho^{-1} &= \Pexp\chichint_{A,B} (-H)~,
        \end{align}
    \end{subequations}
    where \(\tilde F\) is the ordinary Yang--Mills field strength
    \begin{equation}
        \tilde F = F - \mu_1(B) = \dd A+\tfrac12[A,A]~.
    \end{equation}
    Notice that the field \(B\) does not form a \(P_0\sG\)-representation by itself, which is similar to the problem with \(H\) in the unadjusted case. Happily, in the 2-crossed module \(B\) occurs together with \(\tilde F\), and \((B,\tilde F)\) does form a \(P_0\sG\)-representation, which can be exponentiated. Also, now \(H\) is gauge-\emph{in}variant, so that there is no problem defining it.
    We do not have any freedom to choose how to define the components of the parallel transport 3-functor; this is determined by the mapping between the adjusted Weil algebra and the adjusted inner derivation 2-crossed module.
    
    It remains to check that the required Stokes’ theorems hold. The globular identities~\eqref{eq:derivative_globular_identities} are unchanged from the unadjusted case and these correspond to the same Stokes’ theorems~\eqref{eq:derivative_stokes_theorems}, which we rewrite for clarity:
    \begin{subequations}
        \begin{align}
            \dd A+\tfrac12[A,A] &= \tilde F \coloneqq F-\mu_1(B)~,\\
            \dd_A\binom B{-\tilde F} &= \binom H0~.
        \end{align}
    \end{subequations}
    We write it thus to emphasize that \(B\) only forms a \(P_0\sG\)-representation together with \(\tilde F\). The first is the non-abelian Stokes’ theorem as before, and one can easily check that the second equation corresponds to the correct Bianchi identities~\eqref{eq:fake_curvature_bianchi_identity} and~\eqref{eq:H_bianchi_identity} for the adjusted Weil algebra.

    We make a few final remarks. The assignment~\eqref{eq:adj_hol_3_functor} indeed defines a strict 3-functor; verifying functoriality mostly consists of drawing elaborate diagrams, meditating on them, and concluding that they are trivial, especially since this 3-functor is strict. We leave this to the vigilant reader with free time (much as Cervantes dedicates \emph{Don Quijote} to the \emph{desocupado lector}).
    
    Technically, our path 3-groupoids are equivalence class of paths, surfaces, and volumes under \emph{thin homotopy}, which are homotopies of “zero volume” (see appendix~\ref{app:path_space}). Once we grant that the 3-functors are well-defined without this quotienting, a transformation by thin homotopy corresponds to a parallel transport along a zero-volume homotopy, which are given by the integral of the relevant curvatures, but this vanishes because the volume is zero. (In the case of the top curvature \(H\), one uses the Bianchi identity~\eqref{eq:H_bianchi_identity} for it.)

    In retrospect, the assertion of~\cite{Baez:2004in} that fake flatness is required for thin homotopy invariance was but an avatar of the fact that, without adjustment, gauge transformations only close if fake curvature vanishes. In the adjusted case, this defect is absent.
    
    \section*{Acknowledgments}

    This work is partially supported by the Leverhulme Research Project Grant RPG-2018-329 ``The Mathematics of M5-Branes.'' We would like to thank the anonymous referees for their helpful comments and technical corrections.

    \appendices

    \subsection{\texorpdfstring{\(L_\infty\)}{L∞}-algebras}\label{app:L_infinity_defs}
    In this appendix, we give definitions for \(L_\infty\)-algebras and explain our conventions. We only need to work over the field of real numbers. The original references on \(L_\infty\)-algebras are~\cite{Zwiebach:1992ie,Lada:1992wc,Lada:1994mn}; we follow the conventions in~\cite{Jurco:2018sby}, which may also be helpful.
    
    An \emph{\(L_\infty\)-algebra} $\frL$ consists of a $\RZ$-graded\footnote{In this paper, all \(L_\infty\)-algebras used are graded in nonpositive integers only.} vector space
    \begin{equation}
        \frL = \dotsb \oplus \frL_{-2} \oplus \frL_{-1} \oplus \frL_0 \oplus \frL_1\oplus \frL_2 \oplus \dotsb
    \end{equation}
    equipped with a set of \(i\)-ary multilinear totally graded-antisymmetric operations or {\em higher products}
    \begin{equation}
        \mu_i\colon \frL^{\wedge i}\to \frL
    \end{equation}
    for each positive integer \(i\), of degree $|\mu_i| = 2-i$, that satisfy the \emph{homotopy Jacobi identities}
    \begin{equation}\label{eq:hom_rel}
        \sum_{i+j=n}\sum_{\sigma\in S_{i|j}}\chi(\sigma;a_1,\dotsc,a_{n})(-1)^{j}\mu_{j+1}(\mu_i(a_{\sigma(1)},\dotsc,a_{\sigma(i)}),a_{\sigma(i+1)},\dotsc,a_{\sigma(n)})=0~.
    \end{equation}
    Here, the \emph{unshuffles} $S_{i|j}$ consist of permutations of $i+j$ elements in which the first $i$ and last $j$ elements are ordered and $\chi(\sigma;a_1,\dotsc,a_n)$ is the Koszul sign
    \begin{equation}
        a_1\wedge\dotsb\wedge a_n = \chi(\sigma;a_1,\dotsc,a_n)a_{\sigma(1)} \dotsm a_{\sigma(n)}~.
    \end{equation}
    The rather involved identities~\eqref{eq:hom_rel} are in fact simply an alternative way of writing the nilquadraticity of the homological vector field $Q$ of the Chevalley--Eilenberg algebra. To translate between both, let $\tau_A$ be a basis of $\frL$ and $\xi^A$ dual coordinate functions on $E=\frL[1]$. Then
    \begin{equation}
        \begin{aligned}
            Q(\xi^A\otimes \tau_A)&=-\sum_{i\geq 1} \tfrac{1}{i!}\mu_i\big(\xi^{A_1}\otimes \tau_{A_1},\dotsc,\xi^{A_i}\otimes \tau_{A_i}\big)\\
            &=-\sum_{i\geq 1} \tfrac{1}{i!}\zeta(A_1,\dots,A_i)\xi^{A_1}\dotsm \xi^{A_i}\otimes\mu_i\big(\tau_{A_1},\dotsc,\tau_{A_i}\big)~,
        \end{aligned}
    \end{equation}
    where the Koszul sign $\zeta(A_1,\dots,A_i)=\pm 1$ arises from permuting odd elements $\xi^{A_j}$ past odd elements $\tau_{A_k}$ or taking them out of odd higher products $\mu_k$. Expanding $Q^2=0$ then reproduces the homotopy Jacobi identities~\eqref{eq:hom_rel}.
    
    A \emph{strict \(L_\infty\)-algebra}, such as the loop model of the string Lie 2-algebra, is one in which \(\mu_i = 0\) for \(i\ge3\). That is, it is simply a differential graded Lie algebra, and the formidable homotopy Jacobi identities~\eqref{eq:hom_rel} simply reduce to the following:
    \begin{itemize}
        \setlength\itemsep{0em}
        \item[$\triangleright$] the differential \(\mu_1\) is nilquadratic;
        \item[$\triangleright$] the differential \(\mu_1\) acts as a graded derivation with respect to the graded bracket \(\mu_2\);
        \item[$\triangleright$] the graded bracket \(\mu_2\) satisfies the Jacobi identity.
    \end{itemize}

    \subsection{Categorified groups and hypercrossed modules}\label{app:hypercrossed_modules}
    
    Below, we describe Lie 2- and 3-groups in terms of (1-)crossed modules and 2-crossed modules of Lie groups, which are special cases of hypercrossed modules.
    
    \paragraph{Crossed modules.} Crossed modules of Lie groups provide particularly accessible and workable definitions of strict Lie 2-groups. Since every Lie 2-group is categorically equivalent to a strict Lie 2-group~\cite[Prop.~45]{Baez:0307200}, crossed modules are sufficient for most purposes.

    A {\em crossed module of Lie groups} is a pair of Lie groups, together with a group homomorphism $\sft$,
    \begin{equation}
        \sH \xrightarrow{~\sft~} \sG
    \end{equation}
    and a smooth action $\acton$ of $\sG$ on $\sH$ by automorphisms such that $\sft$ is a $\sG$-homomorphism and the Peiffer identity holds:
    \begin{equation}
        \sft(g\acton h_1)=g\sft(h_1)g^{-1}\eand \sft(h_1)\acton h_2=h_1h_2h_1^{-1}
    \end{equation}
    for all $g\in \sG$ and $h_1,h_2\in \sH$. 
    
    Any crossed module of Lie groups encodes a {\em strict Lie 2-group} in the sense of~\cite{Baez:0307200}, which is a strict monoidal category with strictly invertible objects and morphisms. In particular, the crossed module \((\sH\xrightarrow{~\sft~}\sG)\) gives rise to the monoidal category
    \begin{equation}\label{eq:mon_cat_from_2_group}
        \CCC(\sH\xrightarrow\sft\sG)\coloneqq (\sH\rtimes \sG \rightrightarrows \sG )
    \end{equation}
    with morphisms and structure maps
    \begin{equation}
        \begin{gathered}
            \sft(h)g\xleftarrow{~(h,g)~} g~,~~~(h_1,\sft(h_2)g)\circ (h_2,g)=(h_1h_2,g)~,\\
            \id_g=(\unit_\sH,g)~,~~~(h_1,g_1)\otimes (h_2,g_2)=\big(h_1(g_1\acton h_2),g_1g_2\big)~.
        \end{gathered}
    \end{equation}
    Inversely, any Lie 2-group encoded by a monoidal category $\CCG$ gives rise to a crossed module of Lie groups $\CCG_{\rm \cm}$.
    
    Applying the tangent functor to a crossed module and restricting to the units in $\sH$ and $\sG$, we arrive at the notion of a {\em crossed module of Lie algebras}. This is a pair of Lie algebras together with a Lie algebra homomorphism $\sft$,
    \begin{equation}
        \frh \xrightarrow{~\sft~} \frg
    \end{equation}
    and a representation $\acton$ of $\frg$ on $\frh$ such that 
    \begin{equation}
        \sft(a \acton b_1)=[a,\sft(b_1)]\eand \sft(b_1)\acton b_2=[b_1,b_2]
    \end{equation}
    for all $a\in \frg$ and $b_1,b_2\in \frh$. 
    
    \paragraph{2-Crossed modules.} There are several, obvious categorifications of crossed modules of Lie groups. Here, we focus on 2-crossed modules~\cite{Conduche:1984:155,Conduche:2003}, which encode semistrict Lie 3-groups called {\em Gray groups}, i.e.~Gray groupoids with a single object; see~\cite{Kamps:2002aa}.
    
    A {\em 2-crossed module of Lie groups} is a triple of Lie groups, arranged in the normal complex
    \begin{equation}
        \sL\ \xrightarrow{~\sft~}\ \sH\ \xrightarrow{~\sft~}\ \sG~,
    \end{equation}
    and endowed with smooth $\sG$-actions on $\sH$ and $\sL$ by automorphisms such that the maps $\sft$ are $\sG$-equivariant:
    \begin{equation}
        \sft(g\acton \ell)=g\acton \sft(\ell)\eand \sft(g\acton h)=g\sft(h)g^{-1}
    \end{equation}
    for all $g\in\sG$, $h\in\sH$, and $\ell\in\sL$. The Peiffer identity of crossed modules of Lie groups is violated, but this violation is controlled by the {\em Peiffer lifting}, which is a $\sG$-equivariant smooth map
    \begin{equation}
        \begin{gathered}
            \{-,-\}\colon\sH\times \sH\to \sL~,\\
        \end{gathered}
    \end{equation}
    satisfying the following relations:
    \begin{subequations}
        \begin{align}
            \sft(\{h_1,h_2\})&=h_1 h_2 h_1^{-1}(\sft(h_1)\acton h_2^{-1})~, \label{eq:global_peiffer_lifting_identity}\\
            \{\sft(\ell_1),\sft(\ell_2)\}&=\ell_1\ell_2\ell_1^{-1}\ell_2^{-1}~, \\
            \{h_1 h_2,h_3\}&=\{h_1,h_2h_3h_2^{-1}\}(\sft(h_1)\acton\{h_2,h_3\})~,\\
            \{h_1,h_2h_3\}&=\{h_1,h_2\}\{h_1,h_3\}\{\langle h_1,h_3\rangle^{-1},\sft(h_1)\acton h_2\}~,\\
            \ell_1 \left(\sft(h_1)\acton \ell^{-1}_1\right)&=\{\sft(\ell_1),h_1\}\{h_1,\sft(\ell_1)\}
        \end{align}
    \end{subequations}
    for all $h_i\in \sH$ and $\ell_i\in \sL$. 
    
    Given a 2-crossed module of Lie groups \(\sL\to\sH\to\sG\), we can construct a monoidal 2-category 
    \begin{subequations}\label{eq:mon_cat_from_3_group}
        \begin{equation}
            \CCC(\sL\to\sH\to\sG) \coloneqq (
            \sL \rtimes \sH \rtimes \sG \rightrightarrows  \sH \rtimes \sG \rightrightarrows  \sG  )~,
        \end{equation}
        whose globular structure is
        \begin{equation}
            \begin{tikzcd}[column sep=2.5cm,row sep=large]
                \sft(h)g & \ar[l, anchor=center,bend left=45, "{(h,g)}", ""{name=U,inner sep=1pt,above}] \ar[l,anchor=center, bend right=45, "{(\sft(\ell)h,g)}", swap, ""{name=D,inner sep=1pt,below}] 
                g \arrow[Rightarrow,from=U, to=D, "{(\ell,h,g)}",swap]
            \end{tikzcd}
        \end{equation}
    \end{subequations}
    see e.g.~\cite[Section 1.4]{Kamps:2002aa}. Shifting the degrees of all morphisms by one, we define the 3-groupoid $\sB(\CCC(\sL\to\sH\to\sG))$, which is a Gray groupoid. 
    
    Conversely, given a monoidal 2-category $\CCG$ encoding a 3-group, we denote the corresponding 2-crossed module of Lie groups by $\CCG_{\rm \cm}$.
    
    The infinitesimal counterpart of a 2-crossed module of Lie groups is a {\em 2-crossed module of Lie algebras}, which consists of a triple of Lie algebras arranged in the complex
    \begin{equation}
        \frl\ \xrightarrow{~\sft~}\ \frh\ \xrightarrow{~\sft~}\ \frg~.
    \end{equation}
    Additionally, we have $\frg$-actions $\acton$ onto $\frh$ and $\frl$ by derivations. The maps $\sft$ are equivariant with respect to these actions,
    \begin{equation}
        \sft(a\acton c)=a\acton\sft(c)\eand \sft(a\acton b)=[a,\sft(b)]
    \end{equation}
    for all $a\in \frg$, $b\in \frh$, and $c\in \frl$. The violation of the Peiffer identity is controlled by a differential version of the Peiffer lifting, which is a $\frg$-equivariant bilinear map
    \begin{equation}
        \{-,-\}\colon \frh\times \frh\to \frl~,
    \end{equation}
    which also satisfies the following relations:
    \begin{subequations}\label{eq:ax:2-crossed_Lie}
        \begin{align}
            \sft(\{b_1,b_2\})&=[b_1,b_2]-\sft(b_1)\acton b_2~, \label{eq:2cm_axiom2}\\
            \{\sft(c_1),\sft(c_2)\}&=[c_1,c_2]~, \\
            \{b_1,[b_2,b_3]\}&=\{\sft(\{b_1,b_2\}),b_3\}-\{\sft(\{b_1,b_3\}),b_2\}~,\\
            \{[b_1,b_2],b_3\}&=\sft(b_1)\acton\{b_2,b_3\}+\{b_1,[b_2,b_3]\}-\sft(b_2)\acton\{b_1,b_3\}-\{b_2,[b_1,b_3]\}~,\\
            -\sft(b_1)\acton c_1&=\{\sft(c_1),b_1\}+\{b_1,\sft(c_1)\} \label{eq:2cm_axiom6}
        \end{align}
    \end{subequations}
    for all $b_1,b_2,b_3\in \frh$ and $c_1,c_2\in \frl$.
    
    Given a 2-crossed module of Lie algebras $\frl\xrightarrow{~\sft~}\frh \xrightarrow{~\sft~}\frg$, the subcomplexes $\frl\xrightarrow{~\sft ~}\frh$ with action
    \begin{equation}
        b\acton c\coloneqq -\{\sft(c),b\}~,~~~b\in \frh~,~~~c\in \frl
    \end{equation}
    as well as $\sft(\frl)\setminus\frh \xrightarrow{~\sft~}\frg$ with the unmodified action of $\frg$ on $\sft(\frl)\setminus\frh$ also form crossed modules of Lie algebras. 
    
    We explain the relationship between Lie 1-, 2-, and 3-algebras and certain hypercrossed modules of Lie algebras in appendix~\ref{app:Lie3_2_crossed}.

    \subsection{Path and loop groups}\label{app:path_groups}
    
    The construction of the strict 2-group model of the string group~\cite{Baez:2005sn} requires a particular technical choice of path groups and loop groups. In short, path groups are smooth and based; loop groups are based, and consist of loops that are smooth everywhere except at the base point, where they are merely continuous.
    
    Given a finite-dimensional Lie group \(\sG\), the \emph{path group} \(P_0\sG\) is the Fréchet--Lie group of smooth paths \(\gamma\colon [0,1]\to \sG\) such that \(\gamma(0) = \unit_\sG\). The group operation is pointwise multiplication. The \emph{loop group} \(L_0\sG\) is the subgroup of those paths \(\gamma\) such that \(\gamma(0)=\gamma(1)\). We do \emph{not} require any further smoothness at the base point. Thus there is a non-split short exact sequence
    \begin{equation}
        * \to L_0\sG \to P_0\sG \overset\partial\to \sG \to *~,
    \end{equation}
    where \(\partial\colon P_0\sG \to \sG\) is the endpoint evaluation map.
    Given the Lie algebra \(\frg\) of \(\sG\), the corresponding Lie algebras are \(P_0\frg\) and \(L_0\frg\), with obvious definitions and the corresponding non-split short exact sequence
    \begin{equation}
        * \to L_0\frg \to P_0\frg \overset\partial\to \frg \to *~.
    \end{equation}
    
    The Fréchet–Lie group \(\hat L_0\sG\) is the usual Kac--Moody central extension of \(L_0\sG\). Its Lie algebra is
    \begin{equation}
        \hat L_0\frg = L_0\frg \oplus \mathbb R~,
    \end{equation}
    where \(\mathbb R\) is the 1-dimensional abelian Lie algebra and \(\oplus\) is a direct sum of Lie algebras.
    While at the level of Lie algebras \(\hat L_0\frg\) is just a trivial direct sum, at the level of Lie groups \(\hat L_0\sG\) is a nontrivial principal $\sU(1)$-bundle over \(L_0\sG\). We thus have the exact sequences
    \begin{equation}
        * \to \operatorname \sU(1) \to \hat L_0\sG \to P_0\sG \overset\partial\to \sG \to *
    \end{equation}
    and
    \begin{equation}
        * \to \mathbb R\to \hat L_0\frg \to P_0\frg \overset\partial\to \frg \to *~.
    \end{equation}

    \subsection{Strict Lie 3-algebras and 2-crossed modules of Lie algebras}\label{app:Lie3_2_crossed}
    
    Semistrict Lie 3-algebras can be described both by 2-crossed modules of Lie algebras as well as 3-term \(L_\infty\)-algebras. For our purposes, the precise relation between these is important. Because we could not find the relevant statements in the literature, we give them below.
    
    We first mention the comparison theorems between \(n\)-term \(L_\infty\)-algebras and \((n-1)\)-crossed modules of Lie algebras for \(n\le 2\):
    \begin{theorem}
        A 1-term \(L_\infty\)-algebra is the same thing as a 0-crossed module of Lie algebras (i.e.~a Lie algebra).
    \end{theorem}\begin{proof}
        Trivial.
    \end{proof}
    
    \begin{theorem}
        A strict 2-term \(L_\infty\)-algebra is the same thing as a (1-)crossed module of Lie algebras.
    \end{theorem}\begin{proof}
        Given a strict Lie 2-algebra
        \begin{equation}
            \frL=\big(\frL_{-1} \xrightarrow{~\mu_1~} \frL_0\big)~,
        \end{equation}
        we can construct the crossed module of Lie algebras 
        \begin{equation}
            \begin{matrix}
                \frh & \overset{\sft}\longrightarrow & \frg \\
                b & \overset{\sft}\longmapsto & \mu_1(b)
            \end{matrix}
        \end{equation}
        with \(\frg=\frL_0\) and \(\frh=\frL_{-1}\) and
        \begin{equation}
            [a_1,a_2]_\frg=\mu_2(a_1,a_2)~,~~~
            a_1\acton b_1=\mu_2(a_1,b_1)~,~~~
            [b_1,b_2]_\frh=\mu_2(\mu_1(b_1),b_2)
        \end{equation}
        for all $a_1,a_2\in\frg$ and $b_1,b_2\in \frh$. The inverse construction is also evident.
    \end{proof}
    The next step up in the categorification process turns out to be a bit more complicated.
    \begin{theorem}\label{thm:2cm_to_L_infty}
        The complex of Lie algebras underlying a 2-crossed module of Lie algebras comes with a strict 3-term \(L_\infty\)-algebra structure.
    \end{theorem}
    \begin{proof}
        Given a 2-crossed module of Lie algebras
        \begin{equation}
            (\frl \xrightarrow{~\sft~} \frh \xrightarrow{~\sft~}\frg,\acton,\{-,-\})~,
        \end{equation}
        there is a strict 3-term \(L_\infty\)-algebra
        \begin{equation}
            \begin{matrix}
                \frL\quad=\quad(&\frL_{-2}&\overset{\mu_1}\longrightarrow&\frL_{-1}&\overset{\mu_1}\longrightarrow&\frL_0&)~,\\
                &c&\overset{\mu_1}\longmapsto&\sft(c)\\
                &&&b&\overset{\mu_1}\longmapsto&\sft(b)
            \end{matrix}
        \end{equation}
        where \(\frL_{-2}=\frl\) and \(\frL_{-1}=\frh\) and \(\frL_0=\frg\),
        with non-trivial higher products
        \begin{subequations}
            \begin{align}
                \mu_2(a_1,a_2) &\coloneqq [a_1,a_2]_\frg~, \\
                \mu_2(a_1,b_1)&\coloneqq a_1\acton b_1~, \\
                \mu_2(a_1,c)&\coloneqq a_1\acton c~,\\
                \mu_2(b_1,b_2)&\coloneqq -\{b_1,b_2\}-\{b_2,b_1\}\label{eq:Peiffeb_sym_para_mu2_relation}
            \end{align}
        \end{subequations}
        for all $a_1,a_2\in \frg$, $b_1,b_2\in \frh$, and $c\in \frl$. One readily verifies that the homotopy Jacobi identity is satisfied for these higher products, as in Table~\ref{table:2cm_Lie3_correspondence}.
        \begin{table}
            \begin{center}
                \begin{tabular}{cc}
                    \toprule
                    2-crossed module &homotopy Jacobi identity \\ \midrule
                    %% 1 argument
                    \(\sft\circ\sft = 0\) & \(\mu_1(\mu_1(\frL_{-2}))\) \\
                    %% 2 arguments, degree~1
                    \(\mathfrak g\)-equivariance of map \(\sft\colon\mathfrak h\to\mathfrak g\) & \(\mu_2(\frL_0,\mu_1(\frL_{-1}))\) \\
                    %% 2 arguments, degree~2
                    \(\mathfrak g\)-equivariance of map \(\sft\colon\mathfrak l\to\mathfrak h\) & \(\mu_2(\frL_0,\mu_1(\frL_{-2}))\) \\
                    symmetric part of~\eqref{eq:2cm_axiom2} & \(\mu_2(\mu_1(\frL_{_-1}),\frL_{-1})\) \\
                    %% 2 arguments, degree~3
                    \eqref{eq:2cm_axiom6}&\(\mu_2(\mu_1(\frL_{-1}),\frL_{-2})\) \\
                    %% 3 arguments, degree~0
                    Jacobi identity for \(\mathfrak g\) Lie bracket & \(\mu_2(\mu_2(\frL_0,\frL_0),\frL_0)\) \\
                    %% 3 arguments, degree~1
                    \(\mathfrak g\)-action on \(\mathfrak h\) & \(\mu_2(\mu_2(\frL_0,\frL_0),\frL_{-1})\) \\
                    %% 3 arguments, degree~2
                    \(\mathfrak g\)-action on \(\mathfrak l\) & \(\mu_2(\mu_2(\frL_0,\frL_0),\frL_{-2})\) \\
                    symmetric part of \(\mathfrak g\)-equivariance of Peiffer lifting & \(\mu_2(\mu_2(\frL_{-1},\frL_{-1}),\frL_0)\) \\
                    \bottomrule 
                \end{tabular}
            \end{center}
            \caption{Proof that a 2-crossed module of Lie algebras defines an \(L_\infty\)-algebra}\label{table:2cm_Lie3_correspondence}
        \end{table}
    \end{proof}
    \begin{theorem}\label{thm:L_infty_to_2cm}
        A strict Lie 3-algebra \(\frL = \frL_{-2} \oplus \frL_{-1} \oplus \frL_0\) equipped with a choice of graded-symmetric (i.e.~antisymmetric) bilinear map
        \begin{subequations}\label{eq:add_map}
            \begin{equation}
                \llbracket-,-\rrbracket\colon \frL_{-1} \times \frL_{-1} \to \frL_{-2}
            \end{equation}
            which satisfies the identities
            \begin{equation}\label{eq:req_identities}
                \begin{aligned}
                    \llbracket  b_2,\mu_2( b_3,\mu_1( b_1))\rrbracket-\llbracket  b_3,\mu_2( b_2,\mu_1( b_1))\rrbracket+\mu_2(\mu_1( b_1),\llbracket  b_2, b_3\rrbracket)&=0~,\\
                    \llbracket  b_1,\mu_1(\llbracket  b_2, b_3\rrbracket)\rrbracket-
                    \llbracket  b_2,\mu_1(\llbracket  b_1, b_3\rrbracket)\rrbracket+
                    \llbracket  b_3,\mu_1(\llbracket  b_1, b_2\rrbracket)\rrbracket-~~&\\
                    -\tfrac14 \mu_2( b_1,\mu_2( b_2,\mu_1( b_3)))
                    +\tfrac14 \mu_2( b_3,\mu_2( b_2,\mu_1( b_1)))&=0
                \end{aligned}
            \end{equation}
        \end{subequations}
        for all $ b_1,b_2,b_3\in \frL_{-1}$ comes with the structure of a 2-crossed module on its underlying graded vector space, where the Peiffer lifting reads as
        \begin{equation}\label{eq:peiffer_thm}
            \{ b_1, b_2\}=\llbracket b_1, b_2\rrbracket-\tfrac12\mu_2( b_1, b_2)
        \end{equation}
        for all $ b_1, b_2\in \frL_{-1}$.
    \end{theorem}
    \begin{proof}
        Given a 3-term \(L_\infty\)-algebra $\frL=\frL_{-2}\oplus \frL_{-1}\oplus \frL_0$, we construct the complex underlying the 2-crossed module of Lie algebras
        \begin{equation}
            (\frl \xrightarrow{~\sft~} \frh \xrightarrow{~\sft~}\frg,\acton,\{-,-\})
        \end{equation}
        with 
        \begin{equation}
            \frl=\frL_{-2}~,~~~\frh=\frL_{-1}~,~~~\frg=\frL_0~,\eand \sft=\mu_1~.
        \end{equation}
        The Lie bracket on $\frg$ is given by 
        \begin{equation}
            [-,-]_\frg=\mu_2\colon\frg\wedge \frg\to \frg~,
        \end{equation}
        and the actions of $\frg$ on $\frh$ and $\frl$ read as 
        \begin{equation}
            a\acton b\coloneqq \mu_2(a,b)\eand a\acton c\coloneqq \mu_2(a,c)
        \end{equation}
        for $a\in \frg$, $b\in \frh$, $c\in \frl$. The Peiffer lifting~\eqref{eq:peiffer_thm} fixes the Lie brackets on $\frh$ and $\frl$ as
        \begin{subequations}
            \begin{align}
                \label{eq:h_Lie_bracket}
                [b_1,b_2]_{\frh} &\coloneqq \mu_1(\llbracket b_1,b_2\rrbracket) + \tfrac12\big(\mu_2(\mu_1(b_1),b_2)-\mu_2(\mu_1(b_2),b_1)\big)~,\\
                \label{eq:l_Lie_bracket}
                [c_1,c_2]_{\frl} &\coloneqq  \{\mu_1(c_1),\mu_2(c_2)\} = \llbracket\mu_1(c_1),\mu_1(c_2)\rrbracket
            \end{align}
        \end{subequations}
        for all \(b_1,b_2\in\frh\) and $c_1,c_2\in \frl$. Straightforward but lengthy algebraic computations show that these structures satisfy the axioms of a 2-crossed module of Lie algebras~\eqref{eq:ax:2-crossed_Lie} if and only if~\eqref{eq:req_identities} are satisfied.
    \end{proof}
    
    \begin{corollary}\label{cor:trivial_peiffer_lifting}
        Under the correspondence given by theorems~\ref{thm:2cm_to_L_infty} and~\ref{thm:L_infty_to_2cm}, the class of 2-crossed modules of Lie algebras with vanishing Peiffer lifting corresponds precisely to the class of 3-term \(L_\infty\)-algebras with vanishing \(\mu_2\colon \frL_{-1}\wedge\frL_{-1} \to \frL_{-2}\).
    \end{corollary}
    \begin{proof}
        This follows from theorems~\ref{thm:2cm_to_L_infty} and~\ref{thm:L_infty_to_2cm} with the observation that in these cases, the identities~\eqref{eq:req_identities} are trivial.
    \end{proof}

    Altogether, we conclude that 2-crossed modules of Lie algebras readily restrict to strict 3-term $L_\infty$-algebras, but strict 3-term $L_\infty$-algebras can only be extended to 2-crossed modules, if they allow for maps~\eqref{eq:add_map}.

    \subsection{Inner derivations and the Weil algebra}\label{app:inner_weil}
    
    Conceptually, the Weil algebra of a Lie \(n\)-algebra and the inner derivation \(n\)-crossed module of the Lie \(n\)-algebra are similar: both involve doubling the number of generators, with augmented degree, so as to be “topologically (or cohomologically) trivial”. In this appendix, we show that, under the comparison theorems of appendix~\ref{app:Lie3_2_crossed}, the two are in fact precisely the same, for \(n\le2\).
    
    First, we review the case for \(n=1\).
    \begin{thm}\label{thm:Weil_and_inner_Lie1}
        Given a Lie algebra $\frg$, the Chevalley--Eilenberg algebra of the 2-term \(L_\infty\)-algebra corresponding to the crossed module of Lie algebras $\ainn(\frg)$ is isomorphic to $\sW(\frg)$.
    \end{thm}
    \begin{proof}
        The Lie 2-algebra corresponding to the inner derivation crossed module \(\ainn(\frg)\) is $\frg[1]\xrightarrow{~\id~}\frg$
        with binary products
        \begin{equation}
            \mu_2(a_1,a_2)=[a_1,a_2] \eand
            \mu_2(a_1,\hat a_2)=[a_1,\hat a_2]~,
        \end{equation}
        for all $a_1,a_2\in \frg$ and $\hat a_1,\hat a_2\in \frg[1]$. With respect to some basis, its Chevalley--Eilenberg algebra is generated by elements $w^\alpha\in\frg[1]^*$ and $\hat w^\alpha\in\frg[2]^*$ and comes with the differential $Q_{\ainn}$ acting on the generators according to
        \begin{equation}
            \begin{aligned}
                Q_{\ainn}~&\colon~&v^\alpha&\mapsto-\tfrac12 f^\alpha_{\beta\gamma}v^\beta v^\gamma-\hat v^\alpha~,~~~&\hat v^\alpha&\mapsto -f^\alpha_{\beta\gamma}v^\beta\hat v^\gamma~,
            \end{aligned}
        \end{equation}
        where $f^\alpha_{\beta\gamma}$ are the structure constants of $\frg$.
        
        On the other hand, the Weil algebra $\sW(\frg)$ is generated by elements $t^\alpha \in \frg[1]^*$ and $\hat t^\alpha \in \frg[2]^*$ and the differential acts as 
        \begin{equation}
            \begin{aligned}
                Q_{\sW}~&\colon~&t^\alpha&\mapsto-\tfrac12 f^\alpha_{\beta\gamma}t^\beta t^\gamma + \hat t^\alpha~,~~~&\hat t^\alpha&\mapsto -f^\alpha_{\beta\gamma}t^\beta\hat t^\gamma~.
            \end{aligned}
        \end{equation}
        
        Comparing the action of the two differentials, it is obvious that 
        \begin{equation}
            v^\alpha \mapsto t^\alpha~,~~~\hat v^\alpha\mapsto -\hat t^\alpha~
        \end{equation}
        yields an isomorphism (or strict dual quasi-isomorphism) of differential graded algebras.
    \end{proof}
    
    The previous theorem categorifies for Lie 2-algebras.
    \begin{thm}
        Given a crossed module of Lie algebras $(\frh\xrightarrow{~\tilde \sft~}\frg,\tildeacton)$, the Chevalley--Eilenberg algebra of the strict 3-term \(L_\infty\)-algebra obtained as in theorem~\ref{thm:2cm_to_L_infty} from the 2-crossed module of Lie algebras $\ainn(\frh\xrightarrow{~\tilde \sft~}\frg)$ is isomorphic to the Weil algebra of $\frh\xrightarrow{~\tilde \sft~}\frg$.
    \end{thm}
    \begin{proof}
        Theorem~\ref{thm:2cm_to_L_infty} yields the following Lie 3-algebra for  $\ainn(\frh\xrightarrow{~\tilde \sft~}\frg)$:
        \begin{subequations}
        \begin{equation}
            \begin{matrix}
                (&\frh &\xrightarrow{~\mu_1~}& \frh\rtimes \frg&\xrightarrow{~\mu_1~}& \frg&)~, \\
                &b & \xmapsto{~\mu_1~} & (-b,\tilde\sft(b))  \\
                &&& (b,a) & \xmapsto{~\mu_1~} & \tilde\sft(b)+a
            \end{matrix}
        \end{equation}
        with binary products given by
        \begin{equation}
            \begin{aligned}
                \mu_2(a_1,a_2)&=[a_1,a_2]~,~~~&\mu_2\big((b_1,a_1),(b_2,a_2)\big)&=-\big(a_2 \tildeacton b_1+a_1 \tildeacton b_2\big)~,\\
                \mu_2\big(a_1,(b_2,a_2)\big)&=\big(a_1\tildeacton b_2,[a_1,a_2]\big)~,~~~
                &\mu_2(b_1,a_1)&=a_1\tildeacton b_1
            \end{aligned}
        \end{equation}
        \end{subequations}
        for all $a_1,a_2\in \frg$ and $b_1,b_2\in \frh$. Its Chevalley--Eilenberg algebra is generated by elements
        \begin{equation}
            v^\alpha\in\frg[1]^*~,~~~(w^a,\hat v^\alpha)\in(\frh\rtimes\frg)[2]^*~,~~~\hat w^a \in\frh[3]^*
        \end{equation}
        and the differential acts as
        \begin{equation}
            \begin{aligned}
                Q_{\ainn}~&\colon~&v^\alpha&\mapsto-\tfrac12 f^\alpha_{\beta\gamma}v^\beta v^\gamma-f^\alpha_a w^a - \hat v^\alpha~,~~~&\hat v^\alpha&\mapsto -f^\alpha_{\beta\gamma}v^\beta\hat v^\gamma+f_a^\alpha \hat w^a~,\\
                && w^a&\mapsto -f^a_{\alpha b}v^\alpha w^b-\hat w^a~, &\hat w^a &\mapsto -f^a_{\alpha b}v^\alpha \hat w^b+f^a_{\alpha b}\hat v^\alpha w^b~,
            \end{aligned}
        \end{equation}
        where $f^\alpha_a$, $f^\alpha_{\beta\gamma}$, and $f^a_{\alpha b}$ are the structure constants defining $\sft$, the Lie bracket on $\frg$ and the $\frg$-action $\acton$ on $\frh$.
        
        On the other hand, the Weil algebra $\sW(\frh\xrightarrow{~\tilde \sft~}\frg,\tildeacton)$ is generated by elements
        \begin{equation}
            t^\alpha \in \frg[1]^*~,~~~
            r^a \in \frh[2]^*~,~~~
            \hat t^\alpha \in \frg[2]^*~,~~~
            \hat r^a \in \frh[3]^*~,
        \end{equation}
        and the differential acts according to
        \begin{equation}
            \begin{aligned}
                Q_{\sW}~&\colon~&t^\alpha&\mapsto-\tfrac12 f^\alpha_{\beta\gamma}t^\beta t^\gamma-f^\alpha_a r^a+ \hat t^\alpha~,~~~&\hat t^\alpha&\mapsto -f^\alpha_{\beta\gamma}t^\beta\hat t^\gamma+f_a^\alpha \hat r^a~,\\
                && r^a&\mapsto -f^a_{\alpha b}t^\alpha r^b+\hat r^a ~,&\hat r^a &\mapsto -f^a_{\alpha b}t^\alpha \hat r^b+f^a_{\alpha b}\hat t^\alpha r^b~.
            \end{aligned}
        \end{equation}
        
        Comparing the differentials, it is again obvious that 
        \begin{equation}
            v^\alpha \mapsto t^\alpha~,~~~\hat v^\alpha\mapsto -\hat t^\alpha~,~~~w^a\mapsto r^a~,~~~\hat w^a\mapsto -\hat r^a
        \end{equation}
        yields an isomorphism of differential graded algebras.
    \end{proof}
    
    In both theorems, we encountered unfortunate minus signs in the isomorphism, which is a consequence of our being stuck between the hammer of standard conventions for the Weil algebra and the anvil of standard conventions for the semidirect product.
    
    Regardless, the 3-term \(L_\infty\)-algebra encoded in the Weil algebra of a strict Lie 2-algebra is canonically isomorphic as an $L_\infty$-algebra to the 3-term \(L_\infty\)-algebra underlying the inner derivation 2-crossed module of the strict Lie 2-algebra. We stress, however, that the inner derivation 2-crossed module of Lie algebras contains additional data, namely the antisymmetric part of the Peiffer lifting
    \begin{equation}
        \llbracket (b_1,a_1),(b_2,a_2)\rrbracket = \tfrac12\left(a_2 \tildeacton b_1 - a_1 \tildeacton b_2\right)~.
    \end{equation}

    \subsection{Quasi-isomorphisms and strict 2-group equivalences}\label{app:Equivalence}
    
    Morphisms of $L_\infty$-algebras are most readily understood in their dual formulation: as morphisms of differential graded algebras between the corresponding Chevalley--Eilenberg algebras. Such a morphism descends to a morphism between the $\mu_1$-cohomologies of the $L_\infty$-algebras. A \emph{quasi-isomorphism} between two $L_\infty$-algebras $\frL$ and $\tilde \frL$ is a morphism of $L_\infty$-algebras $\phi\colon\frL\to \tilde \frL$, which descends to an isomorphism $\phi_*\colon\operatorname H^\bullet_{\mu_1}(\frL)\to\operatorname H^\bullet_{\mu_1}(\tilde \frL)$. For more details, see e.g.~\cite{Jurco:2018sby}. Quasi-isomorphisms are indeed the appropriate notion of equivalence for most intents and purposes. For example, quasi-isomorphic gauge $L_\infty$-algebras lead to quasi-isomorphic, and thus physically equivalent, BRST complexes~\cite{Saemann:2019dsl}.
    
    By the {\em minimal model theorem}, any $L_\infty$-algebra $\frL$ is quasi-isomorphic to an $L_\infty$-algebra with underlying graded vector space $H^\bullet_{\mu_1}(\frL)$, which is called a {\em minimal model} for $\frL$.
    
    As an example, we explain the quasi-isomorphism between two strict Lie 2-algebras relevant to our discussion, namely
    \begin{equation}
        \frg=(*\xrightarrow{~~~} \frg)\eand \frg_{\rm lp}\coloneqq (L_0\frg\xhookrightarrow{~~~} P_0\frg)~,
    \end{equation}
    where $L_0\frg$ and $P_0 \frg$ are based loop\footnote{Note that \(L_0\frg\), the ordinary loop space, is \emph{not} the same as \(\hat L_0\frg\), which contains a central extension, used in the loop model of the string 2-algebra.} and path spaces in $\frg$, cf.~appendix~\ref{app:path_groups}. Besides the embedding $\mu_1$, the only other non-trivial higher product is in both cases $\mu_2$ given by the obvious commutators. The quasi-isomorphism between these two strict Lie 2-algebras is a truncation of a quasi-isomorphism given in~\cite[Lemma 37]{Baez:2005sn}. We have morphisms of Lie 2-algebras $\phi$ and $\psi$,
    \begin{subequations}\label{eq:quasi-iso-morphs}
    \begin{equation}
        \begin{tikzcd}
            \frg \ar[r,bend left=45, "\phi" pos=0.5] & \frg_{\rm lp}\ar[l,bend left=45, pos=0.45]{}{\psi} 
        \end{tikzcd}
    \end{equation}
    which are given explicitly by the chain maps 
    \begin{equation}
        \begin{tikzcd}
            * \ar[d] \ar[r] & L_0\frg \ar[d,hookrightarrow] \ar[r] & * \ar[d] \\
            \frg \ar[r]{}{\cdot \ell(\tau)}& P_0\frg \ar[r]{}{\dpar} & \frg
        \end{tikzcd}
    \end{equation}
    \end{subequations}
    where $\dpar\colon P_0\frg\to \frg$ is again the endpoint evaluation and $\cdot \ell(\tau)\colon\frg\to P_0\frg$ embeds $\alpha_0\in\frg$ as the line $\alpha(\tau)=\alpha_0 \ell(\tau)$ for some smooth function $\ell\colon [0,1]\to \FR$ with $f(0)=0$ and $f(1)=1$. Both maps $\phi$ and $\psi$ descend to isomorphisms on the cohomologies 
    \begin{equation}
        \frg\cong\operatorname H^\bullet_{\mu_1}(*\to \frg)\cong\operatorname H^\bullet_{\mu_1}(\frg_{\rm lp})=(*\to \frg)~,
    \end{equation}
    and $(*\to \frg)=\operatorname H^\bullet_{\mu_1}(*\to \frg)$ is thus indeed a minimal model for $\frg_{\rm lp}$.
    
    We can complete the morphisms in~\eqref{eq:quasi-iso-morphs} to a categorical equivalence by adding a contracting homotopy: $(\Psi\circ \Phi)_0$ is already the identity, and we have a 2-morphisms of Lie 2-algebras $\eta:\Phi\circ \Psi\to \id_{\frg_{\rm lp}}$ encoded in
    \begin{equation}\label{eq:def_chi}
        \eta:P_0\frg\to L_0\frg~,~~~\eta(\gamma)=\gamma-\ell(\tau)\dpar\gamma~.
    \end{equation}

    Strict Lie 2-algebras integrate to particular Lie groupoids, which carry the structure of a 2-group and $\frg_{\rm lp}$ integrates to the 2-group 
    \begin{equation}
        \sG_{\rm lp}\coloneqq(L_0 \sG\rtimes P_0\sG\rightrightarrows P_0\sG)=\CCC(L_0\sG\xrightarrow{~\sft~}P_0\sG)~.
    \end{equation}
    We thus expect $\sG_{\rm lp}$ to be equivalent to the 2-group $*\rightarrow \sG$ in a suitable sense. This is the case as we will show now.
    
    Since both 2-groups are strict, the appropriate notion of morphism is given by butterflies, cf.~\cite{Aldrovandi:0808.3627}, with an equivalence corresponding to a flippable butterfly. Given two crossed modules of Lie groups $\CCG\coloneqq(\sH\xrightarrow{~\sft}\sG)$ and $\tilde \CCG\coloneqq(\tilde\sH\xrightarrow{~\tilde\sft}\tilde\sG)$, a {\em butterfly} from $\CCG$ to $\tilde \CCG$ is a commutative diagram of group homomorphisms
    \begin{equation}
        \begin{tikzcd}
            \sH \arrow[dd,"\sft",swap] \arrow[dr,"e_1"]& & \tilde \sH\arrow[dl,"e_2",swap] \arrow[dd,"\tilde \sft"]
            \\
            & \hat \sG \arrow[dl,"\pi_1",swap] \arrow[dr,"\pi_2"]& 
            \\
            \sG & & \tilde \sG
        \end{tikzcd}
    \end{equation}
    where $\sH\xrightarrow{~e_1~}\hat \sG \xrightarrow{~\pi_2~}\tilde \sG$ is a complex and $\tilde \sH\xrightarrow{~e_2~}\hat \sG\xrightarrow{~\pi_1~}\sG$ is a group extension. Moreover, we have the equivariance condition
    \begin{equation}
        e_2(\,\pi_2(x)\tilde \acton \tilde h\,)=x^{-1}\,e_2(\tilde h)\,x~~~\mbox{and}~~~e_1(\,\pi_1(x)\acton h\,)=x^{-1}\,e_1( h)\,x
    \end{equation}
    for all $x\in \hat \sG$, $h\in \sH$, and $\tilde h\in \tilde \sH$. A butterfly from $\CCG$ to $\tilde \CCG$ is {\em flippable} if it is also a butterfly from $\tilde \CCG$ to $\CCG$.
    
    We now have the flippable butterfly
    \begin{equation}
        \begin{tikzcd}
            * \arrow[dd] \arrow[dr]& & L_0\sG \arrow[dl,hookrightarrow] \arrow[dd,hookrightarrow]
            \\
            & P_0\sG \arrow[dl,"\dpar",swap] \arrow[dr,"\sfid"]& 
            \\
            \sG & & P_0\sG
        \end{tikzcd}
    \end{equation}
    proving the equivalence of the strict 2-groups $\sG_{\rm lp}$ and $(*\xrightarrow{~~~}\sG)$.

    \subsection{Path groupoids}\label{app:path_space}
    
    We need groupoids and higher groupoids of smooth, parametrized paths, but generic such paths with coincident endpoints fail to compose smoothly and associatively. To remedy this, we follow~\cite{Caetano:1993zf,Schreiber:0705.0452} and introduce sitting instants and factor by thin homotopies. This appendix summarizes some of the technical details underlying our path groupoids.
    
    Suppose we are given a manifold \(M\). A \emph{path with sitting instants} is a smooth map $\gamma\colon [0,1]\to M$, regarded as a morphism
    \begin{equation}
        x_1\xleftarrow{~\gamma~}x_0
    \end{equation}
    with \emph{sitting instants} at the endpoints $x_0=\gamma(0)$, $x_1=\gamma(1)$. That is, there is an $\eps>0$ such that for $i\in\{0,1\}$ and all $|t-i|\leq \eps$, the map $\gamma$ is constant: $\gamma(t)=x_i$. We abbreviate this by writing
    \begin{equation}
        t\approx i\in\{0,1\}~~\Rightarrow~~\gamma(t)= x_i~.
    \end{equation}
    This ensures smooth composition of paths.
    
    A \emph{homotopy with sitting instants} between two paths \(\gamma_0,\gamma_1\colon[0,1]\to M\) sharing common endpoints \(x_0,x_1\in M\) is a smooth homotopy
    \begin{equation}
        \sigma\colon [0,1]\times[0,1]\to M~,\hspace{1.5cm}
        \begin{tikzcd}[column sep=1.5cm,row sep=large]
            x_1 & \ar[l, anchor=center,bend left=45, "{\gamma_0}", ""{name=U,inner sep=1pt,above}] \ar[l,anchor=center, bend right=45, "{\gamma_1}", swap, ""{name=D,inner sep=1pt,below}] 
            x_0 \arrow[Rightarrow,from=U, to=D, "{\sigma}",swap]
        \end{tikzcd}
    \end{equation}     
    with sitting instants
    \begin{equation}\label{eq:surface_sitting_instants}
        \begin{aligned}
            s\approx i\in\{0,1\}~~&\Rightarrow~~\sigma(s,t) = \gamma_i(t)~,\\
            t\approx i\in\{0,1\}~~&\Rightarrow~~\sigma(s,t) = x_i~.
        \end{aligned}
    \end{equation}

    A homotopy with sitting instants \(\sigma\) is \emph{thin} if the rank of \(\dd\sigma\) is at most \(1\) everywhere. The \emph{path groupoid} $\CP M$ is the groupoid whose objects are points in \(M\), and whose 1-morphism from \(x_1\in M\) to \(x_2\in M\) is an equivalence class of paths with sitting instants, which we identify any two paths $\gamma_1,\gamma_2\colon x_0\to x_1$, $x_0,x_1\in M$ between which there is a thin homotopy with sitting instants. This ensures that composition of paths is associative.\footnote{The \emph{fundamental groupoid} $\Pi_1(M)$ is finer than $\CP M$, since in that case we do not impose the condition of rank \(\le1\) on the homotopies. A parallel transport functor whose domain is the fundamental groupoid can only describe flat connections.}
    We neglect details of the topology and smooth structure. Such details can be treated rigorously using \emph{diffeological spaces}; see~\cite{Schreiber:0705.0452, Schreiber:0802.0663, Schreiber:2008aa, Waldorf:0911.3212}, as well as~\cite{Stacey:0803.0611} and references therein for further details.
    
    We can also construct the path 2-groupoid \(\CP_{(2)}M\)~\cite{Schreiber:0802.0663} as follows. The objects are points, and the 1-morphisms are equivalence classes of paths (with sitting instants) under thin homotopies (with sitting instants). The 2-morphisms are be equivalence classes of (not necessarily thin!) homotopies (with sitting instants) under thin homotopies of homotopies (with sitting instants), which we now define.
    
    A \emph{homotopy of homotopies with sitting instants} between homotopies \(\sigma_0,\sigma_1\) between the same paths \(\gamma_0,\gamma_1\) between the same endpoints \(x_0,x_1\) is a smooth map
    \begin{equation}
        \rho\colon [0,1]^3 \to M~,\hspace{1.5cm}
        \begin{tikzcd}[column sep=2.5cm,row sep=large]
            x_1 & \ar[l, bend left=60, "\gamma_0", ""{name=U,inner sep=1pt,above}] \ar[l, bend right=60, "\gamma_1", swap, ""{name=D,inner sep=1pt,below}] x_0
            \arrow[Rightarrow, bend left=70,from=U, to=D, "{\sigma_1}",""{name=L,inner sep=1pt,right}]
            \arrow[Rightarrow, bend right=70,from=U, to=D, "{\sigma_0}",""{name=R,inner sep=1pt,left},swap]
            \arrow[triple,from=R,to=L,"{\rho}",swap]
        \end{tikzcd}
    \end{equation}
    with sitting instants 
    \begin{equation}
        \begin{aligned}
            r\approx i\in\{0,1\}~~&\Rightarrow~~\rho(r,s,t) = \sigma_i(t)~,\\
            s\approx i\in\{0,1\}~~&\Rightarrow~~\rho(r,s,t) = \gamma_i(t)~,\\
            t\approx i\in\{0,1\}~~&\Rightarrow~~\rho(r,s,t) = x_i~.
        \end{aligned}
    \end{equation}
    Such a homotopy is called \emph{thin} if $\dd\rho$ has rank $\le2$ everywhere and $\dd\rho$ has rank \(\le1\) at \((r,s,t)\) with \(s\in\{0,1\}\).\footnote{This ensures that domains and codomains are well defined on equivalence classes of homotopies of homotopies.}
    
    We also need the path 3-groupoid \(\CP_{(3)}M\), whose obvious definition we spell out as well. Its objects, 1-morphisms, and 2-morphisms are as before. Its 3-morphisms are equivalence classes of homotopies of homotopies under thin homotopies of homotopies of homotopies, which we define below. A \emph{homotopy of homotopies of homotopies with sitting instants} between homotopies of homotopies \(\rho_0,\rho_1\) between the same homotopies \(\sigma_0,\sigma_1\) between the same paths \(\gamma_0,\gamma_1\) between the same endpoints \(x_0,x_1\) is a smooth map
    \begin{equation}
        \pi\colon [0,1]^4 \to M~,\hspace{1.5cm}
        \begin{tikzcd}[column sep=3.5cm,row sep=large]
            x_1 & \ar[l, bend left=60, "\gamma_0", ""{name=U,inner sep=1pt,above}] \ar[l, bend right=60, "\gamma_1", swap, ""{name=D,inner sep=1pt,below}] x_0
            \arrow[Rightarrow, bend left=70,from=U, to=D, "{\sigma_1}",""{name=L,inner sep=1pt,right}]
            \arrow[Rightarrow, bend right=70,from=U, to=D, "{\sigma_0}",""{name=R,inner sep=1pt,left},swap]
            \arrow[triple,from=R,to=L,bend left=45,"{\rho_1}",""{name=M1,inner sep=1pt,above}]
            \arrow[triple,from=R,to=L,bend right=45,"{\rho_2}",""{name=M2,inner sep=1pt,below},swap]
            \arrow[quad,from=M1,to=M2,"{\pi}",swap]
        \end{tikzcd}
    \end{equation}
    such that, for \(i\in\{0,1\}\),
    with sitting instants 
    \begin{equation}
        \begin{aligned}
            q\approx i\in\{0,1\}~~&\Rightarrow~~\pi(q,r,s,t) = \rho_i(t)~,\\
            r\approx i\in\{0,1\}~~&\Rightarrow~~\pi(q,r,s,t) = \sigma_i(t)~,\\
            s\approx i\in\{0,1\}~~&\Rightarrow~~\pi(q,r,s,t) = \gamma_i(t)~,\\
            t\approx i\in\{0,1\}~~&\Rightarrow~~\pi(q,r,s,t) = x_i~.
        \end{aligned}
    \end{equation}
    Such a homotopy is called \emph{thin} if $\dd\pi$ has rank $\le3$ everywhere, $\dd\pi$ has rank $\le2$ at $(q,r,s,t)$ with $r\in\{0,1\}$ and $\dd\pi$ has rank $\le1$ at $(q,r,s,t)$ with $s\in\{0,1\}$. Thankfully, this is all we need.

    \subsection{Chen forms}\label{app:chen_forms}
    
    To define path-ordered higher-dimensional integrals, we use the formalism of \emph{Chen forms}. Briefly, the idea is to regard \(n\)-forms as \(1\)-forms on iterated path spaces. The treatment here is not meant to be rigorous, but to give the general flavor of ideas. For technical details the reader should consult~\cite{Baez:2004in,Getzler:1991:339,Hofman:2002ey}.
    
    \paragraph{Surface-ordering.} We want to define a surface-ordered integral of a 2-form, analogous to path-ordered integrals of 1-forms. For this, we must fix an order on the points on a surface $\sigma$, which is evidently not canonical. If $\sigma(s,t)$ is a parametrized surface 
    \begin{equation}
        \sigma\colon [0,1]\times [0,1]\rightarrow M~,
    \end{equation}
    we can define an ordering of points lexicographically: we first sort by \(s\), then by \(t\). This amounts to the following picture.
    
    First, ensure that \(\sigma\) has sitting instants~\eqref{eq:surface_sitting_instants}, reparametrizing as necessary; see figure~\ref{fig:parametrized_surface}. Then the parametrized surface \(\sigma\) forms a bigon between the two parametrized curves
    \begin{align}
        \gamma_1(t) &\coloneqq \sigma(0,t)~,&
        \gamma_2(t) &\coloneqq \sigma(1,t)~.
    \end{align}
    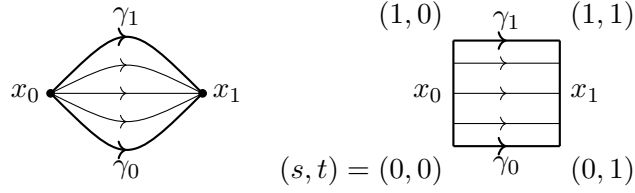
\begin{figure}
    \[
    \begin{tikzpicture}[baseline=-0.1]
    \coordinate (Q0) at (0-1,0);
    \coordinate (Q1) at (1,0);

    \filldraw (Q0) circle (0.05);
    \filldraw (Q1) circle (0.05);
    
    \node [left] at (Q0) {\(x_0\)};
    \node [right] at (Q1) {\(x_1\)};
    
    \begin{scope}[xshift=5cm]
    \coordinate (P0) at (0.7,0.7);
    \coordinate (P1) at (0.7,-0.7);
    \coordinate (P2) at (-0.7,-0.7);
    \coordinate (P3) at (-0.7,0.7);    
    \end{scope}
    
    \begin{scope}[thick,decoration={
        markings,
        mark=at position 0.5 with {\arrow{>}}}
    ]
        \draw [postaction={decorate}] (Q0) .. controls +(1,1) .. (Q1) node[midway, above] {\(\gamma_1\)};
        \draw [postaction={decorate}] (Q0) .. controls +(1,-1) .. (Q1)  node[midway, below] {\(\gamma_0\)};
        \draw (P0) -- (P1) node [right, midway] {\(x_1\)};
        \draw [postaction={decorate}] (P2) -- (P1) node[below, midway] {\(\gamma_0\)};
        \draw (P2) -- (P3) node [left, midway] {\(x_0\)};
        \draw [postaction={decorate}] (P3) -- (P0) node[above, midway] {\(\gamma_1\)};
        
        \begin{scope}[thin]
        \draw [postaction={decorate}] (Q0) -- (Q1);
        \draw [postaction={decorate}] (Q0) .. controls +(1,0.5) .. (Q1);
        \draw [postaction={decorate}] (Q0) .. controls +(1,-0.5) .. (Q1);
        
        \draw [postaction={decorate}] ($(P2)+(0,0.3)$) -- ($(P1)+(0,0.3)$);
        \draw [postaction={decorate}] ($0.5*(P2)+0.5*(P3)$) -- ($0.5*(P1)+0.5*(P0)$);
        \draw [postaction={decorate}] ($(P3)+(0,-0.3)$) -- ($(P0)+(0,-0.3)$);
        \end{scope}
        
    \end{scope}

    \node[above right] at (P0) {\((1,1)\)};
    \node[below right] at (P1) {\((0,1)\)};
    \node[below left] at (P2) {\((s,t)=(0,0)\)};
    \node[above left] at (P3) {\((1,0)\)};
    
    \begin{scope}[dashed]
        %\draw (Q0) -- (P3);
        %\draw (Q0) -- (P2);
        %\draw (Q1) -- (P0);
        %\draw (Q1) -- (P1);
    \end{scope}
        
    \end{tikzpicture}
    \]
    \caption{A parametrized surface with sitting instants, seen as a parametrized curve on the space of parametrized curves between two fixed points.}\label{fig:parametrized_surface}
    \end{figure}
    We can regard $\sigma$ as a parametrized path \(\check\sigma\) between two points \(\gamma_0,\gamma_1\in P_{x_0}^{x_1}\), where \(P_{x_0}^{x_1}\) is the manifold of parametrized paths\footnote{the manifold of parametrized paths with sitting instants, defined similarly to \(\hom_{\CP M}(x_0,x_1)\), but without quotienting by thin homotopies} between $x_0$ and $x_1$.
    \begin{equation}
        \check\sigma_s(t)\coloneqq\sigma(s,t)~,~~~\check\sigma\in P_{\gamma_0}^{\gamma_1}(P_{x_0}^{x_1}(M))~,~~~\check\sigma_s \in P_{x_0}^{x_1}(M)\text{ for all }s\in[0,1]~.
    \end{equation}
    %Note that $\sigma$ has indeed the appropriate sitting instants.

    An ordinary 2-form on $M$ defines a 1-form $\check B=\chint B$ on the locally convex manifold $P_{x_0}^{x_1}(M)$, also known as a {\em Chen form}. To wit, for each path $\gamma\in P_{x_0}^{x_1}(M)$, we can pull back $B$ along the evaluation map $\operatorname{ev}_t\colon\gamma\mapsto \gamma(t)$. We then contract $\operatorname{ev}^*_t B$ with the vector field tangent $R$ acting as $R(\gamma)=\dot \gamma$, which generates reparametrizations of $\gamma$ and whose pushforward is tangent to $\gamma$:
    \begin{equation}
        \check B=\chint B\coloneqq\int_0^1 \dd t~\iota_R (\operatorname{ev}_t^* B)~.
    \end{equation}
    For details, see again~\cite{Baez:2004in,Getzler:1991:339,Hofman:2002ey}. The 1-form $\check B$ can then be further integrated along the path \(\check\sigma\) in $P_{x_0}^{x_1}(M)$.
    
    \paragraph{Lie-algebra valued 2-forms.} To deal with 2-forms $B\in \Omega^2(M)\otimes\frh$ which transform under a gauge group with connection 1-form $A\in \Omega^1(M)\otimes \frg$, i.e.~equipped with an action of $\frg$ on $\frh$, we extend the picture slightly. The integration to a Chen form is now modified by an underlying parallel transport along the path $\gamma\in P_{x_0}^{x_1}(M)$ described by $A$. The 2-form $B$ is decorated by path-ordered integrals of $A$ along parts of $\gamma$:
    \begin{equation}\label{eq:non-abelian_correction}
        \check B=\chint_AB\coloneqq\int_0^1 \dd t~W^{-1}_t(\iota_R (\operatorname{ev}_t^* B))\ewith W_t=\Pexp\int_{\gamma_t} A~,
    \end{equation}
    where $\gamma_t$ is again the path $\gamma$ truncated at $t$ and reparametrized. For details, see~again~\cite{Baez:2004in,Hofman:2002ey}.
    
    \paragraph{Higher-dimensional generalizations.} The higher-dimensional generalization on iterated loop spaces is mostly self-evident: we iterate the procedure, producing Chen forms of lower and lower degree. Given a form \(C\in\Omega^k(M)\), we pick two points \(x_0,x_1\in M\), define the space of parametrized paths (with sitting instants) \(P_{x_0}^{x_1}(M)\), and build the Chen form \(\chint C\in \Omega^{k-1}(P_{x_0}^{x_1}(M))\). We then pick two points \(\gamma_0,\gamma_1\in P_{x_0}^{x_1}(M)\) (i.e.~parametrized paths with sitting instants on \(M\)), define the space \(P_{\gamma_0}^{\gamma_1}(P_{x_0}^{x_1}(M))\) of parametrized surfaces (i.e.~parametrized paths with sitting instants on \(P_{x_0}^{x_1}(M)\)), and build the Chen form \(\chichint C\in\Omega^{k-2}(P_{\gamma_0}^{\gamma_1}(P_{x_0}^{x_1}(M))\). We then iterate the process until we obtain a 1-form on an iterated loop space, over which we can then define a path-ordered integral. 
    
    For ``non-abelian'' forms, that is forms $C\in \Omega^{k}(M)\otimes\frv$ taking values in some vector spaces $\frv$ carrying representations of certain Lie algebras, there is also an evident generalization of~\eqref{eq:non-abelian_correction} by decorating with one-forms on iterated loop spaces. For concreteness, let us explain formula~\eqref{eq:443c} in more detail. The iterated integral there is defined as follows:
    \begin{equation}
        \Pexp\ichichint_{A,B} (-H)\coloneqq\Pexp\ichint_{\check B}\chint_A (-H)~,
    \end{equation}
    where $\chint_A (-H)$ is defined as in~\eqref{eq:non-abelian_correction}, which is consistent as $H$ takes values in $\frh$, which carries a representation of the Lie group $\sG$ integrating $\frg$. Moreover, the second integral is again defined as in~\eqref{eq:non-abelian_correction}, but now on the path space $P_{x_0}^{x_1}(M)$ with the 1-form
    \begin{equation}
        \check B\coloneqq\chint_A B\in\Omega^1(P_{x_0}^{x_1}(M))\otimes\frh~.
    \end{equation}
    Again, $\frh$ clearly carries a representation of the Lie group $\sH$ integrating $\frh$, and thus~\eqref{eq:443c} is indeed well-defined.
        
    \bibliography{bigone}
    
    \iffancy
    \bibliographystyle{latexeu}
    \else
    \bibliographystyle{alphaurl}
    \fi
    
\end{document}